\documentclass[11pt, twoside]{amsart}

\usepackage{fullpage}
\usepackage{amsmath}
\usepackage{amssymb}
\usepackage{amsfonts}
\usepackage{graphicx}
\usepackage{amsthm}
\usepackage{tikz}
\usepackage{hyperref}
\usepackage{enumitem}
\usepackage[nice]{nicefrac}
\newcommand{\tcr}[1]{\textcolor{red}{#1}}

\newcommand{\oRe}{\operatorname{Re}}
\newcommand{\oIm}{\operatorname{Im}}
\newcommand{\dz}{\text{d}z}
\newcommand{\dw}{\text{d}w}
\newcommand{\ds}{\text{d}s}
\newcommand{\dr}{\text{d}r}
\newcommand{\nhalf}{\frac{N}{2}}
\newcommand{\eps}{\varepsilon}
\newcommand{\tB}{\mathtt{B}}
\newcommand{\tW}{\mathtt{W}}
\newcommand{\tE}{\mathtt{E}}
\newcommand{\al}{\alpha}
\newcommand{\tV}{\mathtt{V}}
\newcommand{\im}{\textrm{i}}
\newcommand{\imt}{\emph{i}}

\newtheorem{theorem}{Theorem}[section]
\newtheorem*{theorem*}{Theorem}
\newtheorem{proposition}{Proposition}[section]
\newtheorem{lemma}{Lemma}[section]
\newtheorem{corollary}{Corollary}[section]

\newtheorem{definition}{Definition}[section]
\theoremstyle{remark}
\newtheorem{remark}{Remark}[section]

\newif\ifShowComments
\ShowCommentstrue
\newcounter{CommentCounter}

\title{Split Two-Periodic Aztec Diamond}
\author{Meredith Shea}
\address{Department of Mathematics, Vassar College, Poughkeepsie NY}
\email{mshea@vassar.edu}

\begin{document}

\begin{abstract}
Recent advancements have been made to understand the statistics of the Aztec diamond dimer model under general periodic weights. In this work we define a model that breaks periodicity in one direction by combining two different two-periodic weightings. We compute the correlation kernel for this Aztec diamond dimer model by extending the methods developed by Berggren and Duits (2019), which utilize the Eynard-Mehta theorem and a Wiener-Hopf factorization. From a form of the correlation kernel that is suitable for asymptotics, we compute the local asymptotics of the model in the different macroscopic regions present. We prove that the local asymptotics of the model agree with the typical two-periodic model in the highest order, however the sub-leading order terms are affected.
\end{abstract}

\maketitle

\tableofcontents

\section{Introduction}

\subsection{The Aztec Diamond and Dimer Models}
Over the past quarter century, planar dimer models have been an active area of study. A \textit{dimer covering} of a graph, is a subset of the edges such that each vertex of the graph is incident to exactly one edge in the covering. A \textit{dimer model} is a probability measure on the set of all dimer coverings. To define the measure one can assign \textit{edge weights} to the graph. The probability of a certain covering is proportional to the product of the edge weights of the dimers contained in the covering. The ideas behind dimer models were initially introduced by Kasteleyn \cite{Kas61} and Temperley and Fisher \cite{TF61}.

Dimer models are often studied under some appropriate scaling limit. Cohn, Kenyon and Propp \cite{CKP01} showed that a class of uniform dimer models satisfy a variational principle and that the \textit{height function} of the model converges, in the scaling limit, to a deterministic limit. Subsequently, Kenyon, Okounkov, and Sheffield \cite{KOS03} showed that there are three types of Gibbs measures that can appear in dimer model with doubly periodic edge weights. These \textit{macroscopic regions} are classified as: frozen, rough, and smooth.\footnote{In the literature these regions are also referred to as solid, liquid, and gas, respectively.} In the frozen region dimers are deterministic. In the rough region dimer correlations decay polynomially with distance, while in the smooth region dimer correlations decay exponentially with distance. Not all dimer models exhibit smooth regions, however most exhibit rough regions. 

A seminal example of a dimer model is domino tilings of the \textit{Aztec diamond} \cite{EKLP92}. The asymptotics of the simplest case, the Aztec diamond with uniform edge weights, was studied by Jockusch, Propp and Shor in \cite{JPS98}. They proved that the boundary between the deterministic and non-deterministic regions forms a circle. This is known as the \textit{Arctic Circle Theorem}. A key aspect to understanding correlations of the Aztec diamond is the fact that the model is a \textit{determinantal point process}. One approach to computing correlations is through computing the \textit{inverse Kasteleyn matrix}.\footnote{For a general introduction to dimer models via Kasteleyn theory see \cite{Ken2009lectures}.} The inverse Kasteleyn matrix for the Aztec diamond with two-periodic weights was originally computed by Chhita and Young in \cite{CY14}. Further asymptotics of this model have been studied in \cite{CJY15,CJ16,BCJ18} among others. The two-periodic Aztec diamond is, notably, the simplest Aztec diamond model which exhibits all three macroscopic regions. 

\begin{figure}[h]
    \centering
    \includegraphics[scale=0.5]{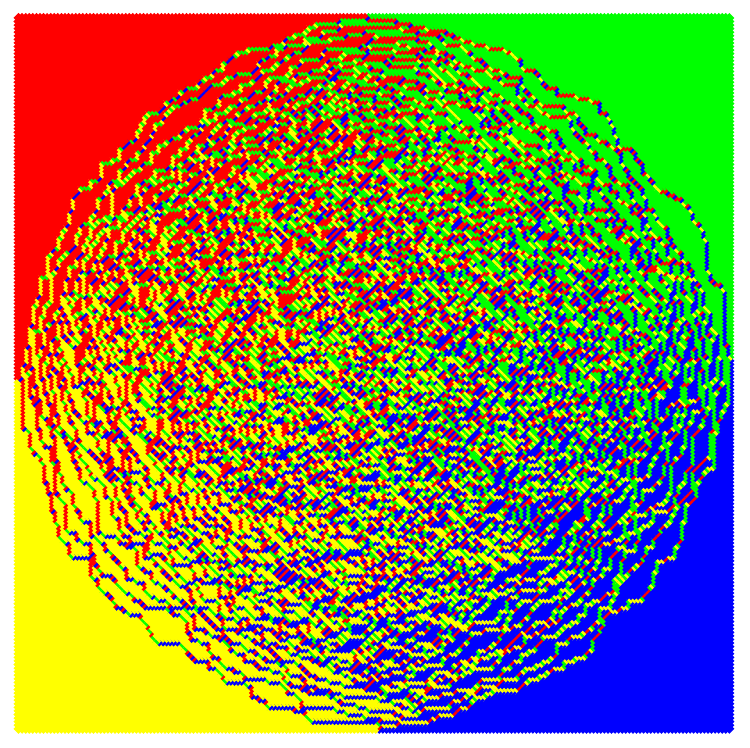} \hspace{1cm} \includegraphics[scale=0.5]{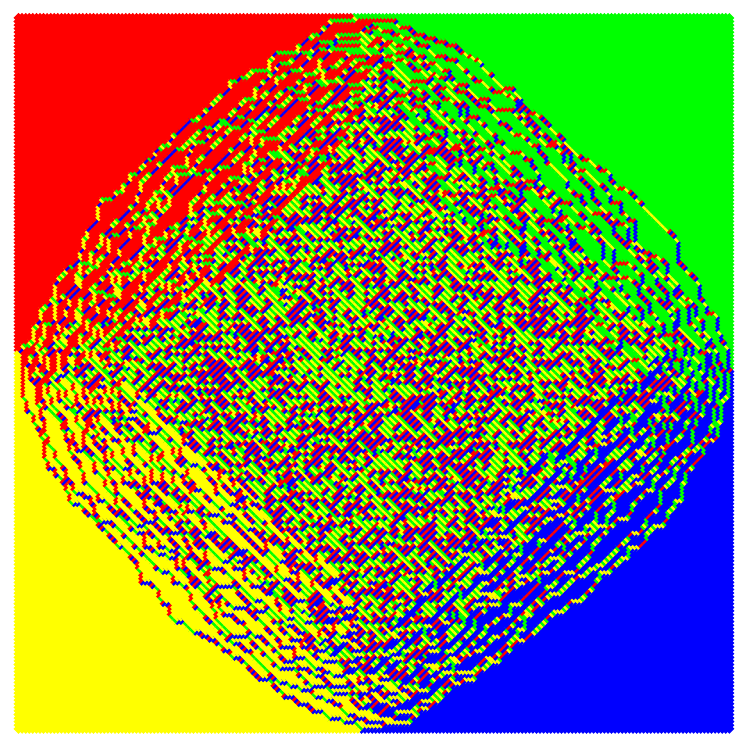}
    \caption{On the left is a simulation of the Aztec diamond with uniform weights. On the right is a simulation of the Aztec diamond with two periodic weights. Original code for simulation was provided by Sunil Chhita.}
    \label{fig:ADexamples}
\end{figure}

An alternative approach to understanding the correlations of the dimer model uses a bijection between the Aztec diamond and a non-intersecting paths model along side the Eynard-Mehta theorem \cite{EM97}. The correlation kernel of the two-periodic Aztec diamond was computed using this approach by Duits and Kuijlaars \cite{DK21}, via \textit{matrix valued orthogonal polynomials} (MVOP) and Berggren and Duits \cite{BD19}, via a \textit{Wiener-Hopf factorization}. In both of these instances, the methods produced forms of the correlation kernel that are well suited for asymptotic analysis. The Wiener-Hopf method has been related to the Kasteleyn treatment in \cite{CD23} and to the MVOP method in \cite{KP24}.

Recently, there have been many results about the Aztec diamond with more general periodic weightings, see \cite{BB24,Ber19,BD23,BB23}. There is less work, however, on dimer models that involve non-periodic weightings. One notable non-periodic set up is the $q^{vol}$-weightings \cite{CY14,BBCCR17}. Additionally in \cite{BdT24} Boutillier and de Tili\`{e}re compute the inverse Kasteleyn matrix for the Aztec diamond with Fock's weights, which are only quasi-periodic. Other set ups and approaches, such as $t$-embeddings \cite{CLR23,BNR24} and the Enyard-Mehta theorem \cite{DK21,BD19}, seem to have possible applications to certain non-periodic set ups. 

In Section \ref{section:modeldef}, we define an extension of the typical two-periodic Aztec diamond, which we refer to as the \textit{split two-periodic Aztec diamond}. In the split two-periodic model, periodicity is broken along a line which we refer to as the \textit{interface} of the model. A simulation of the split two-periodic Aztec diamond is shown in Figure \ref{fig:SplitADsimulation}. We compute the correlation kernel of this new model by extending the methods developed in \cite{BD19}. The statement of the correlation kernel is found in Section \ref{section:kernelstatement}. The rest of our discussion in Section \ref{section:results} relates to defining the macroscopic regions and computing the local asymptotic of the model away from the boundaries and interface. 

\begin{figure}[h]
    \centering
    \includegraphics[scale=0.7]{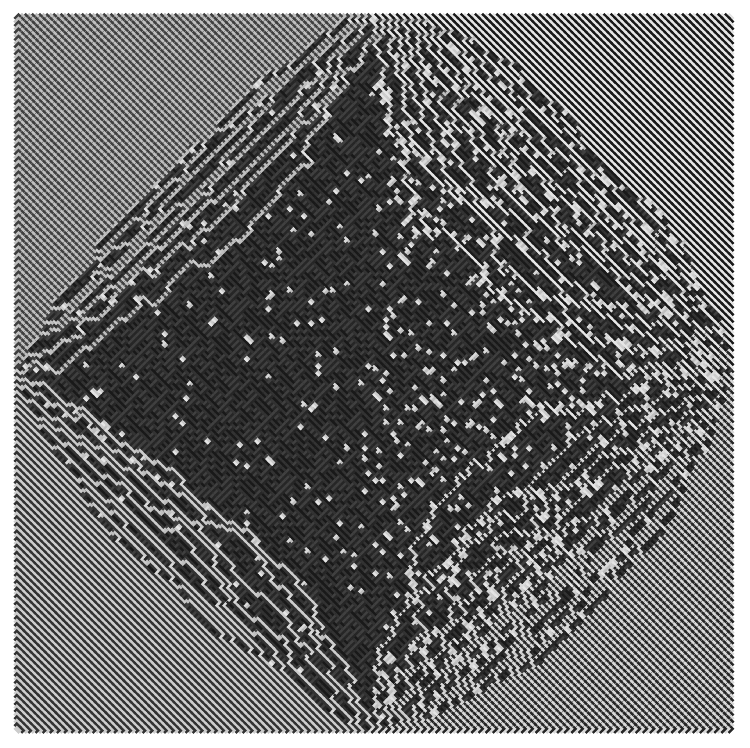}
    \caption{Simulated split Aztec diamond where $N = 200$, $\alpha = 1/4$, and $\beta = 1/2$. The tiles are colored by an 8 color gray scale to accentuate the smooth region. Original code for simulation was provided by Sunil Chhita. Simulations are sampled exactly from the distribution using domino-shuffling, see \cite{Pro03}.}
    \label{fig:SplitADsimulation}
\end{figure}

The local asymptotics illustrate how the new model maintains much of the asymptotic behavior seen in the two-periodic case. The differences between the split two-periodic model and its two-periodic counter parts are highlighted via the discussion in Section \ref{section:asymptoticsstatement}. The biggest difference in the local asymptotic behavior is seen in the decay of the sub-leading order term of the kernel in certain parts of the smooth region of the model. Further work should be done to study the behavior of the model near and across the interface. These results suggest that where the smooth-rough boundary meets the interface behaves differently than the typical rough-smooth boundary cusp found in the two-periodic Aztec diamond. 

This work also presents a prescription for deriving the correlation kernel for other non-periodic models. One can follow the work presented in Section \ref{section:BDstuff} to obtain an intermediate kernel, like the one stated in Lemma \ref{lemma:BDstuff}, for instances when the interface is moved to the right or left in the scaled picture of the Aztec diamond. Model simulations suggest that the correlation between the two sides of the interface has more complicated behavior under these circumstances. 

\subsubsection{Acknowledgements}
The author would like to thank Sunil Chhita for fruitful discussion throughout this project. Part of this research was performed while the author was visiting the Institute for Pure and Applied Mathematics (IPAM), which is supported by the National Science Foundation Grant No. DMS-1925919.

%%%%%%%%%%%%%%%%%%%%%%%%%%%%%%%%%%%%%%%%%%%%%%%%%%%%%%%%%%%%%%%%%%%%%%%%%%%%%%%%%%%%%%%%%%%%%%%%%%%%%%%%%%%%%%%%%%%%%%%%%%%%%%%%%%%%%%%%%%%%%%%%
\subsection{The Split Two-Periodic Aztec Diamond} \label{section:modeldef}
We start by defining the vertex and edge sets of the Aztec diamond of size $n$. The vertices are give by,
\begin{equation} \label{eq:coordstart}
    \tW^{\text{Az}}_n = \Big\{ (2j+1,2k) : 0 \leq j \leq n-1 \text{ and } 0 \leq k \leq n \Big\}
\end{equation}
and
\begin{equation} \label{eq:coordend}
    \tB^{\text{Az}}_n = \Big\{ (2j,2k+1) : 0 \leq j \leq n \text{ and } 0 \leq k \leq n-1 \Big\}.
\end{equation}
The Aztec diamond is a bipartite graph, so we write the set of all vertices as $\tV_n^{\text{Az}} = \tW^{\text{Az}}_n \sqcup \tB^{\text{Az}}_n$. The edges are given by,
\begin{multline} \label{eq:edgecoord}
    \tE^{\text{Az}}_n = \Big\{ ((2j+1, 2k), (2j+1\pm 1, 2k+1)) : 0 \leq j \leq n-1, 0 \leq k \leq n-1 \Big\} \\
    \cup \Big\{ ((2j+1, 2k), (2j+1\pm 1, 2k-1)) : 0 \leq j \leq n-1, 1\leq k \leq n \Big\}.
\end{multline}
\begin{figure}[h]
    \centering
    \begin{tikzpicture}[scale=0.75]
    \foreach \t in {1,3,5}
        {
        \draw (0,\t) -- (\t,0);
        \draw (\t,6) -- (6,\t);
        \draw (\t,0) -- (6,6-\t);
        \draw (0,\t) -- (6-\t,6);
        }
    \foreach \j in {0,1,2,3}
        \foreach \k in {0,1,2}
            {
            \filldraw[black] (2*\j,2*\k+1) circle (2pt);
            }
    \foreach \j in {0,1,2}
        \foreach \k in {0,1,2,3}
            {
            \filldraw[color=black,fill=white] (2*\j+1,2*\k) circle (2pt);
            }
    \foreach \j in {0,1,2,3,4,5,6}
        {
        \node at (\j,-1) {$\j$};
        \node at (-1,\j) {$\j$};
        }
    \end{tikzpicture}
    \hspace{1cm}
    \begin{tikzpicture}[scale=0.75]
    \foreach \t in {1,3,5,7}
        {
        \draw (0,\t) -- (\t,0);
        \draw (\t,8) -- (8,\t);
        \draw (\t,0) -- (8,8-\t);
        \draw (0,\t) -- (8-\t,8);
        }
    \foreach \j in {0,1,2,3,4}
        \foreach \k in {0,1,2,3}
            {
            \filldraw[black] (2*\j,2*\k+1) circle (2pt);
            }
    \foreach \j in {0,1,2,3}
        \foreach \k in {0,1,2,3,4}
            {
            \filldraw[color=black,fill=white] (2*\j+1,2*\k) circle (2pt);
            }
    \foreach \j in {0,1,2,3,4,5,6,7,8}
        {
        \node at (\j,-1) {$\j$};
        \node at (-1,\j) {$\j$};
        }
    \end{tikzpicture}
    \caption{Examples of the Aztec diamond graph. On the left is the Aztec diamond of size $n = 3$, on the right is the Aztec diamond of size $n = 4$.}
    \label{fig:AztecDiamondGraph}
\end{figure}
Examples of the Aztec diamond graph are given in Figure \ref{fig:AztecDiamondGraph}. For the split two-periodic Aztec diamond we will assume that $n$ is even, so we write $n = 2N$. Now we are ready to define the weights of the split two-periodic Aztec diamond. Let $b = (b_x,b_y) \in \tB$ and $w = (w_x,w_y)\in \tW$. We say that $b$ and $w$ are neighbors if there is an edge between the two vertices. If this is the case we write $b \sim w$. We will use the notation $\text{wt}(b,w)$ to denote the weight of the edge connecting the vertices $b$ and $w$. Note that $\text{wt}(b,w) = \text{wt}(w,b)$. The weight function for the split two-period Aztec diamond is, 
\begin{equation} \label{eq:edgewts}
    \text{wt}(b,w) = \begin{cases}
        \al^2 & \text{if } w_x \equiv 1 \bmod{4}, \, w_y \equiv 0 \bmod{4}, \, w-b = (\pm1,-1), \text{ and } w_x < 2N\\
        \al^{-2} & \text{if }  w_x \equiv 1 \bmod{4}, \, w_y \equiv 2 \bmod{4}, \, w-b = (\pm1,-1), \text{ and } w_x < 2N\\
        \beta^2 & \text{if } w_x \equiv 1 \bmod{4}, \, w_y \equiv 0 \bmod{4}, \, w-b = (\pm1,-1), \text{ and } w_x > 2N\\
        \beta^{-2} & \text{if } w_x \equiv 1 \bmod{4}, \, w_y \equiv 2 \bmod{4}, \, w-b = (\pm1,-1), \text{ and } w_x > 2N\\
        1 & \text{if } w_x \equiv 1 \bmod{4}, \text{ and } w-b = (\pm1,1)\\
        1 & \text{if } w_x \equiv 3 \bmod{4}, \text{ and } w-b = (\pm1,\pm1)\\
        0 & \text{otherwise}
    \end{cases}
\end{equation}
where $0<\al,\beta<1$ are constants. The last case in the above formula indicates the the weight is 0 if the vertices are not neighbors. If $\alpha = \beta$, this model recovers the usual two-periodic Aztec diamond. If $\alpha = \beta = 1$ we recover the uniform Aztec diamond, however we will generally assume that the constants are not equal to 1. The split two-periodic Aztec diamond of size $2N=4$ is shown in Figure \ref{fig:AztecDiamondExample}. We will refer to the line were the weighting changes as the \textit{interface} of the model. In this coordinate convention, the interface occurs along the vertical line $x = 2N$.

\begin{figure}[h]
    \centering
    \begin{tikzpicture}
    \foreach \t in {1,3,5,7}
        {
        \draw (0,\t) -- (\t,0);
        \draw (\t,8) -- (8,\t);
        \draw (\t,0) -- (8,8-\t);
        \draw (0,\t) -- (8-\t,8);
        }
    \foreach \j in {0,1,2,3,4}
        \foreach \k in {0,1,2,3}
            {
            \filldraw[black] (2*\j,2*\k+1) circle (2pt);
            }
    \foreach \j in {0,1,2,3}
        \foreach \k in {0,1,2,3,4}
            {
            \filldraw[color=black,fill=white] (2*\j+1,2*\k) circle (2pt);
            }
    \foreach \j in {0,1,2,3,4,5,6,7,8}
        {
        \node at (\j,-1) {$\j$};
        \node at (-1,\j) {$\j$};
        }
    \node at (0.7,0.7) {$\scriptstyle \alpha^2$};
    \node at (1.4,0.7) {$\scriptstyle \alpha^2$};
    \node at (0.7,2.7) {$\scriptstyle \alpha^{-2}$};
    \node at (1.4,2.7) {$\scriptstyle \alpha^{-2}$};
    \node at (0.7,4.7) {$\scriptstyle \alpha^2$};
    \node at (1.4,4.7) {$\scriptstyle \alpha^2$};
    \node at (0.7,6.7) {$\scriptstyle \alpha^{-2}$};
    \node at (1.4,6.7) {$\scriptstyle \alpha^{-2}$};
    \node at (4.7,0.7) {$\scriptstyle \beta^2$};
    \node at (5.4,0.7) {$\scriptstyle \beta^2$};
    \node at (4.7,2.7) {$\scriptstyle \beta^{-2}$};
    \node at (5.4,2.7) {$\scriptstyle \beta^{-2}$};
    \node at (4.7,4.7) {$\scriptstyle \beta^2$};
    \node at (5.4,4.7) {$\scriptstyle \beta^2$};
    \node at (4.7,6.7) {$\scriptstyle \beta^{-2}$};
    \node at (5.4,6.7) {$\scriptstyle \beta^{-2}$};
    \end{tikzpicture}
    \caption{The Aztec diamond of size $2N = 4$ with coordinates given by equations \eqref{eq:coordstart}-\eqref{eq:coordend} and weights given by equation \eqref{eq:edgewts}. All unlabeled edged have weight 1. }
    \label{fig:AztecDiamondExample}
\end{figure}

\subsection{A Related Paths Process and the Correlation Kernel} \label{section:pathsprocess}
To study the split two-periodic Aztec diamond we will extend the work of Berggren and Duits in \cite{BD19}. Their work focuses on a class of non-intersecting paths models which, under appropriate conditions on the edge weights, are equivalent to certain Aztec diamond dimer models. The equivalence of the Aztec diamond to certain non-intersecting paths models was initially realized by Johansson \cite{Jo03}. Here we briefly introduce the non-intersection paths process which is equivalent to the split two-periodic Aztec diamond. Section \ref{section:BDstuff} details the equivalence between these processes. 

We will consider a graph which we will refer to as the BD-paths graph of size $n=2N$. The graph has the following vertex set, 
\begin{equation*}
    \mathtt{V}^{\text{BD}}_n = \Big\{ (j,k) : 0 \leq j \leq 2n, \, -n -1 \leq k \leq -1 \Big\}  
\end{equation*}
and the following edge set,
\begin{multline*}
    \mathtt{E}^{\text{BD}}_n = \Big\{((j,k),(j+1,k)) : 0 \leq j < 2n, \, -n -1 \leq k \leq -1 \Big\} \\
    \cup \Big\{((2j,k),(2j,k-1)) : 1 \leq j \leq n, \, -n \leq k \leq -1 \Big\} \\
    \cup \Big\{((2j,k),(2j+1,k-1)) : 0 \leq j \leq n, \, -n \leq k \leq -1 \Big\}. 
\end{multline*}
The edges are split up into horizontal, vertical, and diagonal edges. The BD-paths graph of size $4$ is show in Figure \ref{fig:BDgraph}. On the graph we assume that there are $n$ non-intersecting paths that start at the points $(0,-k)$ for $-n \leq k \leq -1$ and that end at the points $(2j,-n-1)$ for $1 \leq j \leq 2n$. We say that there is a point at the vertex $(j,k)$ if a path crosses through the vertex and the point $(j,k+1)$ does not lie on the same path. An example configuration of the non-intersecting paths process is show in Figure \ref{fig:BDpaths}. Figure \ref{fig:BDpaths} also shows the equivalent Aztec diamond dimer covering (more details on this can be found in Section \ref{section:BDstuff}).

\begin{figure}
    \centering
    \begin{tikzpicture}[scale=0.5]
        \foreach \t in {0,2,4,6,8,10,12,14,16}
            {
            \foreach \k in {0,2,4,6,8}
                {
                \filldraw[black] (\t,\k) circle (2pt);
                }
            }
        \foreach \t in {0,2,4,6,8}
            {
            \draw (0,\t) -- (16,\t);
            }
        \foreach \t in {4,8,12,16}
            {
            \draw (\t,0) -- (\t,8);
            }
        \foreach \t in {0,4,8,12}
            {
            \foreach \k in {8,6,4,2}
                { 
                \draw (\t,\k) -- (\t+2,\k-2);
                }
            }
        \foreach \t in {0,1,2,3,4,5,6,7,8}
            {
            \node at (2*\t,-1) {$\t$};
            }
        \foreach \k in {1,2,3,4,5}
            {
            \node at (-1,10-2*\k) {$-\k$};
            }
        \node at (1,1) {$\scriptstyle\alpha^2$};
        \node at (1,3) {$\scriptstyle\alpha^{-2}$};
        \node at (1,5) {$\scriptstyle\alpha^2$};
        \node at (1,7) {$\scriptstyle\alpha^{-2}$};
        \node at (4,1) {$\scriptstyle\alpha^2$};
        \node at (4,3) {$\scriptstyle\alpha^{-2}$};
        \node at (4,5) {$\scriptstyle\alpha^2$};
        \node at (4,7) {$\scriptstyle\alpha^{-2}$};
        \node at (9,1) {$\scriptstyle\beta^2$};
        \node at (9,3) {$\scriptstyle\beta^{-2}$};
        \node at (9,5) {$\scriptstyle\beta^2$};
        \node at (9,7) {$\scriptstyle\beta^{-2}$};
        \node at (12,1) {$\scriptstyle\beta^2$};
        \node at (12,3) {$\scriptstyle\beta^{-2}$};
        \node at (12,5) {$\scriptstyle\beta^2$};
        \node at (12,7) {$\scriptstyle\beta^{-2}$};
    \end{tikzpicture}
    \caption{An example of the BD-paths graph for $n = 4$. The graph is also equipped with the edge weights that make it equivalent to the split two-periodic Aztec diamond, see Section \ref{section:BDstuff} for more details. Unlabeled edges have weight 1.}
    \label{fig:BDgraph}
\end{figure}

\begin{figure}
    \centering
    \begin{tikzpicture}[scale=0.45]
        \foreach \t in {0,2,4,6,8,10,12,14,16}
            {
            \foreach \k in {0,2,4,6,8}
                {
                \filldraw[black] (\t,\k) circle (2pt);
                }
            }
        \foreach \t in {0,2,4,6,8}
            {
            \draw (0,\t) -- (16,\t);
            }
        \foreach \t in {4,8,12,16}
            {
            \draw (\t,0) -- (\t,8);
            }
        \foreach \t in {0,4,8,12}
            {
            \foreach \k in {8,6,4,2}
                { 
                \draw (\t,\k) -- (\t+2,\k-2);
                }
            }
        \foreach \k in {2,4,6,8}
            {
            \filldraw[black] (0,\k) circle (5pt);
            }
        \draw[line width=0.1cm] (0,2) -- (2,0) -- (4,0);
        \draw[line width=0.1cm] (0,4) -- (4,4) -- (6,2) -- (8,2) -- (8,0);
        \draw[line width=0.1cm] (0,6) -- (6,6) -- (8,4) -- (12,4) -- (12,0);
        \draw[line width=0.1cm] (0,8) -- (16,8) -- (16,0);
        {
        \filldraw[black] (2,8) circle (5pt);
        \filldraw[black] (4,8) circle (5pt);
        \filldraw[black] (6,8) circle (5pt);
        \filldraw[black] (8,8) circle (5pt);
        \filldraw[black] (10,8) circle (5pt);
        \filldraw[black] (12,8) circle (5pt);
        \filldraw[black] (14,8) circle (5pt);
        \filldraw[black] (16,8) circle (5pt);
        \filldraw[black] (2,6) circle (5pt);
        \filldraw[black] (4,6) circle (5pt);
        \filldraw[black] (6,6) circle (5pt);
        \filldraw[black] (8,4) circle (5pt);
        \filldraw[black] (10,4) circle (5pt);
        \filldraw[black] (12,4) circle (5pt);
        \filldraw[black] (2,4) circle (5pt);
        \filldraw[black] (4,4) circle (5pt);
        \filldraw[black] (6,2) circle (5pt);
        \filldraw[black] (8,2) circle (5pt);
        \filldraw[black] (2,0) circle (5pt);
        \filldraw[black] (4,0) circle (5pt);
        }
    \end{tikzpicture}
    \hspace{0.5cm}
    \begin{tikzpicture}[scale=0.5]
        \foreach \t in {1,3,5,7}
            {
            \draw (0,\t) -- (\t,0);
            \draw (\t,8) -- (8,\t);
            \draw (\t,0) -- (8,8-\t);
            \draw (0,\t) -- (8-\t,8);
            }
        \foreach \j in {0,1,2,3,4}
            \foreach \k in {0,1,2,3}
                {
                \filldraw[black] (2*\j,2*\k+1) circle (2pt);
                }
        \foreach \j in {0,1,2,3}
            \foreach \k in {0,1,2,3,4}
                {
                \filldraw[color=black,fill=white] (2*\j+1,2*\k) circle (2pt);
                }
        \draw[line width=0.1cm] (0,7) -- (1,8);
        \draw[line width=0.1cm] (2,7) -- (3,8);
        \draw[line width=0.1cm] (4,7) -- (5,8);
        \draw[line width=0.1cm] (6,7) -- (7,8);
        \draw[line width=0.1cm] (7,6) -- (8,7);
        \draw[line width=0.1cm] (7,4) -- (8,5);
        \draw[line width=0.1cm] (7,2) -- (8,3);
        \draw[line width=0.1cm] (7,0) -- (8,1);
        \draw[line width=0.1cm] (0,5) -- (1,6);
        \draw[line width=0.1cm] (0,3) -- (1,4);
        \draw[line width=0.1cm] (0,1) -- (1,0);
        \draw[line width=0.1cm] (3,0) -- (4,1);
        \draw[line width=0.1cm] (5,0) -- (6,1);
        \draw[line width=0.1cm] (1,2) -- (2,1);
        \draw[line width=0.1cm] (2,3) -- (3,2);
        \draw[line width=0.1cm] (2,5) -- (3,4);
        \draw[line width=0.1cm] (5,2) -- (6,3);
        \draw[line width=0.1cm] (4,3) -- (5,4);
        \draw[line width=0.1cm] (3,6) -- (4,5);
        \draw[line width=0.1cm] (5,6) -- (6,5);
    \end{tikzpicture}
    \caption{An example of the non-intersecting path process for $n = 4$. The equivalent dimer covering of the size 4 Aztec diamond is also pictured.}
    \label{fig:BDpaths}
\end{figure}

To make the paths process have the same probability measure as the split two-periodic Aztec diamond we assign the following weights to the edges: to the horizontal edges, 
$$\text{wt}((j,k),(j+1,k)) = 1$$
to the vertical edges, 
$$\text{wt}((2j,k),(2j,k-1)) = 
\begin{cases}
    1 & \text{if } j \equiv 0 \bmod{2} \\
    \alpha^2 & \text{if } j \equiv 1 \bmod{2}, \, j < n, \text{ and } k \equiv 0 \bmod{2} \\
    \alpha^{-2} & \text{if } j \equiv 1 \bmod{2}, \, j < n, \text{ and } k \equiv 1 \bmod{2} \\
    \beta^2 & \text{if } j \equiv 1 \bmod{2}, \, j > n, \text{ and } k \equiv 0 \bmod{2} \\
    \beta^{-2} & \text{if } j \equiv 1 \bmod{2}, \, j > n, \text{ and } k \equiv 1 \bmod{2}
\end{cases}$$
and to the diagonal edges,
$$\text{wt}((2j,k),(2j+1,k-1)) = 
\begin{cases}
    1 & \text{if } j \equiv 1 \bmod{2} \\
    \alpha^2 & \text{if } j \equiv 0 \bmod{2}, \, j < n, \text{ and } k \equiv 0 \bmod{2} \\
    \alpha^{-2} & \text{if } j \equiv 0 \bmod{2}, \, j < n, \text{ and } k \equiv 1 \bmod{2} \\
    \beta^2 & \text{if } j \equiv 0 \bmod{2}, \, j \geq n, \text{ and } k \equiv 0 \bmod{2} \\
    \beta^{-2} & \text{if } j \equiv 0 \bmod{2}, \, j \geq n, \text{ and } k \equiv 1 \bmod{2}
\end{cases}.$$
As stated in \cite{BD19}, this paths process of size $n=2N$ is a determinantal point process (as is the Aztec diamond dimer model), meaning there exists a matrix $\mathbb{K}_N$ such that 
$$\mathbb{P}\big(\text{points at } (m_1, u_1), \, \dots, \, (m_k,u_k)\big) = \det\big( \mathbb{K}_N(m_i,u_i;m_j,u_j)\big)_{i,j=1}^k.$$
We call this matrix the \textit{correlation kernel} of the process. 
\begin{remark}
    The correlation kernel of a determinantal point process is, in general, not unique. For those familiar with the Kasteleyn approach to dimer models (see \cite{Ken2009lectures} and \cite{Gor20} for overviews), the inverse Kasteleyn is a correlation kernel of the dimer model, however the Kasteleyn matrix depends on a choice of Kasteleyn weighting (and is thus not unique, though the matrices are relatable). 
\end{remark} 
\begin{remark}
    We will use the coordinate notation from the paths process when defining the correlation kernel of the split two-periodic Aztec diamond in Section \ref{section:kernelstatement}. One can, heuristically, think of these coordinate changes as scalar transformations, even though the bijection between the models is a bit more complicated (see Section \ref{section:BDstuff}). One can also refer to the work in \cite{CD23} for a detailed explanation on the relationship between these processes and their kernels. 
\end{remark}

%%%%%%%%%%%%%%%%%%%%%%%%%%%%%%%%%%%%%%%%%%%%%%%%%%%%%%%%%%%%%%%%%%%%%%%%%%%%%%%%%%%%%%%%%%%%%%%%%%%%%%%%%%%%%%%%%%%%%%%%%%%%%%%%%%%%%%%%%%%%%%%%%%%%%%%%%%%%%
\section{Summary of Results} \label{section:results}
In this section, we will detail the main results of this work. This starts with the statement of the correlation kernel of the split two-periodic Aztec diamond, which we later compute by extending the methods of Berggren and Duits \cite{BD19}. We then relate this kernel to the kernel of the typical two-periodic Aztec diamond model. Lastly, we classify the macroscopic regions of the model by way of computing the local asymptotics. In particular, we emphasize the differences in the local asymptotics of the split two-periodic model compared to the typical two-periodic model.

\subsection{An Integral Representation of the Correlation Kernel} \label{section:kernelstatement}
In order to state the correlation kernel of the model, we first define some preliminary functions. We start with the matrices, 
\begin{align*}
    \phi_{\eps,1}(z) &= \begin{pmatrix}1 & \eps^2 z^{-1} \\ \eps^{-2} & 1\end{pmatrix}, & \phi_{\eps,2}(z) &= \frac{1}{1-z^{-1}}\begin{pmatrix}1 & \eps^2 z^{-1} \\ \eps^{-2} & 1\end{pmatrix},\\
    \phi_3(z) &= \begin{pmatrix}1 & z^{-1} \\ 1 & 1\end{pmatrix}, & \phi_4(z) &= \frac{1}{1-z^{-1}}\begin{pmatrix}1 & z^{-1} \\ 1 & 1 \end{pmatrix},
\end{align*}
and define
\begin{align}
\Phi_{\eps}(z) = \phi_{\eps,1}(z)\phi_{\eps,2}(z)\phi_3(z)\phi_4(z).
\end{align}
It is useful to express $\Phi_{\eps}(z)$ in terms of its eigen-decomposition, 
\begin{equation}
    \Phi_\eps(z) = E_\eps(z) \begin{pmatrix}r_{\eps,1}(z) & 0 \\ 0 & r_{\eps,2}(z) \end{pmatrix}E_\eps(z)^{-1}.
\end{equation}
The eigenvalues of $\Phi_\eps(z)$ are explicitly, 
\begin{equation} \label{eq:r1}
    r_{\eps,1}(z) = \frac{1}{(z-1)^2}\left((z+1)^2+2z(\eps^2 + \eps^{-2}) + 2 (\eps + \eps^{-1}) \sqrt{z^3+(\eps^2+\eps^{-2})z^2+z} \right)
\end{equation} 
and
\begin{equation} \label{eq:r2}
    r_{\eps,2}(z) = \frac{1}{(z-1)^2}\left((z+1)^2+2z(\eps^2 + \eps^{-2}) - 2 (\eps + \eps^{-1}) \sqrt{z^3+(\eps^2+\eps^{-2})z^2+z} \right).
\end{equation} 
We chose the branch cuts for the above to be $(-\infty,-\eps^{-2}] \cup [-\eps^2,0]$ and for $z > 0$ we choose the positive square root. We also adopt the notation,
\begin{align} \label{eq:f1}
    F_{\eps,1}(z) &= E_\eps(z) \begin{pmatrix}1 & 0 \\ 0 & 0\end{pmatrix} E_\eps(z)^{-1}\\ \nonumber
    &= \begin{pmatrix}
        \frac{1}{2} - \frac{z(\eps^2-1)}{2\sqrt{z(z+\eps^2)(1+\eps^2 z)}} & -\frac{\eps^2(z+1)}{2\sqrt{z(z+\eps^2)(1+\eps^2z)}} \\
        -\frac{z(z+1)}{2\sqrt{z(z+\eps^2)(1+\eps^2z)}} & \frac{1}{2} + \frac{z(\eps^2-1)}{2\sqrt{z(z+\eps^2)(1+\eps^2 z)}}
    \end{pmatrix} 
\end{align} 
and 
\begin{align} \label{eq:f2}
    F_{\eps,2}(z) &= E_\eps(z) \begin{pmatrix}0 & 0 \\ 0 & 1\end{pmatrix} E_\eps(z)^{-1} \\ \nonumber
    &= \begin{pmatrix}
        \frac{1}{2} + \frac{z(\eps^2-1)}{2\sqrt{z(z+\eps^2)(1+\eps^2 z)}} & \frac{\eps^2(z+1)}{2\sqrt{z(z+\eps^2)(1+\eps^2z)}} \\
        \frac{z(z+1)}{2\sqrt{z(z+\eps^2)(1+\eps^2z)}} & \frac{1}{2} - \frac{z(\eps^2-1)}{2\sqrt{z(z+\eps^2)(1+\eps^2 z)}}
        \end{pmatrix}.
\end{align}
which uses the same branch cuts as above. Lastly we define the function,
\begin{equation} \label{eq:g}
    g_{\alpha,\beta}(z) = \frac{2z(1+\alpha^2\beta^2)+\beta^2(z^2+1)+\alpha^2(z^2+1)}{4\sqrt{(z+\alpha^2)(1+\alpha^2z)(z+\beta^2)(1+\beta^2z)}}.
\end{equation}
Note that $g_{\alpha,\beta}(z) = g_{\beta,\alpha}(z)$, and $g_{\alpha,\alpha}(z) = \nicefrac{1}{2}$. Assuming $\beta > \alpha$,  $g_{\alpha,\beta}(z)$ has branch cuts $[-\alpha^{-2},-\beta^{-2}] \cup [-\beta^2,-\alpha^2]$. We are now prepared to state the main theorem of this paper, 
\begin{theorem} \label{theorem:kernel}
    Let $-N \leq \xi,\xi' \leq -1$ and $0 < m < N$. The split two-periodic Aztec diamond of size $2N$ has correlation kernel given by 
    \begin{multline} \label{eq:asidekernel}
        \Big[\mathbb{K}_N(4m',2\xi'+j;4m,2\xi+i)\Big]_{i,j=0}^1 = -\frac{\mathbb{I}_{m>m'}}{2 \pi \imt} \oint_{\gamma_{0,1}} \frac{\dz}{z} z^{\xi'-\xi} \Phi_{\alpha}(z)^{\frac{N}{2}-m'}\Phi_{\eps}(z)^{m-\frac{N}{2}} \\
        + \frac{1}{(2\pi\imt)^2} \oint_{\gamma_1} \dw \oint_{\gamma_{0,1}} \frac{\dz}{z(z-w)} \frac{w^{\xi'+N}(z-1)^{N}}{z^{\xi+N}(w-1)^N} r_{\alpha,1}(w)^{\nhalf-m'} F_{\alpha,1}(w)\Phi_\eps(z)^{m-\nhalf}\\
        + \frac{1}{(2\pi\imt)^2} \oint_{\gamma_1} \dw \oint_{\gamma_{0,1}} \frac{\dz}{z(z-w)} \frac{w^{\xi'+N}(z-1)^{N}}{z^{\xi+N}(w-1)^N} r_{\alpha,2}(w)^{\nhalf-m'} \frac{2 F_{\alpha,2}(w)F_{\beta,1}(w)F_{\alpha,1}(w)}{1+2g_{\alpha,\beta}(w)}\Phi_\eps(z)^{m-\nhalf} 
    \end{multline}
    when $0 < m' \leq \nicefrac{N}{2}$ and 
    \begin{multline} \label{eq:bsidekernelfinal}
        \Big[\mathbb{K}_N(4m',2\xi'+j;4m,2\xi+i)\Big]_{i,j=0}^1 = -\frac{\mathbb{I}_{m>m'}}{2 \pi \imt} \oint_{\gamma_{0,1}} \frac{\dz}{z} z^{\xi'-\xi} \Phi_{\beta}(z)^{\frac{N}{2}-m'}\Phi_{\eps}(z)^{m-\frac{N}{2}} \\
        + \frac{1}{(2\pi\imt)^2}\oint_{\gamma_1} \dw \oint_{\gamma_{0,1}} \frac{\dz}{z(z-w)} \frac{w^{\xi'+N}(z-1)^{N}}{z^{\xi+N}(w-1)^N} r_{\beta,1}(w)^{\nhalf-m'} F_{\beta,1}(w)\Phi_\eps(z)^{m-\nhalf}\\
        + \frac{1}{(2\pi\imt)^2} \oint_{\gamma_1} \dw \oint_{\gamma_{0,1}} \frac{\dz}{z(z-w)} \frac{w^{\xi'+N}(z-1)^{N}}{z^{\xi+N}(w-1)^N} r_{\beta,1}(w)^{\nhalf-m'} \frac{2F_{\beta,1}(w)F_{\alpha,1}(w)F_{\beta,2}(w)}{1+2g_{\alpha,\beta}(w)} \Phi_\eps(z)^{m-\nhalf}
    \end{multline}
    when $\nicefrac{N}{2} < m' <N$. For both cases, if $m \leq \nicefrac{N}{2}$ then $\eps = \alpha$ and if $m > \nicefrac{N}{2}$ then $\eps = \beta$. Additionally, $\gamma_1$ is a contour surrounding $1$ and not $0$, while $\gamma_{0,1}$ is a contour surrounding $0$ and $\gamma_1$. Both are positively oriented. We use the index $j$ to denote the matrix row and the index $i$ to denote the matrix column.  
\end{theorem}
\begin{figure}[h]
    \centering
    \begin{tikzpicture}
    \draw (0,-2) -- (0,2);
    \draw (-2,0) -- (2,0);
    \clip (-3,-2) rectangle (3,2);
    \filldraw[black] (1,0) circle (2pt) node[anchor=north]{\footnotesize{1}};
    \draw[ultra thick, blue] (0.25,0) ellipse (1.5 and 1.25);
    \draw[ultra thick, red] (1,0) circle (0.5);
    \draw[- stealth, ultra thick, red] (1.1,0.5) -- (0.95,0.5);
    \draw[-stealth, ultra thick, blue] (0.4,1.25) -- (0.2,1.25);
    \node at (0.7,-0.7) {$\gamma_1$};
    \node at (2.1,0.3) {$\gamma_{0,1}$};
    \end{tikzpicture}
    \caption{Depiction of the contours $\gamma_{0,1}$ and $\gamma_1$.}
    \label{fig:maincontour}
\end{figure}

\begin{remark} \label{rmk:1}
    This is not the complete version of the kernel as it only accounts for vertices of the form $(4m, 2\xi+i)$. In particular, there is a restriction on the horizontal coordinate being divisible by four. One can extend this kernel to account for all vertices by appending by certain matrix products that depend on the coordinates. For example, in the case where $m' \leq \nicefrac{N}{2}$ we can extend the kernel by writing,
    \begin{multline} \label{eq:extendedkernel}
        \Big[\mathbb{K}_N(4m'-k,2\xi'+j;4m+l,2\xi+i)\Big]_{i,j=0}^1 \\
        = -\frac{\mathbb{I}_{m>m'}}{2 \pi \im} \oint_{\gamma_{0,1}} \frac{\dz}{z} z^{\xi'-\xi} T_{k,\alpha}(z)\Phi_{\alpha}(z)^{\frac{N}{2}-m'}\Phi_{\eps}(z)^{m-\frac{N}{2}}S_{l,\eps}(z) \\
        + \frac{1}{(2\pi\im)^2} \oint_{\gamma_1} \dw \oint_{\gamma_{0,1}} \frac{\dz}{z(z-w)} \frac{w^{\xi'+N}(z-1)^{N}}{z^{\xi+N}(w-1)^N} r_{\alpha,1}(w)^{\nhalf-m'} T_{k,\alpha}(w)F_{\alpha,1}(w)\Phi_\eps(z)^{m-\nhalf}S_{l,\eps}(z)\\
        + \frac{1}{(2\pi\im)^2} \oint_{\gamma_1} \dw \oint_{\gamma_{0,1}} \frac{\dz}{z(z-w)} \frac{w^{\xi'+N}(z-1)^{N}}{z^{\xi+N}(w-1)^N} r_{\alpha,2}(w)^{\nhalf-m'} T_{k,\alpha}(w)\frac{2 F_{\alpha,2}(w)F_{\beta,1}(w)F_{\alpha,1}(w)}{1+2g_{\alpha,\beta}(w)} \\
        \times\Phi_\eps(z)^{m-\nhalf} S_{l,\eps}(z)
    \end{multline}
    where 
    $$S_{0,\eps}(z) = \mathbb{I}, \; S_{1,\eps}(z) = \phi_{\eps,1}(z), \; S_{2,\eps}(z) = \phi_{\eps,1}(z)\phi_{\eps,2}(z), \; S_{3,\eps}(z) = \phi_{\eps,1}(z)\phi_{\eps,2}(z)\phi_{3}(z)$$
    and 
    $$T_{0,\eps}(z) = \mathbb{I}, \; T_{1,\eps}(z) = \phi_{4}(z), \; T_{2,\eps}(z) = \phi_{3}(z)\phi_{4}(z), \; T_{3,\eps}(z) = \phi_{\eps,2}(z)\phi_{3}(z)\phi_{4}(z).$$
    Now the kernel can be computed for any vertex. In this work, we will only use the limited version of the kernel presented in Theorem \ref{theorem:kernel} instead of this longer version. This is because the limited kernel has enough detail for understanding the asymptotic results of the model, without becoming overly excessive in notation. See \cite[Remark 2.6]{DK21} for a similar conversation regarding the kernel of the two-periodic Aztec diamond. 
\end{remark}

The proof of Theorem \ref{theorem:kernel} is in Section \ref{section:proofofkernel1}, with some details deferred to Sections \ref{section:BDstuff} and \ref{section:analysisofeigen}. The kernel presented in Theorem \ref{theorem:kernel} is written in a way that makes comparisons to the typical two-periodic Aztec diamond more immediate. There is a way to represent the kernel with only one double contour integral. This form of the kernel is shown in equation \eqref{eq:bsidekernelinitial}, which is located in the proof of Theorem \ref{theorem:kernel} when $m' > \nicefrac{N}{2}$. 
 
%%%%%%%%%%%%%%%%%%%%%%%%%%%%%%%%%%%%%%%%%%%%%%%%%%%%%%%%%%%%%%%%%%%%%%%%%%%%%%%%%%%%%
\subsection{Comparison to the Two-Periodic Aztec Diamond} \label{section:twoperkernel}
Before we analyze the behavior of the model, we first want to comment on how the form of the correlation kernel relates to the correlation kernel of the typical two-periodic Aztec diamond. From the definition of the model, we can recover the two-periodic Aztec diamond (with two-periodic parameter $\eps$) from the split two-periodic Aztec diamond by setting $\alpha = \beta = \eps$. The correlation kernel of the two-periodic Aztec diamond, as formulated in \cite{BD19}, can be stated using the notation presented in this paper as 
\begin{theorem*}[\cite{BD19}, Theorem 5.2]
    Let $-N \leq \xi,\xi' \leq -1$ and $0 < m < N$. The two-periodic Aztec diamond of size $2N$ has correlation kernel given by 
    \begin{multline}
        \Big[\mathbb{K}_{TP,N}^{\eps}(4m',2\xi'+j;4m,2\xi+i)\Big]_{i,j=0}^1 = -\frac{\mathbb{I}_{m>m'}}{2 \pi \imt} \oint_{\gamma_{0,1}} \frac{\dz}{z} z^{\xi'-\xi} \Phi_{\eps}(z)^{m-m'} \\
        + \frac{1}{(2\pi\imt)^2} \sum_{k = 1,2} \oint_{\gamma_1} \dw \oint_{\gamma_{0,1}} \frac{\dz}{z(z-w)} \frac{w^{\xi'+N}(z-1)^{N}}{z^{\xi+N}(w-1)^N} r_{\eps,1}(w)^{\nhalf-m'} r_{\eps,k}(z)^{m-\frac{N}{2}} F_{\eps,1}(w)F_{\eps,k}(z).
    \end{multline}
\end{theorem*}
\begin{remark}
    Once again, we have written the kernel only for coordinates of the form $(4m,2\xi+i)$. The kernel above can be extended to all coordinates via the method outlined in Remark \ref{rmk:1}. The matrix notation above is the same convention as the matrix notation in Theorem \ref{theorem:kernel}.
\end{remark}
At the level of the correlation kernel, we recover the kernel of the typical two-periodic Aztec diamond when we set $\alpha = \beta = \eps$. This is because the second double contour integral presented in equations \eqref{eq:asidekernel} and \eqref{eq:bsidekernelfinal} will contain the matrix products $F_{\alpha,1}(w)F_{\alpha,2}(w)$ and $F_{\alpha,2}(w)F_{\alpha,1}(w)$, respectively. Both of these products are exactly zero. 

Now consider the split two-periodic Aztec diamond under the condition that the coordinates lie on the same side of the interface, e.g. $m, \, m' \leq \nicefrac{N}{2}$ or $m,\,m' \geq \nicefrac{N}{2}$. The correlation kernel under these conditions simplifies to,
\begin{multline} \label{eq:samesidecorker}
    \Big[\mathbb{K}_N(4m',2\xi'+j;4m,2\xi+i)\Big]_{i,j=0}^1 = \Big[\mathbb{K}_{TP,N}^{\eps}(4m',2\xi'+j;4m,2\xi+i)\Big]_{i,j=0}^1\\
    + \frac{1}{(2\pi\im)^2} \sum_{k = 1,2} \oint_{\gamma_1} \dw \oint_{\gamma_{0,1}} \frac{\dz}{z(z-w)} \frac{w^{\xi'+N}(z-1)^{N}}{z^{\xi+N}(w-1)^N} r_{\eps,l}(w)^{\nhalf-m'} r_{\eps,k}(z)^{m-\frac{N}{2}} \\
    \times  \frac{2}{1+2g_{\alpha,\beta}(w)} F_{\eps,l}(w)F_{\bar{\eps},1}(w)F_{\eps,\bar{l}}(w)F_{\eps,k}(z) \\ 
\end{multline}
where 
\begin{equation*}
    (\eps, \bar{\eps}) = \begin{cases}
    (\alpha,\beta) & \text{if } m, \, m' < \nicefrac{N}{2} \\
    (\beta, \alpha) & \text{if } m, \, m' > \nicefrac{N}{2} 
    \end{cases}
\end{equation*}
and 
\begin{equation*}
    (l, \bar{l}) = \begin{cases}
    (2,1) & \text{if } m, \, m' < \nicefrac{N}{2} \\
    (1,2) & \text{if } m, \, m' > \nicefrac{N}{2} 
    \end{cases}.
\end{equation*}
From this expression we see how the correlation kernel of the split model relates to the typically two-periodic model. In order to understand how the local asymptotics of the split two-periodic Aztec diamond differ from the local asymptotics of the typical two-periodic Aztec diamond we only need to understand the asymptotics of the additional two double contour integrals. These will be our main focus when describing the local asymptotics in Section \ref{section:asymptoticsstatement}.\\

%%%%%%%%%%%%%%%%%%%%%%%%%%%%%%%%%%%%%%%%%%%%%%%%%%%%%%%%%%%%%%%%%%%%%%%%%%%%%%%%%%%
\subsection{Classification of the Macroscopic Regions} \label{section:classofphases}
Before we compute the explicit local asymptotics of the model, we would like to introduce the saddle functions associated with the kernel of the model. In doing so, we will consider the behavior of the model around some \textit{asymptotic coordinate}. Let $(x,y) \in (0,\nicefrac{1}{2})\cup(\nicefrac{1}{2},1) \times (-1,0)$, where we omit the case where the asymptotic coordinate is exactly on the interface of the model. Let $(m,\xi) = (xN, yN)$, then we say the vertex $(4m,2\xi+i)$ has asymptotic coordinate $(x,y)$.\footnote{Calling $(x,y)$ the asymptotic coordinate means that the asymptotic picture fits inside the square $[0,1] \times [-1,0]$, which is the typical asymptotic view of the Aztec diamond.} Now we define the functions,
\begin{equation} \label{eq:saddlefunc}
    \psi_{\eps,k}(z;x,y) = (y+1)\log z - \log(z-1) + \left(\frac{1}{2}-x\right)\log r_{\eps,k}(z)
\end{equation}
where $k = 1,2$. These functions play a crucial role in the asymptotics of the typical two-periodic Aztec diamond and are also key to understanding the asymptotics of the split two-periodic model. 
\begin{remark}
    The saddle functions $\psi_{\eps,k}(z)$ are, up to a scalar and change of parameters, the saddle functions defined and analyzed in \cite{DK21}.
\end{remark}
In other words, most of the analytic work we need in order to prepare our integrals for saddle point methods was accomplished in \cite{DK21}. We state the saddle functions in \cite[p. 15]{DK21} for comparison. 
\begin{equation*}
    \Phi_k(z) = -(1+\xi_2)\log z + 2 \log(z-1) + \xi_1 \log \lambda_k(z)
\end{equation*}
where $k=1,2$. They define 
\begin{equation*}
    \lambda_{1,2}(z) = \frac{\left((\alpha_{DK}+\beta_{DK})z\pm\sqrt{z(z+\alpha_{DK}^2)(z+\beta_{DK}^2)}\right)^2}{z(z-1)^2}
\end{equation*}
and $-1<\xi_1,\xi_2<1$ are parameters that define the asymptotic coordinate, albeit different from our $x$ and $y$. Note that the above $\alpha_{DK}$ and $\beta_{DK}$ are from the convention in \cite{DK21} and not the same $\alpha$ and $\beta$ presented in this text. Letting $\alpha_{DK} = \eps$ and $\beta_{DK} = \eps^{-1}$ it is easy to check that, 
\begin{equation*}
    r_{\eps,k}(z) = \lambda_k(z)
\end{equation*}
for $k = 1,2$. Using the change of parameter, $\xi_1 = 2x-1$ and $\xi_2 = 2y+1$, we obtain 
\begin{equation*} \label{eq:phipsi1}
    \Phi_k(z) = - 2 \psi_{\eps,k}(z;x,y)
\end{equation*}
Much like the two-periodic Aztec diamond, the split two-periodic model contains all three possible regions: smooth, rough, and frozen. The following piece-wise saddle function will help us succinctly define the boundary between these macroscopic regions. 
\begin{equation} \label{eq:saddlefunctionpw}
    \psi_k(z) = \begin{cases}
        \psi_{\alpha,k}(z; x,y) & \text{for } x < \nicefrac{1}{2} \\
        \psi_{\beta,k}(z; x,y) & \text{for } x > \nicefrac{1}{2}
    \end{cases}
\end{equation}
For fixed $(x,y)$ we can think of the saddle functions $\psi_1(z)$ and $\psi_2(z)$ as a single function, $\Psi(z)$, on some Riemann surface. The Riemann surface is two sheets of $\mathbb{C}$ glued together along the branch cuts of $\psi_k(z)$.\footnote{The branch cuts are dependent on $x$. If $x<\nicefrac{1}{2}$, then the branch cuts are $(-\infty,-\alpha^{-2}]\cup[-\alpha^2,0]$, while if $x > \nicefrac{1}{2}$ the branch cuts are $(-\infty,-\beta^{-2}]\cup[-\beta^2,0]$.} On one sheet of the surface, $\Psi(z) = \psi_1(z)$, while on the other sheet, $\Psi(z) = \psi_2(z)$. Note that we can consider $\Psi(z)$ to be a multivariate function on some Riemann surface as long as we restrict our asymptotic coordinate to one side of the interface.

There are four saddle points, counting multiplicity, of the function $\Psi(z)$ as stated in \cite{DK21}. The location of these saddle points on the Riemann surface classify the macroscopic region of the model,
\begin{definition} \label{def:macroscopicboundary}
Let $(x,y) \in (0,\nicefrac{1}{2})\cup(\nicefrac{1}{2},1) \times (-1,0)$ denote the asymptotic coordinate on the Aztec diamond.
    \begin{enumerate}
        \item If $\Psi$ has two simple saddle points on either copy of the positive real axis, then $(x,y) \in \mathfrak{F}$, the frozen region of the model.
        \item If $\Psi$ has two complex saddle points (on either sheet of the Riemann surface), then $(x,y) \in \mathfrak{R}$, the rough region of the model. 
        \item If all four saddle points of $\Psi$ are simple and located on either copy of the negative real axis away from the branch cuts, then $(x,y) \in \mathfrak{S}$, the smooth region of the model.
    \end{enumerate}
\end{definition}
This is the same definition of regions presented for the typical two-periodic model in \cite[Definition 2.9]{DK21}. Moreover, Duits and Kuijlaars go on to define the phase transitions as tuples, $(x,y)$, where the saddle points of the function coalesce. We do not state this full definition here, but the same applies as long as $x \neq \nicefrac{1}{2}$. The saddle function, $\Psi(z)$, can also be used to determine the boundaries between the macroscopic regions. Any point $(x,y)$ where saddle points coalesce is on some boundary between the macroscopic regions. Additionally, due to the piece wise nature of the saddle functions, we must check for transitions along the line $x = \nicefrac{1}{2}$. Figure \ref{fig:boundariesexample} depicts an example of the macroscopic boundaries for a split two-periodic model.
\begin{figure}[h]
    \centering
    \includegraphics[scale=0.5]{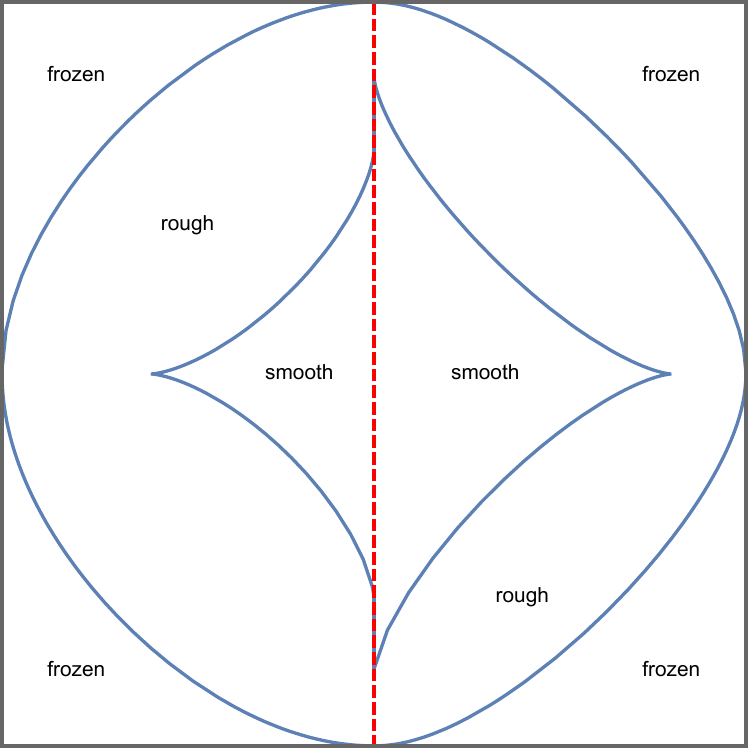}
    \caption{Boundaries between macroscopic regions of the split Aztec diamond when $\alpha = \nicefrac{1}{2}$ and $\beta = \nicefrac{1}{3}$. The dashed red line denoted the interface of the model. Boundary computed numerically using Mathematica.}
    \label{fig:boundariesexample}
\end{figure}

For either half of the model, the definition and boundaries between macroscopic regions agree with the typically two-periodic case.\footnote{See \cite{DK21} for a similar classification in the two periodic case.} In the section below we will show how this follows from the local asymptotics of the correlation kernel.

%%%%%%%%%%%%%%%%%%%%%%%%%%%%%%%%%%%%%%%%%%%%%%%%%%%%%%%%%%%%%%%%%%%%%%%%%%%%%%%%%%%%%%%%%%%%%%%%%%%%%%%%%%%%%%%%%%%%%%%%%%%%%%%%%%%%%%%%%%%%%%%%%%%%%%%%%%%%%%%%%%%%%%%%%%%%%%%%%%%%%%%%%%%%%%%%%%%%%%%%%%%%%%%%%%%%%%%%%%%%%%%%%%%%
\subsection{Asymptotic Behavior} \label{section:asymptoticsstatement}
Our goal in this section is to summarize the local asymptotics of the model, so we will assume the two coordinates $(4m', 2\xi'+j)$ and $(4m, 2\xi+i)$ are on the same side of the interface and that they are asymptotically close. We define this precisely. Start by writing, 
\begin{align} 
    m = xN + x_1 && m' = xN + x_2 \label{eq:localm}\\ 
    \xi = yN + y_1 && \xi' = yN + y_2 \label{eq:localxi}
\end{align}
where we assume $x \neq \nicefrac{1}{2}$, and $x_1$ and $x_2$ are small enough in magnitude so that the coordinates $(4m,2\xi+i)$ and $(4m',2\xi'+j)$ lie on the same side of the interface. In general, we will assume that $x_1, \, x_2, \, y_1,$ and $y_2$ are $o(N)$. Under these considerations, we rewrite the correlation kernel in terms of the saddle functions,
\begin{multline}
    \Big[\mathbb{K}_N(4m',2\xi'+j;4m,2\xi+i)\Big]_{i,j=0}^1 = \Big[\mathbb{K}_{TP,N}^{\eps}(4m',2\xi'+j;4m,2\xi+i)\Big]_{i,j=0}^1\\
    + \frac{1}{(2\pi\im)^2} \sum_{k = 1,2} \oint_{\gamma_1} \dw \oint_{\gamma_{0,1}} \frac{\dz}{z(z-w)}\frac{2}{1+2g_{\alpha,\beta}(w)} F_{\eps,l}(w)F_{\bar{\eps},1}(w)F_{\eps,\bar{l}}(w)F_{\eps,k}(z) \\
    \times \exp \left[ N\left(\psi_l(w;x,y) - \psi_k(z;x,y)\right) + \varphi_{\eps,l}(w;x_2,y_2) - \varphi_{\eps,k}(z;x_1,y_1) \right]
\end{multline}
We are using the same definitions of $\eps, \, \bar{\eps}, \, l,$ and $\bar{l}$ as defined below equation \eqref{eq:samesidecorker}. The saddle functions $\psi_k(z)$ are defined by equation \eqref{eq:saddlefunctionpw} and we define the functions $\varphi_{\eps,k}(z;x_1,y_1)$ as
\begin{equation}\label{eq:varphidef}
    \varphi_{\eps,k}(z;x_1,y_1) = y_1 \log z - x_1\log r_{\eps,k}(z)
\end{equation}
For brevity we let
\begin{multline} \label{eq:I2kdef}
    \Big[I^\eps_{l,k}(4m',2\xi'+j;4m,2\xi+i)\Big]_{i,j=0}^1 \\
    = \frac{1}{(2\pi\im)^2} \oint_{\gamma_1} \dw \oint_{\gamma_{0,1}} \frac{\dz}{z(z-w)}\frac{2}{1+2g_{\alpha,\beta}(w)} F_{\eps,l}(w)F_{\bar{\eps},1}(w)F_{\eps,\bar{l}}(w)F_{\eps,k}(z) \\
    \times \exp \left[ N\left(\psi_l(w;x,y) - \psi_k(z;x,y)\right) + \varphi_{\eps,l}(w;x_2,y_2) - \varphi_{\eps,k}(z;x_1,y_1) \right]
\end{multline}
for $k = 1, 2$. So we may write the kernel, compactly, as 
\begin{multline}
    \Big[\mathbb{K}_N(4m',2\xi'+j;4m,2\xi+i)\Big]_{i,j=0}^1 = \Big[\mathbb{K}_{TP,N}^{\eps}(4m',2\xi'+j;4m,2\xi+i)\Big]_{i,j=0}^1 \\
    + \sum_{k=1,2} \Big[I^\eps_{l,k}(4m',2\xi'+j;4m,2\xi+i)\Big]_{i,j=0}^1
\end{multline}
This work will focus on the asymptotic behavior on the $\alpha$-side of the interface, in other words when $m, \, m' < \nicefrac{N}{2}$. In this case the above equation becomes, 
\begin{multline}
    \Big[\mathbb{K}_N(4m',2\xi'+j;4m,2\xi+i)\Big]_{i,j=0}^1 = \Big[\mathbb{K}_{TP,N}^{\alpha}(4m',2\xi'+j;4m,2\xi+i)\Big]_{i,j=0}^1 \\
    + \sum_{k=1,2} \Big[I^\alpha_{2,k}(4m',2\xi'+j;4m,2\xi+i)\Big]_{i,j=0}^1
\end{multline}
Ultimately, our goal is to show that the contributions from the contour integrals $I^\alpha_{2,1}$ and $I^\alpha_{2,2}$ tend to $0$ as $N \to \infty$. This will justify the definition of the macroscopic boundary given in Definition \ref{def:macroscopicboundary}. The order of decay depends on the region of the model of the asymptotic coordinate. We will find that, in certain regions, these additional terms affect the second order term of the local asymptotics of $\mathbb{K}_N$. \\

Before we state the asymptotic results, we will give a brief justification on why we limit our analysis to the $\alpha$-side of the interface. The reason this is justified is due to a symmetry between the saddle functions $\psi_{\eps,1}(z)$ and $\psi_{\eps,2}(z)$. In particular, since $r_{\eps,1}(z) = r_{\eps,2}(z)^{-1}$, this implies 
$$\psi_{\eps,1}(z;x,y) = \psi_{\eps,2}(z;1-x,y)$$
which means there is a symmetry between the saddle functions when the lattice is flipped along the interface. This symmetry is not enough for equality between the contour integrals, but it is enough, along with the proofs detailed in Section \ref{section:asymptoticresults}, to conclude that the asymptotic behavior is the same. The difference in matrices still present in the integral only change a constant factor, not the rate of decay. If we let $m$, $m'$, $\xi$, $\xi'$ be defined by equations \eqref{eq:localm} and \eqref{eq:localxi} and $m,\,m' < N/2$, then $I^\eps_{2,2}(4m',2\xi'+j;4m,2\xi+i)$ has the same asymptotic behavior as $I^\eps_{1,2}(4(N-m'),2\xi'+j;4(N-m),2\xi+i)$. Additionally, $I^\eps_{2,1}(4m',2\xi'+j;4m,2\xi+i)$ has the same asymptotic behavior as $I^\eps_{1,1}(4(N-m'),2\xi'+j;4(N-m),2\xi+i)$. Thus below we only state asymptotic results for $I_{2,2}^\alpha$ and $I_{2,1}^\alpha$. The asymptotics on the $\beta$-side of the interface can be proven in an identical manner. 
\begin{remark}
    There is something interesting going on with regard to the symmetry of the model. There is not a discrete symmetry, like in the two-periodic model. One could try to switch parameters, flip the lattice, and then gauge transform but the interface becomes a problem and there exists no gauge transform. The is, however, a symmetry of the limit shape because the additional integrals (compared to the two-periodic model) behave the same, up to a scalar, asymptotically. 
\end{remark}
To examine all possible cases of model behavior we will consider the following two cases: (a) $0 < \beta < \alpha < 1$, and (b) $0 < \alpha < \beta < 1$. This is because the asymptotic behavior of $I_{2,2}^\alpha$ not only depends on the macroscopic regions of the model, but, if we are in case (b) above, it also depends on the region of the model in other ways. We restrict to the case of $x < \nicefrac{1}{2}$ and case (b) above. We can then consider the regions of the model where $(x,y)$ satisfies one of the following inequalities: 
\begin{enumerate}[label=(\roman*)]
    \item $\psi'_{\alpha,1}(-\beta^2;x,y) > 0$
    \item $\psi'_{\alpha,1}(-\beta^{-2};x,y) < 0$
\end{enumerate}
Since $\psi'_{\alpha,1}(z)$ is linear in $x$ and $y$, the inequalities just denote trapezoidal regions of the model. As we will show later in Lemma \ref{lem:psimax}, $\oRe \psi_{\alpha,1}(z)$ has exactly one maximum on the interval $[-\alpha^{-2},-\alpha^2]$, thus we conclude that no $(x,y)$ satisfies both inequalities. In the case of $\beta < \alpha$ we define no such region. The exact size and shape of these regions depends on the exact values of $\alpha$ and $\beta$. Figure \ref{fig:I22regions} gives examples of these regions for different values of $\alpha$ and $\beta$. 
\begin{figure}[h]
    \centering
    \rotatebox[origin=c]{180}{\includegraphics[scale=0.4]{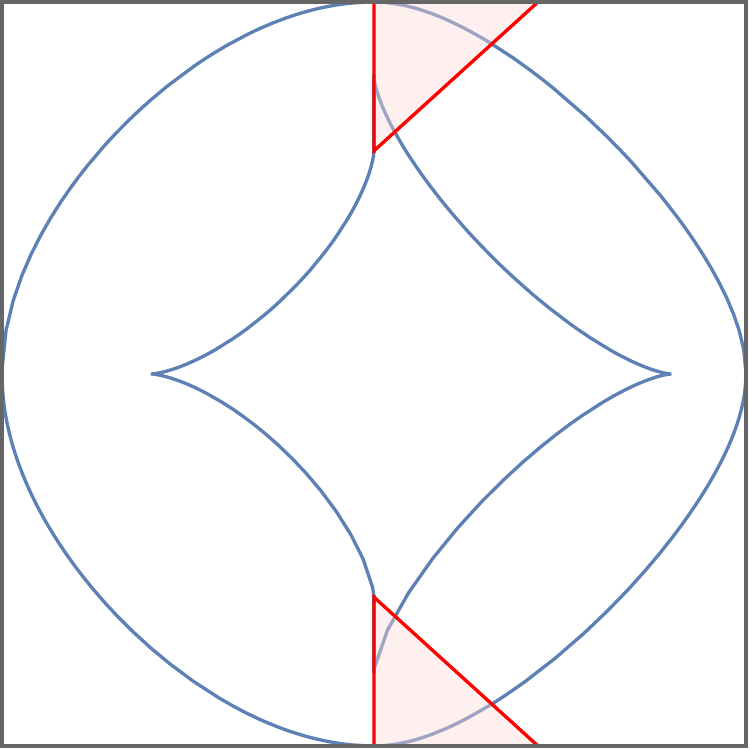}}
    \hspace{1cm}
    \rotatebox[origin=c]{180}{\includegraphics[scale=0.4]{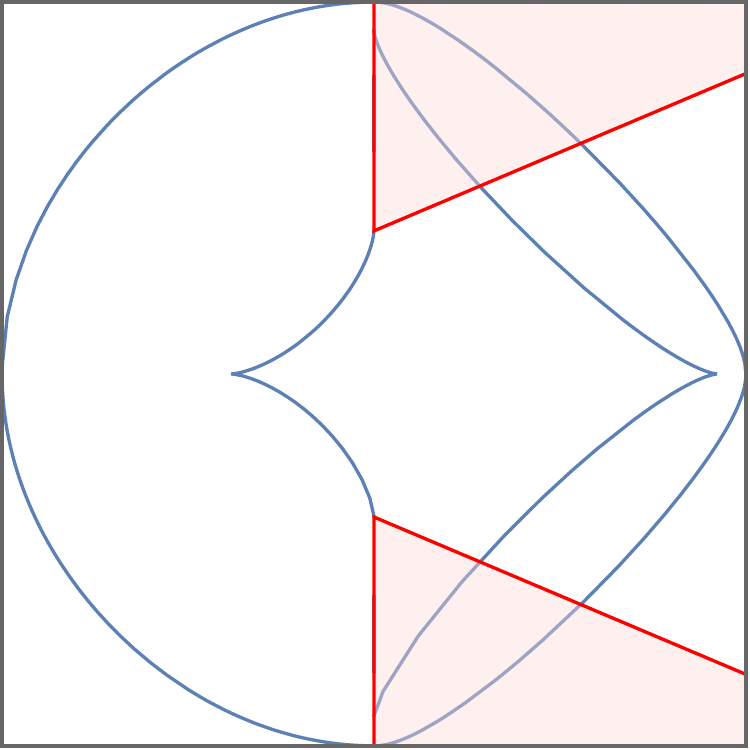}}
    \caption{Regions of the model where (i) and (ii) are satisfied are shown in red. We refer to these as the strong-coupling region.}
    \label{fig:I22regions}
\end{figure}
We will denote the asymptotic region where either (i) or (ii) is satisfied as the \textit{strong-coupling} region of the model, and denote it by $\mathfrak{C}$. As we will see below, the coupling between the two weightings of the model is strongest in this region. We state the following asymptotic results regarding the behavior of $I^\alpha_{2,2}$,
\begin{proposition} \label{prop:I22decay}
    Let $(m,\xi)$ and $(m',\xi')$ be given by equations \eqref{eq:localm} and \eqref{eq:localxi} such that $x < \nicefrac{1}{2}$. 
    \begin{enumerate}
        \item If $(x,y) \in \mathfrak{R}$, then 
        \begin{equation*}
            |I^\alpha_{2,2}(4m',2\xi'+j;4m,2\xi+i)| = c N^{-1} + o(N^{-1})
        \end{equation*}
        for some positive constant $c$.
        \item If $(x,y) \in \mathfrak{S}$ and $(x,y) \not\in \mathfrak{C}$, then 
        \begin{equation*}
            |I^\alpha_{2,2}(4m',2\xi'+j;4m,2\xi+i)| = ce^{-c_1N} + o(e^{-c_1N})
        \end{equation*}
        for some positive constants $c$ and $c_1$.
        \item If $(x,y) \in \mathfrak{S} \cap \mathfrak{C}$, then 
        \begin{equation*}
            |I^\alpha_{2,2}(4m',2\xi'+j;4m,2\xi+i)| = c N^{-1} + o(N^{-1})
        \end{equation*}
        for some positive constant $c$.
    \end{enumerate}
\end{proposition}
The exact constants can be extracted from the proof in Section \ref{section:asymptoticresults}. The behavior of the second double contour integral, $I^\alpha_{2,1}$, is less complicated as it does not depend on the macroscopic region of the model. We state the following proposition describing its behavior,
\begin{proposition} \label{prop:I21decay}
    Let $(m,\xi)$ and $(m',\xi')$ be given by equations \eqref{eq:localm} and \eqref{eq:localxi} such that $x < \nicefrac{1}{2}$, then 
    \begin{equation*}
        |I^\alpha_{2,1}(4m',2\xi'+j;4m,2\xi+i)| = c e^{-c_1N} + o(e^{-c_1N})
    \end{equation*}
    for some positive constants $c$ and $c_1$.
\end{proposition}
Moreover, the inequalities presented in the above propositions are optimal. Lastly, we should mention how these results combine with the local asymptotic behavior of the typical two-periodic kernel. The local asymptotics for the two-periodic Aztec diamond were initially stated in \cite[Theorem 2.6]{CJ16}. Let $\mathbb{K}_{FP}^\eps(\bar{p}_1;\bar{p}_2)$ denote the full plane correlation kernel for the infinite square lattice with the same two-periodic weighting as the two-periodic Aztec diamond described in Section \ref{section:twoperkernel}. The results from \cite[Theorem 2.6]{CJ16} and also \cite[Theorem 2.10]{DK21} state\footnote{The original results in \cite{CJ16} just focus along a diagonal of the model, the later results \cite{DK21} look at the entire smooth region.} that if $(x,y) \in \mathfrak{S}$ then, 
\begin{equation}
    \mathbb{K}^\eps_{TP,N}(4m',2\xi'+j;4m,2\xi+i) = \mathbb{K}_{FP}^\eps(x_1,y_1;x_2,y_2) + ce^{-c_1N} + O(e^{-c_1N}),
\end{equation}
where $c$ is a non-zero constant. However, in the case of the split two-periodic Asymptotic we see a difference in the sub leading order term. In particular, we wist to highlight that if $(x,y) \in \mathfrak{S} \cap \mathfrak{C}$ then
\begin{equation}
    \mathbb{K}^\eps_{TP,N}(4m',2\xi'+j;4m,2\xi+i) = \mathbb{K}_{FP}^\eps(x_1,y_1;x_2,y_2) + cN^{-1} + O(N^{-1}),
\end{equation}
where $c$ is again a non-zero constant. In the region $\mathfrak{S} \cap \mathfrak{C}$, not only do the additional integrals contribute to the sub-leading order term, but they change the decay of it significantly. Note that this does not happen in $\mathfrak{R} \cap \mathfrak{C}$, because in the rough region the two-periodic kernel converges to the related full plane kernel with an error of $N^{-1/2}$.

%%%%%%%%%%%%%%%%%%%%%%%%%%%%%%%%%%%%%%%%%%%%%%%%%%%%%%%%%%%%%%%%%%%%%%%%%%%%%%%%%%%%%%%%%%%%%%%%%%%%%%%%%%%%%%%%%%%%%%%%%%%%%%%%%%%%%%%%%%%%%%%%%%%%%%%%%%%%%
\section{Derivation of the Correlation Kernel} \label{section:proofofkernel1}
Our main goal in this section is to prove the results stated in Theorem \ref{theorem:kernel}. In order to prove Theorem \ref{theorem:kernel}, we start in Section \ref{section:intermediatestatement} by stating an intermediate version of the kernel which can be derived by using techniques from \cite{BD19}. The form of the kernel differs slightly depending on whether the coordinate $(4m',2\xi'+j)$ is on the $\alpha$-side or $\beta$-side of the interface. In Section \ref{section:mlessN/2proof}, we prove the statement of the kernel for $m' \leq \nicefrac{N}{2}$ and in Section \ref{section:m>N/2}, we prove the statement of the kernel for $m' > \nicefrac{N}{2}$.

\subsection{An Intermediate Correlation Kernel} \label{section:intermediatestatement}
We start by introducing some more necessary notation. We define,  
\begin{equation} \label{eq:phiNprod}
    \phi_N(z) = \Phi_\alpha(z)^\nhalf\Phi_\beta(z)^\nhalf
\end{equation}
and its eigen-decomposition,
\begin{equation} \label{eq:phiNdecomp}
    \phi_N(z) = r_{1,N}(z)F_{1,N}(z) + r_{2,N}(z)F_{2,N}(z).
\end{equation}
The explicit formulas for the eigenvalues and matrices are quite involved. We can, however, write the trace of $\phi_N(z)$ in a rather succinct way. We let,
\begin{multline} \label{eq:trace}
    t_N(z) = \text{tr } \phi_N(z) = \left(\frac{1}{2} + g_{\alpha,\beta}(z)\right)\left(r_{\alpha,1}(z)^{\nhalf}r_{\beta,1}(z)^{\nhalf} + r_{\alpha,2}(z)^{\nhalf}r_{\beta,2}(z)^{\nhalf}\right) \\
    +\left(\frac{1}{2} - g_{\alpha,\beta}(z)\right)\left(r_{\alpha,2}(z)^{\nhalf}r_{\beta,1}(z)^{\nhalf} + r_{\alpha,1}(z)^{\nhalf}r_{\beta,2}(z)^{\nhalf}\right) 
\end{multline}
which follows from the formulas in \eqref{eq:f1}, \eqref{eq:f2}, \eqref{eq:g}. Since $\det \phi_N(z) = 1$, the eigenvalues of $\phi_N(z)$ can be expressed as, 
\begin{align}
    r_{1,N}(z) = \frac{1}{2}\left(t_N(z) + \sqrt{t_N(z)^2-4}\right) && r_{2,N}(z) = \frac{1}{2}\left(t_N(z) - \sqrt{t_N(z)^2-4}\right)
\end{align}
If $E_N(z)$ is the eigenvector matrix that respects this ordering of eigenvalues, then $F_{k,N}(z)$ are the matrices, 
\begin{align}
    F_{1,N}(z) = E_N(z)\begin{pmatrix}1&0\\0&0\end{pmatrix}E_N(z)^{-1} && F_{2,N}(z) = E_N(z)\begin{pmatrix}0&0\\0&1\end{pmatrix}E_N(z)^{-1}
\end{align}
We now may state, 
\begin{lemma} \label{lemma:BDstuff}
   Let $0 < m < N$, and $-N \leq \xi, \xi' \leq -1$. The split two-periodic Aztec diamond of size $2N$ has a correlation kernel given by,
    \begin{multline}
        \Big[\mathbb{K}_N(4m',2\xi'+j;4m,2\xi+i)\Big]_{i,j=0}^1 = -\frac{\mathbb{I}_{m>m'}}{2 \pi \text{\emph{i}}} \oint_{\gamma_{0,1}} \frac{\dz}{z} z^{\xi'-\xi} \Phi_{\alpha}(z)^{\frac{N}{2}-m'}\Phi_{\eps}(z)^{m-\frac{N}{2}} \\
        + \frac{1}{(2\pi\text{\emph{i}})^2} \oint_{\gamma_1} \dw \oint_{\gamma_{0,1}} \frac{\dz}{z(z-w)} \frac{w^{\xi'+N}(z-1)^{N}}{z^{\xi+N}(w-1)^N} \Phi_\alpha(w)^{-m'}F_{1,N}(w)\Phi_\alpha(w)^{\frac{N}{2}}\Phi_\eps(z)^{m-\frac{N}{2}}
    \end{multline}
    when $0 < m' \leq \nicefrac{N}{2}$. Alternatively, when $\nicefrac{N}{2} < m' <N$, the correlation kernel is
    \begin{multline}
        \Big[\mathbb{K}_N(4m',2\xi'+j;4m,2\xi+i)\Big]_{i,j=0}^1 = -\frac{\mathbb{I}_{m>m'}}{2 \pi \imt} \oint_{\gamma_{0,1}} \frac{\dz}{z} z^{\xi'-\xi} \Phi_{\alpha}(z)^{\frac{N}{2}-m'}\Phi_{\eps}(z)^{m-\frac{N}{2}} \\
        + \frac{1}{(2\pi\imt)^2} \oint_{\gamma_1} \dw \oint_{\gamma_{0,1}} \frac{\dz}{z(z-w)} \frac{w^{\xi'+N}(z-1)^{N}}{z^{\xi+N}(w-1)^N} \Phi_{\beta}(w)^{\nhalf-m'}\Phi_{\alpha}(w)^{-\nhalf}F_{1,N}(w)\Phi_\alpha(w)^{\frac{N}{2}}\Phi_\eps(z)^{m-\frac{N}{2}}
    \end{multline}
    The contours $\gamma_1$ and $\gamma_{0,1}$ are the same contours detailed in Theorem \ref{theorem:kernel}.
\end{lemma}
The proof of the above lemma is deferred to Section \ref{section:BDstuff}. While we now have a correlation kernel for the process, this kernel is not set up well for asymptotic analysis. The main part of the proof of Theorem \ref{theorem:kernel} manipulates the above kernel in a way that simplifies the analysis. 
%%%%%%%%%%%%%%%%%%%%%%%%%%%%%%%%%%%%%%%%%%%%%%%%%%%%%%%%%%%%%%%%%%%%%%%%%%%%%%%%%%
\subsection{Proof of Theorem \ref{theorem:kernel} when $0 < m' \leq \nicefrac{N}{2}$} \label{section:mlessN/2proof}
First recall that we can decompose the matrix $\Phi_{\alpha}(w)$ in the following manner,
$$\Phi_{\alpha}(w) = r_{\alpha,1}(w)F_{\alpha,1}(w)+r_{\alpha,2}(w)F_{\alpha,2}(w)$$
where $r_{\alpha,k}(w)$ and $F_{\alpha,k}(w)$ are given by equations \eqref{eq:r1}-\eqref{eq:f2}. Plugging this formula into the double contour integral from Lemma \ref{lemma:BDstuff} gives,
\begin{multline} \label{eq:intermediatekernel}
    \Big[\mathbb{K}_N(4m',2\xi'+j;4m,2\xi+i)\Big]_{i,j=0}^1 = -\frac{\mathbb{I}_{m>m'}}{2 \pi \text{\emph{i}}} \oint_{\gamma_{0,1}} \frac{\dz}{z} z^{\xi'-\xi} \Phi_{\alpha}(z)^{\frac{N}{2}-m'}\Phi_{\eps}(z)^{m-\frac{N}{2}} \\
    + \frac{1}{(2\pi\im)^2} \sum_{k_1 = 1,2} \sum_{k_2 = 1,2} \oint_{\gamma_1} \dw \oint_{\gamma_{0,1}} \frac{\dz}{z(z-w)} \frac{w^{\xi'+N}(z-1)^{N}}{z^{\xi+N}(w-1)^N} r_{\alpha,k_1}(w)^{-m'}r_{\alpha,k_2}(w)^{\frac{N}{2}}\\
    \times F_{\alpha,k_1}(w)F_{1,N}(w)F_{\alpha,k_2}(w)\Phi_\eps(z)^{m-\frac{N}{2}}
\end{multline}
Our goal is to replace the matrix product $F_{\alpha,k_1}(w)F_{1,N}(w)$, as the matrix $F_{1,N}(w)$ is hard to state precisely and thus is not amenable to asymptotic analysis. We start by rearranging equation \eqref{eq:phiNdecomp},
\begin{align}
    F_{1,N}(w) &= r_{1,N}(w)^{-1}\phi_N(w) - r_{1,N}(w)^{-1}r_{2,N}(w)F_{2,N}(w) \label{eq:F1nexp} \\ 
    &= r_{2,N}(w)\phi_N(w) - r_{2,N}(w)^2F_{2,N}(w) \nonumber
\end{align}
Above we have utilized the fact that $r_{1,N}(w) = r_{2,N}(w)^{-1}$. Now we plug the eigen-decompositions of $\Phi_\alpha(w)$ and $\Phi_\beta(w)$ into equation \eqref{eq:phiNprod} to obtain,
\begin{equation*} 
    \phi_N(w) = \Big(r_{\alpha,1}(w)^\nhalf F_{\alpha,1}(w) +  r_{\alpha,2}(w)^\nhalf F_{\alpha,2}(w) \Big) \Big(r_{\beta,1}(w)^\nhalf F_{\beta,1}(w) +  r_{\beta,2}(w)^\nhalf F_{\beta,2}(w) \Big).
\end{equation*}
This means that 
\begin{equation}
    F_{\alpha,k_1}(w)\phi_N(w) = r_{\alpha,k_1}(w)^\nhalf F_{\alpha,k_1}(w)\Big(r_{\beta,1}(w)^\nhalf F_{\beta,1}(w) +  r_{\beta,2}(w)^\nhalf F_{\beta,2}(w) \Big).
\end{equation}
Above we used the fact that $F_{\alpha,1}(z)^2 = F_{\alpha,1}(z)$ and $F_{\alpha,1}(z)F_{\alpha,2}(z) = 0$. Combining the above with equation \eqref{eq:F1nexp} gives, 
\begin{multline} \label{eq:f1nreplacement1}
    F_{\alpha,1}(w)F_{1,N}(w) = r_{2,N}(w)r_{\alpha,1}(w)^\nhalf r_{\beta,1}(w)^\nhalf F_{\alpha,1}(w)F_{\beta,1}(w) \\
    + r_{2,N}(w)r_{\alpha,1}(w)^\nhalf r_{\beta,2}(w)^\nhalf F_{\alpha,1}(w)F_{\beta,2}(w) - r_{2,N}(w)^2F_{\alpha,1}(w)F_{2,N}(w) 
\end{multline}
and
\begin{multline} \label{eq:f1nreplacement2}
    F_{\alpha,2}(w)F_{1,N}(w) = r_{2,N}(w)r_{\alpha,2}(w)^\nhalf r_{\beta,1}(w)^\nhalf F_{\alpha,2}(w)F_{\beta,1}(w) \\
    + r_{2,N}(w)r_{\alpha,2}(w)^\nhalf r_{\beta,2}(w)^\nhalf F_{\alpha,2}(w)F_{\beta,2}(w) - r_{2,N}(w)^2F_{\alpha,2}(w)F_{2,N}(w) 
\end{multline}
We can now insert these replacements into the double contour integrals of the correlation kernel. There are a total of four double contour integrals in equation \eqref{eq:intermediatekernel}, each having the form, 
\begin{equation*}
    \oint_{\gamma_1} \oint_{\gamma_{0,1}} \frac{\dz\dw}{z(z-w)} \frac{w^{\xi'+N}(z-1)^{N}}{z^{\xi+N}(w-1)^N} r_{\alpha,k_1}(w)^{-m'}r_{\alpha,k_2}(w)^{\frac{N}{2}}F_{\alpha,k_1}(w)F_{1,N}(w)F_{\alpha,k_2}(w)F_{\eps,k}(z)\Phi_\eps(z)^{m-\frac{N}{2}}
\end{equation*}
To prevent our text from getting long, let's focus on only the following product
\begin{equation} \label{eq:1singularityterms}
    \frac{1}{(w-1)^N}r_{\alpha,k_1}(w)^{-m'}r_{\alpha,k_2}(w)^{\nhalf}F_{\alpha,k_1}(w)F_{1,N}(w)
\end{equation}
This product contains all the terms relevant to the singularity (or zero) at $w=1$. We would like to keep track of the order of this singularity/zero, as it is the only potentially non-analytic point inside $\gamma_1$. We will consider the expanded terms of this product when we substitute in equations \eqref{eq:f1nreplacement1} and \eqref{eq:f1nreplacement2}. First, consider the product in \eqref{eq:1singularityterms} when $k_1=k_2=1$,
$$\frac{1}{(w-1)^N}r_{\alpha,k_1}(w)^{-m'}r_{\alpha,k_2}(w)^{\frac{N}{2}}F_{\alpha,k_1}(w)F_{1,N}(w) = \frac{r_{\alpha,1}(w)^{\nhalf-m'}}{(w-1)^N} F_{\alpha,1}(w)F_{1,N}(w)$$
Now using the replacement in \eqref{eq:f1nreplacement1} we get
\begin{multline}
    \frac{r_{\alpha,1}(w)^{\nhalf-m'}}{(w-1)^N} F_{\alpha,1}(w)F_{1,N}(w) = \frac{r_{2,N}(w)r_{\alpha,1}(w)^{N-m'} r_{\beta,1}(w)^\nhalf}{(w-1)^N} F_{\alpha,1}(w)F_{\beta,1}(w) \\
    + \frac{r_{2,N}(w)r_{\alpha,1}(w)^{N-m'} r_{\beta,2}(w)^\nhalf}{(w-1)^N} F_{\alpha,1}(w)F_{\beta,2}(w) - \frac{r_{2,N}(w)^2r_{\alpha,1}(w)^{\nhalf-m'}}{(w-1)^N}F_{\alpha,1}(w)F_{2,N}(w)
\end{multline}
So we can consider each of these three terms, and their behavior at $w=1$, separately. The scalar terms,
$$\frac{r_{2,N}(w)r_{\alpha,1}(w)^{N-m'} r_{\beta,2}(w)^\nhalf}{(w-1)^N}$$
and 
$$\frac{r_{2,N}(w)^2r_{\alpha,1}(w)^{\nhalf-m'}}{(w-1)^N}$$
have no singularity at $w = 1$. To justify this statement we need the following two lemmas,
\begin{lemma} \label{lem:rpole}
    The following statements are true regarding the eigenvalues $r_{\eps,1}(z)$ and $r_{\eps,2}(z)$:
    \begin{enumerate}
        \item For any $\eps \in (0,1)$, $r_{\eps,1}(z)$ has a pole of order 2 at $z = 1$ and $r_{\eps,2}(z)$ has a zero of order two at $z = 1$.
        \item For any $\eps \in (0,1)$ and any $z \in \mathbb{C}\backslash \{1\}$, $r_{\eps,1}(z)r_{\eps,2}(z) = 1$.
    \end{enumerate}
\end{lemma}
\begin{lemma} \label{lem:rnpole}
    The eigenvalue $r_{1,N}(z)$ has a pole of order $2N$ at $z=1$ for all $N$ and the eigenvalue $r_{2,N}(z)$ has a zero of order $2N$ at $z=1$ for all $N$.
\end{lemma}
Lemma \ref{lem:rpole} follows from equations \eqref{eq:r1} and \eqref{eq:r2}. Lemma \ref{lem:rnpole} is a direct consequence of Lemma \ref{lem:r2Nexpansion}, which is stated later in this section and proven in Section \ref{section:analysisofeigen}. We also make note of two more lemmas, 
\begin{lemma} \label{lem:Fanalytic}
    For any $\eps \in (0,1)$, the matrix valued functions $F_{\eps,1}(z)$ and $F_{\eps,2}(z)$ are analytic for all $z \in \mathbb{C} \backslash \{ (-\infty, -\eps^{-2}] \cup [-\eps^2,0]\}$.
\end{lemma}
\begin{lemma} \label{lem:FNanalytic}
    The matrix-valued function $F_{2,N}(z)$ has no pole at $z = 1$. 
\end{lemma}
Lemma \ref{lem:Fanalytic} follows from equations \eqref{eq:f1} and \eqref{eq:f2}. The proof of Lemma \ref{lem:FNanalytic} is deferred to Section \ref{section:analysisofeigen}. By the residue theorem we can make the following simplification,
\begin{multline}
    \oint_{\gamma_1} \oint_{\gamma_{0,1}} \frac{\dz\dw}{z(z-w)} \frac{w^{\xi'+N}(z-1)^{N}}{z^{\xi+N}(w-1)^N} r_{\alpha,1}(w)^{\nhalf-m'} F_{\alpha,1}(w)F_{1,N}(w)F_{\alpha,1}(w)\Phi_\eps(z)^{m-\frac{N}{2}}\\
    = \oint_{\gamma_1} \oint_{\gamma_{0,1}} \frac{\dz\dw}{z(z-w)} \frac{w^{\xi'+N}(z-1)^{N}}{z^{\xi+N}(w-1)^N}r_{2,N}(w)r_{\alpha,1}(w)^{N-m'} r_{\beta,1}(w)^\nhalf \\ 
    \times F_{\alpha,1}(w)F_{\beta,1}(w)F_{\alpha,1}(w)\Phi_\eps(z)^{m-\frac{N}{2}}
\end{multline}
We now inspect equation \eqref{eq:1singularityterms} for all the other possible combinations of $k_1$ and $k_2$. For each case, we replace the matrix product $F_{\alpha,k_1}(w)F_{1,N}(w)$ with either equation \eqref{eq:f1nreplacement1} or \eqref{eq:f1nreplacement2}. When $k_1=1$ and $k_2 = 2$ we obtain,
\begin{multline}
    \frac{r_{\alpha,1}(w)^{-\nhalf-m'} }{(w-1)^N}F_{\alpha,1}(w)F_{1,N}(w) = \frac{r_{2,N}(w)r_{\alpha,1}(w)^{-m'} r_{\beta,1}(w)^\nhalf}{(w-1)^N} F_{\alpha,1}(w)F_{\beta,1}(w) \\
    + \frac{r_{2,N}(w)r_{\alpha,1}(w)^{-m'} r_{\beta,2}(w)^\nhalf}{(w-1)^N} F_{\alpha,1}(w)F_{\beta,2}(w) - \frac{r_{2,N}(w)^2r_{\alpha,1}(w)^{-\nhalf-m'}}{(w-1)^N}F_{\alpha,1}(w)F_{2,N}(w)
\end{multline}
All of the above terms have no pole at $w=1$. Thus,
\begin{equation*}
    \oint_{\gamma_1} \dw \oint_{\gamma_{0,1}} \frac{\dz}{z(z-w)} \frac{w^{\xi'+N}(z-1)^{N}}{z^{\xi+N}(w-1)^N} r_{\alpha,1}(w)^{-\nhalf-m'} F_{\alpha,1}(w)F_{1,N}(w)F_{\alpha,2}(w)\Phi_\eps(z)^{m-\frac{N}{2}} = 0
\end{equation*}
For $k_1=2$ and $k_2=1$ the product in \eqref{eq:1singularityterms} is, 
\begin{multline}
    \frac{r_{\alpha,1}(w)^{\nhalf+m'} }{(w-1)^N}F_{\alpha,2}(w)F_{1,N}(w) = \frac{r_{2,N}(w)r_{\alpha,1}(w)^{m'} r_{\beta,1}(w)^\nhalf}{(w-1)^N} F_{\alpha,2}(w)F_{\beta,1}(w) \\
    + \frac{r_{2,N}(w)r_{\alpha,1}(w)^{m'} r_{\beta,2}(w)^\nhalf}{(w-1)^N} F_{\alpha,2}(w)F_{\beta,1}(w) - \frac{r_{2,N}(w)^2r_{\alpha,1}(w)^{\nhalf+m'}}{(w-1)^N}F_{\alpha,2}(w)F_{2,N}(w)
\end{multline}
In this case, the second and third terms have no pole at $w=1$, so only the first term contribution non-trivially to the contour integral. Lastly we check equation \eqref{eq:1singularityterms} for the case $k_1=k_2=2$,
\begin{multline}
    \frac{r_{\alpha,2}(w)^{\nhalf-m'} }{(w-1)^N}F_{\alpha,2}(w)F_{1,N}(w) = \frac{r_{2,N}(w)r_{\alpha,2}(w)^{N-m'} r_{\beta,2}(w)^\nhalf}{(w-1)^N} F_{\alpha,2}(w)F_{\beta,1}(w) \\
    + \frac{r_{2,N}(w)r_{\alpha,2}(w)^{N-m'} r_{\beta,2}(w)^\nhalf}{(w-1)^N} F_{\alpha,2}(w)F_{\beta,1}(w) - \frac{r_{2,N}(w)^2r_{\alpha,2}(w)^{\nhalf-m'}}{(w-1)^N}F_{\alpha,2}(w)F_{2,N}(w)
\end{multline}
All three terms have no singularity at $w=1$, so the resulting double contour integral is trivial. This leaves us with the simplified kernel,
\begin{multline}
    \Big[\mathbb{K}_N(4m',2\xi'+j;4m,2\xi+i)\Big]_{i,j=0}^1 = -\frac{\mathbb{I}_{m>m'}}{2 \pi \im} \oint_{\gamma_{0,1}} \frac{\dz}{z} z^{\xi'-\xi} \Phi_{\alpha}(z)^{\frac{N}{2}-m'}\Phi_{\eps}(z)^{m-\frac{N}{2}} \\
    + \frac{1}{(2\pi\im)^2} \oint_{\gamma_1} \dw \oint_{\gamma_{0,1}} \frac{\dz}{z(z-w)} \frac{w^{\xi'+N}(z-1)^{N}}{z^{\xi+N}(w-1)^N} r_{2,N}(w)r_{\alpha,1}(w)^{N-m'}r_{\beta,1}(w)^\nhalf \\
    \times F_{\alpha,1}(w)F_{\beta,1}(w)F_{\alpha,1}(w)\Phi_\eps(z)^{m-\frac{N}{2}} \\
    + \frac{1}{(2\pi\im)^2} \oint_{\gamma_1} \dw \oint_{\gamma_{0,1}} \frac{\dz}{z(z-w)} \frac{w^{\xi'+N}(z-1)^{N}}{z^{\xi+N}(w-1)^N} r_{2,N}(w)r_{\alpha,1}(w)^{m'} r_{\beta,1}(w)^\nhalf  \\
    \times F_{\alpha,2}(w)F_{\beta,1}(w)F_{\alpha,1}(w)\Phi_\eps(z)^{m-\frac{N}{2}}
\end{multline}
Next, we use a similar procedure to replace the term $r_{2,N}(w)$ in the double contour integrals. We need an expansion of $r_{2,N}(w)$ in terms of the eigenvalues $r_{\alpha,2}(w)$ and $r_{\beta,2}(w)$, which we state below.  
\begin{lemma} \label{lem:r2Nexpansion}
    Let $\delta = \max\{\alpha, \beta\}$ and let $\mathcal{B} = (-\infty, -\delta^{-2}] \cup [-\delta^2,0]$.\footnote{We let $\mathcal{B}$ denote the larger of the possible branch cuts.} For $z \in \mathbb{C}\backslash \mathcal{B}$ and $N$ sufficiently large the following expansion converges,
    \begin{equation} \label{eq:r2Nexp}
        r_{2,N}(z) = \sum_{p=0}^\infty \sum_{q=0}^\infty c_{p,q}(z)r_{\alpha,2}(z)^{\frac{(2p+1)N}{2}} r_{\beta,2}(z)^{\frac{(2q+1)N}{2}}
    \end{equation}
    Where the coefficients $c_{p,q}(z)$ are non-singular and non-zero at $z=1$.
\end{lemma}
The proof of Lemma \ref{lem:r2Nexpansion} is in Section \ref{section:analysisofeigen}. It is important to check that the coefficients of the above series do not affect the pole at $w=1$, but besides that, the exact statement of the coefficients is mostly unimportant. The only coefficient we need, and thus state explicitly, is
\begin{equation}
    c_{0,0}(w) = \frac{2}{1+2g_{\alpha,\beta}(w)}
\end{equation}
When we plug this expansion into the double contour integrals, most of the terms evaluate to zero due to the residue theorem. To see this, first consider the product
\begin{equation*}
    \frac{1}{(w-1)^N}f(w)r_{\alpha,1}(w)^{N-m'}r_{\beta,1}(w)^\nhalf 
\end{equation*}
This has no pole at $w = 1$ as long as $f(w)$ has a zero at $w=1$ that is at least of order $4N$. Similarly, the product
\begin{equation*}
    \frac{1}{(w-1)^N}f(w)r_{\alpha,1}(w)^{m'}r_{\beta,1}(w)^\nhalf
\end{equation*}
has no pole at $w=1$ as long as $f(w)$ has a zero at $w=1$ of at least order $3N$. This leaves us with the following simplification of the correlation kernel,
\begin{multline}
    \Big[\mathbb{K}_N(4m',2\xi'+j;4m,2\xi+i)\Big]_{i,j=0}^1 = -\frac{\mathbb{I}_{m>m'}}{2 \pi \im} \oint_{\gamma_{0,1}} \frac{\dz}{z} z^{\xi'-\xi} \Phi_{\alpha}(z)^{\frac{N}{2}-m'}\Phi_{\eps}(z)^{m-\frac{N}{2}} \\
    + \frac{1}{(2\pi\im)^2} \oint_{\gamma_1} \dw \oint_{\gamma_{0,1}} \frac{\dz}{z(z-w)} \frac{w^{\xi'+N}(z-1)^{N}}{z^{\xi+N}(w-1)^N} r_{\alpha,1}(w)^{\nhalf-m'} \frac{2F_{\alpha,1}(w)F_{\beta,1}(w)F_{\alpha,1}(w)}{1+2g(w)}\Phi_\eps(z)^{m-\frac{N}{2}} \\
    + \frac{1}{(2\pi\im)^2}\oint_{\gamma_1} \dw \oint_{\gamma_{0,1}} \frac{\dz}{z(z-w)} \frac{w^{\xi'+N}(z-1)^{N}}{z^{\xi+N}(w-1)^N} r_{\alpha,2}(w)^{\nhalf-m'} \frac{2 F_{\alpha,2}(w)F_{\beta,1}(w)F_{\alpha,1}(w)}{1+2g_{\alpha,\beta}(w)}\Phi_\eps(z)^{m-\frac{N}{2}}.
\end{multline}
To obtain the final form we first notice that 
\begin{equation}
    \frac{2F_{\alpha,1}(w)F_{\beta,1}(w)F_{\alpha,1}(w)}{1+2g(w)} = F_{\alpha,1}(w).
\end{equation}
One can check the above equality using the explicit formulas given in equations \eqref{eq:f1}, \eqref{eq:f2}, \eqref{eq:g}. 
\begin{remark}
    We are not surprised that an equality like this. Since the $F$-matrices are rank one, we expect the product $F_{\alpha,1}(w)F_{\beta,1}(w)F_{\alpha,1}(w)$, to be a scalar multiple of $F_{\alpha,1}(w)$. Why does our integral produce the ``right'' constant for simplification? I do not have a satisfactory answer to this question. 
\end{remark}
%%%%%%%%%%%%%%%%%%%%%%%%%%%%%%%%%%%%%%%%%%%%%
\subsection{Proof of Theorem \ref{theorem:kernel} when $\nicefrac{N}{2}< m' < N$} \label{section:m>N/2}
To prove Theorem \ref{theorem:kernel} for $m' > \nicefrac{N}{2}$, we follow a similar procedure to the case when $m'<\nicefrac{N}{2}$. We first expand the double contour integral using eigen-decompositions of $\Phi_\alpha(w)$ and $\Phi_\beta(w)$. Then we simplify the $w$-integral by replacing the matrix products $F_{\beta,1}(w)F_{1,N}(w)$ and $F_{\beta,2}(w)F_{1,N}(w)$ with the expressions given in equations \eqref{eq:f1nreplacement1} and \eqref{eq:f1nreplacement2}. Lastly, we use Lemma \ref{lem:r2Nexpansion} to replace the term $r_{2,N}(w)$. The differences are in the details, which we will highlight below. First, we state the expanded version of the kernel,
\begin{multline} \label{eq:kernelexpandedmess}
    \Big[\mathbb{K}_N(4m',2\xi'+j;4m,2\xi+i)\Big]_{i,j=0}^1 = -\frac{\mathbb{I}_{m>m'}}{2 \pi \im} \oint_{\gamma_{0,1}} \frac{\dz}{z} z^{\xi'-\xi} \Phi_{\alpha}(z)^{\frac{N}{2}-m'}\Phi_{\eps}(z)^{m-\frac{N}{2}} \\
    + \frac{1}{(2\pi\im)^2} \sum_{k_1=1,2} \sum_{k_2=1,2} \sum_{k_3=1,2}\oint_{\gamma_1} \oint_{\gamma_{0,1}} \frac{\dz\dw}{z(z-w)} \frac{w^{\xi'+N}(z-1)^{N}}{z^{\xi+N}(w-1)^N} r_{\beta,k_1}(w)^{\nhalf-m'}r_{\alpha,k_2}(w)^{-\nhalf}r_{\alpha,k_3}(w)^{\frac{N}{2}}\\
    \times F_{\beta,k_1}(w)F_{\alpha,k_2}(w)F_{1,N}(w)F_{\alpha,k_3}(w)\Phi_\eps(z)^{m-\frac{N}{2}}
\end{multline}
Again our goal will be to track count the order of the pole/zero at $w = 1$. The term 
$$\frac{1}{(w-1)^N}r_{\beta,k_1}(w)^{\nhalf-m'}r_{\alpha,k_2}(w)^{-\nhalf}r_{\alpha,k_3}(w)^{\frac{N}{2}}$$
considers all the parts of the integral, besides $F_{1,N}(w)$, that potentially contribute to the pole at $w=1$. Depending on the value of $k_1$, $k_2$, and $k_3$ the order of the pole (or zero) differs. We first track the highest possible pole order (or lowest possible zero order) when $k_2 = 1$,
\begingroup
    \begin{center}
    \renewcommand{\arraystretch}{1.5}
    \begin{tabular}{c|c|c}
        & $k_1=1$ & $k_1=2$ \\
        \hline
        $k_3=1$ & pole of order $< N$ & pole of order $<2N$ \\
        \hline
        $k_3=2$ & zero of order $>N$  &  zero of order $>0$ \\
    \end{tabular}
    \end{center}
\endgroup
\noindent Next we do the same procedure, assuming $k_2 = 2$,
\begingroup
    \begin{center}
    \renewcommand{\arraystretch}{1.5}
    \begin{tabular}{c|c|c}
        & $k_1=1$ & $k_1=2$ \\
        \hline
        $k_3=1$ & pole of order $< 3N$ & pole of order $<4N$ \\
        \hline
        $k_3=2$ & pole of order $<N$  &  pole of order $< 2N$ \\
    \end{tabular}
    \end{center}
\endgroup
We are now ready to replace the terms $F_{\alpha,k_2}(w)F_{1,N}(w)$ by using either equation \eqref{eq:f1nreplacement1} or equation \eqref{eq:f1nreplacement2}. We are careful to observe which terms actually have a pole at $w=1$. As was the case in Section \ref{section:mlessN/2proof}, many of the terms drop out of the contour integral as a result of the residue theorem. We simplify the original double contour integrals down to four integrals,
\begin{multline} 
    \sum_{k_1=1,2} \sum_{k_2=1,2} \sum_{k_3=1,2} \oint_{\gamma_1} \oint_{\gamma_{0,1}} \frac{\dz\dw}{z(z-w)} \frac{w^{\xi'+N}(z-1)^{N}}{z^{\xi+N}(w-1)^N} r_{\beta,k_1}(w)^{\nhalf-m'}r_{\alpha,k_2}(w)^{-\nhalf}r_{\alpha,k_3}(w)^{\frac{N}{2}}\\
    \times F_{\beta,k_1}(w)F_{\alpha,k_2}(w)F_{1,N}(w)F_{\alpha,k_3}(z)\Phi_\eps(z)^{m-\frac{N}{2}} \\
    = \sum_{k=1,2} \oint_{\gamma_1} \oint_{\gamma_{0,1}} \frac{\dz\dw}{z(z-w)} \frac{w^{\xi'+N}(z-1)^{N}}{z^{\xi+N}(w-1)^N} r_{\beta,k}(w)^{\nhalf-m'}r_{2,N}(w)r_{\alpha,1}(w)^\nhalf r_{\beta,1}(w)^\nhalf\\
    \times F_{\beta,k}(w)F_{\alpha,1}(w)F_{\beta,1}(w)F_{\alpha,1}(w)\Phi_\eps(z)^{m-\frac{N}{2}} \\
    + \sum_{k=1,2}\oint_{\gamma_1} \oint_{\gamma_{0,1}} \frac{\dz\dw}{z(z-w)} \frac{w^{\xi'+N}(z-1)^{N}}{z^{\xi+N}(w-1)^N} r_{\beta,k}(w)^{\nhalf-m'}r_{\alpha,1}(w)^\nhalf r_{2,N}(w)r_{\beta,1}(w)^\nhalf\\
    \times F_{\beta,k}(w)F_{\alpha,2}(w)F_{\beta,1}(w)F_{\alpha,1}(w)\Phi_\eps(z)^{m-\frac{N}{2}} \\
\end{multline}
The first sum of integrals and second sum of integrals can be combined, by matching up the integrals with the same value of $k$. In particular, we get 
\begin{align*}
    F_{\beta,k}(w)F_{\alpha,1}(w)F_{\beta,1}(w)F_{\alpha,1}(w)&\Phi_\eps(z)^{m-\frac{N}{2}} + F_{\beta,k}(w)F_{\alpha,2}(w)F_{\beta,1}(w)F_{\alpha,1}(w)\Phi_\eps(z)^{m-\frac{N}{2}} \\
    &= F_{\beta,k}(w)\Big(F_{\alpha,1}(w) + F_{\alpha,2}(w)\Big)F_{\beta,1}(w)F_{\alpha,1}(w)\Phi_\eps(z)^{m-\frac{N}{2}} \\
    &= F_{\beta,k}(w)F_{\beta,1}(w)F_{\alpha,1}(w)\Phi_\eps(z)^{m-\frac{N}{2}}
\end{align*}
This simplification show that when $k = 2$, the matrix product is zero due to the orthogonality of the $F_\beta$-matrices. Thus we are left with only one double contour integral, 
\begin{equation} \label{eq:intkernelbside}
    \oint_{\gamma_1} \oint_{\gamma_{0,1}} \frac{\dz\dw}{z(z-w)} \frac{w^{\xi'+N}(z-1)^{N}}{z^{\xi+N}(w-1)^N} r_{\beta,1}(w)^{N-m'}r_{2,N}(w)r_{\alpha,1}(w)^\nhalf F_{\beta,1}(w)F_{\alpha,1}(w)\Phi_\eps(z)^{m-\frac{N}{2}} \\
\end{equation}
We are now ready to apply Lemma \ref{lem:r2Nexpansion}. Notice that the product $(w-1)^{-N}r_{\beta,1}(w)^{N-m'}r_{\alpha,1}(w)^\nhalf$ has a pole at $w=1$ of order less than $3N$. This means that only the first term of the expansion in Lemma \ref{lem:r2Nexpansion} contributes non-trivially to the integral. The double contour integral in \eqref{eq:intkernelbside} then becomes, 
\begin{equation}
    \oint_{\gamma_1} \oint_{\gamma_{0,1}} \frac{\dz\dw}{z(z-w)} \frac{w^{\xi'+N}(z-1)^{N}}{z^{\xi+N}(w-1)^N} r_{\beta,1}(w)^{\nhalf-m'}\frac{2}{1+2g_{\alpha,\beta}(w)} F_{\beta,1}(w)F_{\alpha,1}(w)\Phi_\eps(z)^{m-\frac{N}{2}} 
\end{equation}
We now have the following form of the correlation kernel,
\begin{multline} \label{eq:bsidekernelinitial}
    \Big[\mathbb{K}_N(4m',2\xi'+j;4m,2\xi+i)\Big]_{i,j=0}^1 = -\frac{\mathbb{I}_{m<m'}}{2 \pi \imt} \oint_{\gamma_{0,1}} \frac{\dz}{z} z^{\xi-\xi'} \Phi_{\beta}(z)^{\nhalf-m'}\Phi_{\eps}(z)^{m-\nhalf} \\
    + \frac{1}{(2\pi\imt)^2} \oint_{\gamma_1} \dw \oint_{\gamma_{0,1}} \frac{\dz}{z(z-w)} \frac{w^{\xi'+N}(z-1)^{N}}{z^{\xi+N}(w-1)^N} r_{\beta,1}(w)^{\nhalf-m'}\frac{2F_{\beta,1}(w)F_{\alpha,1}(w)}{1+2g_{\alpha,\beta}(w)}  \Phi_\eps(z)^{m-\nhalf}
\end{multline}
for $m' > \nicefrac{N}{2}$. This is a valid kernel that is amenable to asymptotic analysis, however we wish to make one additional change to the kernel so that it can better be compared to the case of $m' \leq \nicefrac{N}{2}$. To do so, we use the following identity, 
\begin{equation*}
    \frac{2}{1+2g_{\alpha,\beta}(w)} F_{\beta,1}(w)F_{\alpha,1}(w) = F_{\beta,1}(w) + \frac{2}{1+2g_{\alpha,\beta}(w)} F_{\beta,1}(w)F_{\alpha,1}(w)F_{\beta,2}(w).
\end{equation*}
To prove this identity one can multiply the left hand side on the right by
$$I = F_{\beta,1}(w) + F_{\beta,2}(w).$$
The result of this substitution is the kernel presented in \eqref{eq:bsidekernelfinal}.
%%%%%%%%%%%%%%%%%%%%%%%%%%%%%%%%%%%%%%%%%%%%%%%%%%%%
%%%%%%%%%%%%%%%%%%%%%%%%%%%%%%%%%%%%%%%%%%%%%%%%%%%%%%%%%%%%%%%%%%%%%%%%%%%%%%%%%%%%%%%%%%%%%%%%%%%%%%%%%%%%%%%%%%%%%%%%%%%%%%%%%%%%%%%%%%%%%%%%%%%%
\section{Local Asymptotics of the Model} \label{section:asymptoticresults}
The goal of this section is to prove the local asymptotic results stated in Propositions \ref{prop:I22decay} and \ref{prop:I21decay}. We begin in Section \ref{section:saddles} by stating some necessary results about the behavior of the saddle functions. In particular, we need precise statements about the location of saddle points and the contours of steepest ascent and descent in the smooth and rough regions of the model. Once we have the necessary preliminary information, we prove Propositions \ref{prop:I22decay} and \ref{prop:I21decay} in Sections \ref{section:I22proofs} and \ref{section:I21proofs}, respectively.

%%%%%%%%%%%%%%%%%%%%%%%%%%%%%%%%%%%%%%%%%%%%%%%%%%%%%%%%%%%%%%%%%%%%%%
\subsection{Preliminaries of Saddle Functions} \label{section:saddles}
Recall the piece-wise definition of the saddle function $\psi_k(z)$,
\begin{equation}
    \psi_k(z;x,y)=\begin{cases}
        \psi_{\alpha,k}(z; x,y) & \text{for } x < \nicefrac{1}{2} \\
        \psi_{\beta,k}(z; x,y) & \text{for } x > \nicefrac{1}{2}
    \end{cases}
\end{equation}
where $\psi_{\eps,k}(z)$ is restated below. 
$$\psi_{\eps,k}(z) = (1+y)\log z - \log(z-1) + \left(\frac{1}{2}-x\right)\log r_{\eps,k}(z)$$
We start by stating a basic lemma regarding the asymptotic behavior of the function $\oRe \psi_k(z)$. 
\begin{lemma} \label{lem:realpsiasymptotics}
    As $z \to \infty$, the saddle functions exhibit the following asymptotic behavior, 
    $$\oRe \psi_k(z) = y \log |z| + O(1)$$
\end{lemma}
\begin{proof}
    One can observe from the definition of $r_{\eps,1}(z)$ and $r_{\eps,2}(z)$ that 
    $$\lim_{z\to\infty} r_{\eps,k}(z) = 1$$
    for $k = 1,2$. Also, as $z \to \infty$ we have,
    $$(y+1)\log|z| - \log|z-1| = y \log|z| + O(z^{-1})$$
\end{proof}
A direct consequence of this lemma is that for any $(x,y) \in (0,1) \times (-1,0)$, $\oRe \psi_k(z) \to -\infty$ as $z \to \infty$. \\

We know that $\psi_k(z)$ always has at least one saddle point between the branch cuts on the negative real axis. We denote this saddle point $z_k^*$.\footnote{This is the same notation used in \cite{DK21}.} We start by stating, 
\begin{lemma} \label{lem:psimax}
    For any $(x,y) \in (0,\nicefrac{1}{2}) \times (-1,0)$, the following is true regarding the saddle point $z_k^*$
    \begin{enumerate}[label=(\roman*)]
        \item $z_1^*$ is a local maximum of $\psi_1(z)$,
        \item if $(x,y) \in \mathfrak{S}$, then $z_2^*$ is a local maximum of $\psi_2(z)$, otherwise it is a local minimum. 
    \end{enumerate}
\end{lemma}
This is largely a subset of Lemma 6.8 in \cite{DK21}. There are a few key differences we take the time to acknowledge and justify. Firstly, since there is an overall sign change between their saddle function and the saddle function presented here, there is a reversal of maxima and minima. Additionally, while the asymptotic coordinate is restricted to the smooth region in the proof presented by Duits and Kuijlaars, we address cases where this is not necessarily true. We present the proof of these additional part below. 

\begin{proof}
Assume that $(x,y) \in \mathfrak{F} \cup \mathfrak{R}$ in addition to the fact that $x < \nicefrac{1}{2}$. In this case, there is exactly one saddle point of $\psi_k(z)$ on the interval $(-\alpha^{-2},-\alpha^2)$. We consider the behavior of $\oRe \psi'_k(z)$ as we approach the branch points. In particular, we have
\begin{equation}
    \lim_{z \to (-\alpha^2)^-} \frac{r'_{\alpha,1}(z)}{r_{\alpha,1}(z)} = -\infty \qquad \lim_{z \to (-\alpha^{-2})^+} \frac{r'_{\alpha,1}(z)}{r_{\alpha,1}(z)} = +\infty
\end{equation}
\begin{equation}
    \lim_{z \to (-\alpha^2)^-} \frac{r'_{\alpha,2}(z)}{r_{\alpha,2}(z)} = +\infty \qquad \lim_{z \to (-\alpha^{-2})^+} \frac{r'_{\alpha,2}(z)}{r_{\alpha,2}(z)} = -\infty
\end{equation}
From this we deduce that $z_2^*$ is a local minimum of $\oRe \psi_2(z)$ and $z_1^*$ is a local maximum of $\oRe \psi_1(z)$.
\end{proof}

Since we are restricting our analysis to $x < \nicefrac{1}{2}$ we can make some additional remarks regarding the general location of the saddle points. Under this restriction, $\psi_1(z)$ always has a single saddle point, while $\psi_2(z)$ has three saddle points (counting multiplicity). This is mostly a result of Lemma \ref{lem:psimax} along with that $\Psi'(z) \, dz$ is a single-valued meromorphic function (for fixed $(x,y)$) on the Riemann surface $\mathcal{R}$ described before Definition \ref{def:macroscopicboundary}. We will refer to the other two saddle points of $\psi_2(z)$ as $z_3^*$ and $z_4^*$. When $(x,y) \in \mathfrak{S}$, these additional saddle points also lie on the interval $(-\alpha^{-2},-\alpha^2)$. We adopt the convention, $z_3^* < z_2^* < z_4^*$. The points $z_3^*$ and $z_4^*$ are both local minimums of $\oRe \psi_2(z)$.

The \textit{contour of steepest ascent} from $z_k^*$ is the contour through $z_k^*$ where $\oIm \psi_k(z)$ is held constant and $\oRe \psi_k(z)$ increases the fastest as we move away from the saddle point. Since $z_1^*$ is a maximum of $\oRe \psi_1(z)$, the path of steepest ascent is perpendicular to the real axis at $z_1^*$ (the same is true for $z_2^*$ when the asymptotic coordinate is in the smooth region). We define the following regions of the complex plane,\footnote{While this notation is the same as in \cite{DK21} the regions will be reversed in our notation due to the change in sign of the saddle functions presented here.} 
\begin{equation}
    \Omega^+_k = \{ z \in \mathbb{C} : \oRe\psi_k(z) > \oRe\psi_k(z_k^*) \} 
\end{equation}
\begin{equation}
    \Omega^-_k = \{ z \in \mathbb{C} : \oRe\psi_k(z) < \oRe\psi_k(z_k^*) \} 
\end{equation}
Now we state the following lemma, 
\begin{lemma} \label{lem:psi2contours}
    For any $(x,y) \in (0,\nicefrac{1}{2}) \times (-1,0)$ the following holds, 
    \begin{enumerate}[label=(\roman*)]
        \item The steepest ascent contour of $\psi_1(z;x,y)$ through $z_1^*$, denoted $\gamma_{s,1}$, is a simple closed curve that contains the interval $[-\alpha^2,0]$ and intersects the positive real line at $z = 1$.
        \item $\Omega^+_1$ is a bounded set containing a single component which contains $\gamma_{s,1} \backslash \{z_1^*\}$.
        \item $\Omega^-_1$ is an open set with an unbounded component that contains a contour going around $(-\infty,-\alpha^{-2}]$ and with a bounded component that contains a contour around $[-\alpha^2,0]$.
    \end{enumerate}
    Under the additional condition that $(x,y) \in \mathfrak{S}$, the following also holds
    \begin{enumerate}[label=(\roman*)]
        \item The steepest ascent contour of $\psi_2(z;x,y)$ through $z_2^*$, denoted $\gamma_{s,2}$, is a simple closed curve that contains the interval $[-\alpha^2,0]$ and intersects the positive real line at $z = 1$.
        \item $\Omega^+_2$ is a bounded set containing at most three components that are a positive distance apart. One of which contains $\gamma_{s,2} \backslash \{z_2^*\}$.
        \item $\Omega^-_2$ is an open set with an unbounded component that contains a contour going around $(-\infty,-\alpha^{-2}]$ and with a bounded component that contains a contour around $[-\alpha^2,0]$.
    \end{enumerate}
\end{lemma}
Figure \ref{fig:smoothcontours} shows an example of the regions $\Omega_1^+$ and $\Omega_1^-$ along with the steepest ascent contour. This lemma is a slightly modified version of Lemma 6.9 of \cite{DK21}. The reader should refer to their text for a the proof. One can check that the proof is valid for $\psi_1(z)$ for any $x < \nicefrac{1}{2}$. This is because $z_1^*$ is always a local maximum of $\oRe \psi_1(z)$. In the first part, we conclude that $\Omega^+_1$ is a single component, because there is exactly one saddle point on the interval $(-\alpha^{-2},-\alpha^2)$. The case of three components is only possible if there are three saddle points on the interval $(-\alpha^{-2},-\alpha^2)$.

In our asymptotic proofs we will use a statement a bit stronger than (both) statement (iii) above. In particular, we state the following corollary, 
\begin{corollary} \label{cor:contoursstronger}
    Let $r_1$ and $r_2$ be negative real numbers such that $r_1 < z_1^* < r_2$. Then there exists contours surrounding the intervals $(-\infty, r_1]$ and $[r_2,0]$ that are contained entirely inside of $\Omega_j^-$ for $j = 1, \, 2$.
\end{corollary}
\begin{proof}
    In the case where $\Omega_j^+$ is only one component, the statement immediately follows from the fact that $z_1^*$ is the only point on the negative not a positive distance away from $\Omega^+_j$. In the case where $\Omega_j^+$ has three components, we must also invoke the fact that the components are all a positive distance away from each other. 
\end{proof}

\begin{figure}[ht]
    \centering
    \includegraphics[scale=0.7]{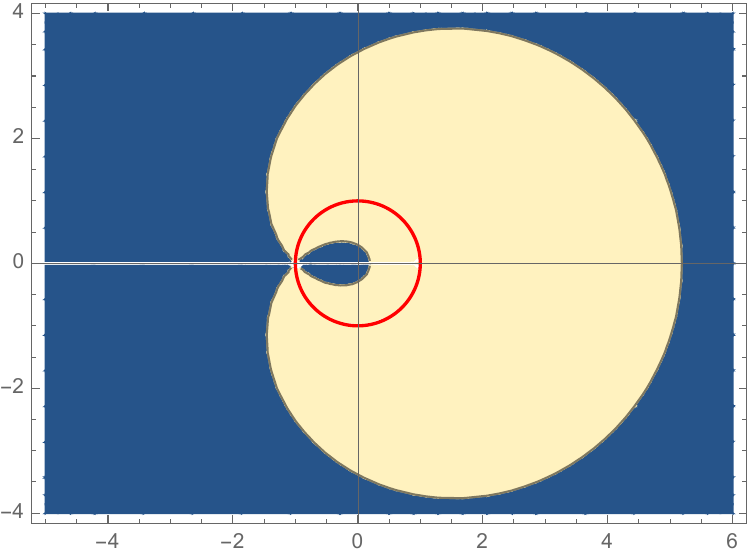}
    \caption{An example of the steepest ascent contour of $\psi_1(z;x,y)$ through $z_1^*$ and the regions $\Omega^+_1$ and $\Omega^-_1$. The steepest ascent contour is shown in red, $\Omega^+_1$ is the region in yellow, and $\Omega^-_1$ is the region in blue. Contour plot was generated with Mathematica.}
    \label{fig:smoothcontours}
\end{figure}

The above information is enough for the analysis of $I_{2,1}^\alpha$ and for the analysis of $I^\alpha_{2,2}$ in the smooth region. In order to analyze $I^\alpha_{2,2}$ in the rough region, we will need some information on the contours of steepest descent and ascent in the rough region. For $x \in (0,\nicefrac{1}{2})$ and $(x,y) \in \mathfrak{R}$, the additional saddle points of $\psi_2(z)$, $z_3^*$ and $z_4^*$, are complex. In particular, $z_3^* = \overline{z_4^*}$. We will let $z_3^*$ denote the complex saddle point with positive imaginary part. We define the following regions of the complex plane, 
\begin{equation}
    \Pi^+ = \{ z \in \mathbb{C} : \oRe\psi_2(z) > \oRe\psi_2(z_3^*) \} 
\end{equation}
\begin{equation}
    \Pi^- = \{ z \in \mathbb{C} : \oRe\psi_2(z) < \oRe\psi_2(z_3^*) \} 
\end{equation}
It suffices to define these regions just using $z_3^*$, since $\oRe \psi_2(z_3^*) = \oRe \psi_2(z_4^*)$. We now describe some important aspects about these regions, 
\begin{lemma} \label{lem:roughsaddles}
    Let $(x,y) \in \mathfrak{R}$ and $x \in (0,\nicefrac{1}{2})$, then
    \begin{enumerate}[label=(\roman*)]
        \item $\Pi^+$ contains at least two and at most three bounded components. One component will always contain $z = 1$, while the other guaranteed component contains at least one of the branch points, $-\alpha^2$ or $-\alpha^{-2}$. The third component, if it exists, must contain the other branch point, intersect the branch cut, and be positive distance away from the other components. 
        \item $\Pi^-$ contains a bounded component, which contains $z=0$, and an unbounded component.
    \end{enumerate}
\end{lemma}
\begin{proof}
    We will focus only on the upper half plane, since the regions $\Pi^\pm$ are symmetric about the real axis. The boundary between these regions are contours where $\oRe \psi_2(z) = \oRe \psi_2(z_3^*)$. These contours remain bounded, since $\oRe \psi_2(z) \to -\infty$ as $z \to \infty$. Additionally, these contours do not intersect anywhere in the upper half plane or on the real axis, because this would violate the maximum principle of harmonic functions. This means that the four half-paths emanating from $z_3^*$ intersect the real axis at four different points. To see the possible intersection points, we note the following limits,
    \begin{itemize}
        \item $\lim_{z \to \infty} \oRe \psi_2(z)  = -\infty$
        \item $\lim_{z \to 0} \oRe \psi_2(z)  = -\infty$
        \item $\lim_{z \to 1} \oRe \psi_2(z)  = +\infty$
    \end{itemize}
    We also have the following behavior on intervals, 
    \begin{itemize}
        \item $\oRe\psi_2(z)$ is strictly decreasing on $(1,\infty)$
        \item $\oRe\psi_2(z)$ is strictly increasing on $(0,1)$
        \item $\oRe\psi_2(z)$ is strictly increasing on $(-\infty,-\alpha^{-2})$
        \item $\oRe\psi_2(z)$ is strictly decreasing on $(-\alpha^2,z_1^*)$
        \item $\oRe\psi_2(z)$ is strictly increasing on $(z_1^*,-\alpha^2)$
        \item $\oRe\psi_2(z)$ is strictly decreasing on $(-\alpha^2,0)$
    \end{itemize}
    Note that $\oRe \psi_1(z) = \oRe \psi_2$ on the branch cuts, so the choice of analytic continuation does not matter. This means that one half-path emanating from $z_3^*$ intersects the real axis at some value greater than one and one half-path intersects the real axis on the interval $(0,1)$. The other two half paths should intersect the negative real axis. It is clear from the limits that $\Pi^+$ should contain the point $z=1$, $\Pi^-$ should contain the point $z=0$, and that $\Pi^-$ should contain an unbounded component. 

    Since $\oRe\psi_2(z)$ is not strictly increasing on the interval $(-\infty,0)$, where the two half-paths intersect this interval can vary. Either $-\alpha^2$ or $-\alpha^{-2}$ is a global maximum of $\oRe \psi_2(z)$ on $(-\infty,0)$, so at least one of them (possibly both) must be contained in $\Pi^+$. Alternatively, one of the branch points may be contained in a separate third component that intersect the branch cut.\footnote{It is unclear whether this actually occurs, but it does not affect our analysis either way.} There can only be one additional component since any additional components must intersect the branch cuts and $\oRe \psi_2(x)$ is monotonic along each cut. Similarly, there is no additional component of $\Pi^-$ because it would have to be away from the branch cuts and thus it would break the maximum/minimum principle of harmonic functions.
\end{proof}
\begin{figure}[ht]
    \centering
    \includegraphics[scale=0.5]{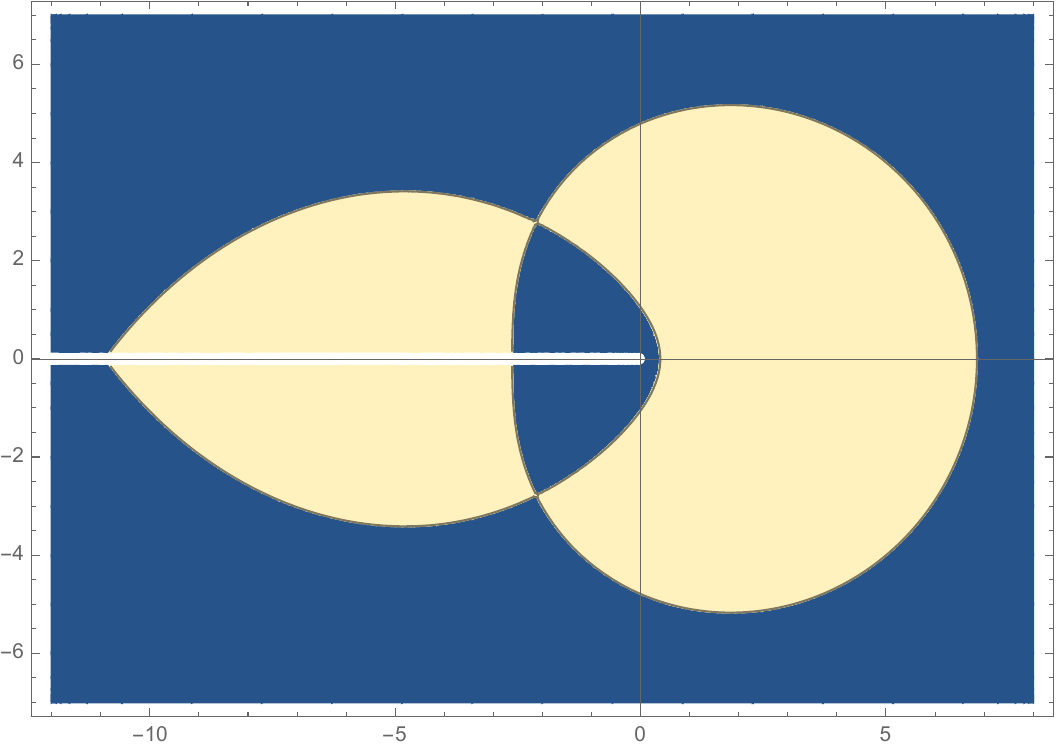}
    \caption{An example of the regions $\Pi^\pm$ for $(x,y) \in \mathfrak{R}$, $\Pi^-$ is the region in blue while $\Pi^+$ is the region in yellow. Contour plot was generated in Mathematica.}
    \label{fig:roughregions}
\end{figure}
An example of the regions $\Pi^\pm$ are shown in Figure \ref{fig:roughregions}. Since we know the contours of steepest ascent and descent through $z_3^*$ and $z_4^*$ should lie entirely in $\Pi^+ \backslash \{z_3^*,z_4^*\}$ and $\Pi^- \backslash \{z_3^*,z_4^*\}$, respectively, understanding these regions gives us enough information to understand the contours. We let $\gamma_{R,A}$ denote the contour of steepest ascent and we let $\gamma_{R,D}$ denote the contour of steepest descent. An example of these contours is given in Figure \ref{fig:roughcontours}. 
\begin{figure}[ht]
    \centering
    \begin{tikzpicture}
    \draw (0,-2) -- (0,2);
    \draw (-2,0) -- (2,0);
    \clip (-3,-2) rectangle (3,2);
    \draw[ultra thick] (-3,0) -- (-2,0);
    \draw[ultra thick] (-0.5,0) -- (0,0);
    \draw[ultra thick, blue] (0,0) ellipse (1.25 and 1);
    \draw[- stealth, ultra thick, red] (-1.6,1.72) -- (-1.4,1.45);
    \draw[ultra thick, red] plot [smooth] coordinates {(-2,2.3) (-1,0.9) (-0.5,0.4) (0.1,0) (-0.5,-0.4) (-1,-0.9) (-2,-2.3)};
    \draw[- stealth, ultra thick, blue] (0.1,1) -- (-0.1,1);
    \end{tikzpicture}
    \caption{Steepest descent and ascent contours for $\psi_2(z)$ in the rough region. The point where the contours cross are the two complex saddle points. The red contour is the path of steepest descent, while the blue contour is the path of steepest ascent.}
    \label{fig:roughcontours}
\end{figure}

%%%%%%%%%%%%%%%%%%%%%%%%%%%%%%%%%%%%%%%%%%%%%%%%%%%%%%%%%%%%%%%%%%%%%%%%%%%%%%%%%%%%%%%%%%%%%%%%%%%%%%%%%%%%%%%%%
\subsection{Initial Deformation of the Contours of $I_{2,k}^\alpha$} \label{section:contours}
Before we detail the main asymptotic analysis of the contour integrals $I^\alpha_{2,k}$, we will first deform the $w$-contours of the integral from whats presented in equation \eqref{eq:I2kdef}. We first push the $w$-contour, $\gamma_1$, out to infinity. This maneuver is allowed since there is no singularity at $w=\infty$. Then we wrap the contour around the branch cuts. The exact branch cuts depend on the relative magnitudes of $\alpha$ and $\beta$.\footnote{If $0 < \alpha < \beta < 1$ the branch cuts are $(-\infty, -\beta^{-2}] \cup [-\beta^2,0]$, otherwise if $0 < \beta < \alpha < 1$ they are $(-\infty, -\alpha^{-2}] \cup [-\alpha^2,0]$.} We will denote this new contour $\gamma_{\text{Br}}$. The new contours are shown in Figure \ref{fig:contourdeformation1}.
\begin{figure}[h]
    \centering
    \begin{tikzpicture}
    \draw (0,-2) -- (0,2);
    \draw (-4,0) -- (2,0);
    \clip (-4,-2) rectangle (3,2);
    \draw[ultra thick, blue] (0.25,0) ellipse (1.5 and 1.25);
    \draw[ultra thick, red] (-0.2,0) ellipse (0.3 and 0.1); 
    \draw[ultra thick, red] (-3.5,0) ellipse (1.2 and 0.1);
    \node at (0.3,0.3) {\tcr{$\gamma_{\text{Br},1}$}};
    \node at (-2,0.3) {\tcr{$\gamma_{\text{Br},2}$}};
    \node at (2.1,0.3) {\textcolor{blue}{$\gamma_{0,1}$}};
    \end{tikzpicture}
    \caption{Depiction of the new contours described in Section \ref{section:contours}, where $\gamma_{\text{Br}} = \gamma_{\text{Br},1} \cup \gamma_{\text{Br},2}$. The contours $\gamma_{\text{Br},1}$ and $\gamma_{\text{Br},2}$ are oriented clockwise.}
    \label{fig:contourdeformation1}
\end{figure}

Since the $w$-contour passes over $\gamma_{0,1}$ when performing this procedure, it introduces an extra contour integral because of the residue at $w=z$. So we have, 
\begin{multline} \label{eq:I2kdeformation}
    \Big[I^\alpha_{2,k}(4m',2\xi'+j;4m,2\xi+i)\Big]_{i,j=0}^1 \\
    = \frac{1}{2\pi\im}\oint_{\gamma_{0,1}} \frac{\dz}{z} z^{\xi' - \xi} r_{\alpha,2}(z)^{\nhalf-m'} r_{\alpha,k}(z)^{m-\frac{N}{2}}\frac{2}{1+2g_{\alpha,\beta}(z)} F_{\alpha,2}(z)F_{\beta,1}(z)F_{\alpha,1}(z)F_{\alpha,k}(z) \\
    + \frac{1}{(2\pi\im)^2} \oint_{\gamma_{\text{Br}}} \dw \oint_{\gamma_{0,1}} \frac{\dz}{z(z-w)} \frac{w^{\xi'+N}(z-1)^{N}}{z^{\xi+N}(w-1)^N} r_{\alpha,2}(w)^{\nhalf-m'} r_{\alpha,k}(z)^{m-\frac{N}{2}} \\
    \times \frac{2}{1+2g_{\alpha,\beta}(w)} F_{\alpha,2}(w)F_{\beta,1}(w)F_{\alpha,1}(w)F_{\alpha,k}(z).
\end{multline}
When $k = 2$, the single contour integral above is exactly zero. This is due to the orthogonality of the $F_\alpha$-matrices. So we can write, 
\begin{multline}
    \Big[I^\alpha_{2,2}(4m',2\xi'+j;4m,2\xi+i)\Big]_{i,j=0}^1 \\
    = \frac{1}{(2\pi\im)^2} \oint_{\gamma_1} \dw \oint_{\gamma_{0,1}} \frac{\dz}{z(z-w)} \frac{w^{\xi'+N}(z-1)^{N}}{z^{\xi+N}(w-1)^N} r_{\alpha,2}(w)^{\nhalf-m'} r_{\alpha,2}(z)^{m-\frac{N}{2}} \\
    \times \frac{2}{1+2g_{\alpha,\beta}(w)} F_{\alpha,2}(w)F_{\beta,1}(w)F_{\alpha,1}(w)F_{\alpha,2}(z).
\end{multline}
We may also write it in terms of the saddle functions, 
\begin{multline} \label{eq:I22saddle1}
    \Big[I^\eps_{2,2}(4m',2\xi'+j;4m,2\xi+i)\Big]_{i,j=0}^1 \\
    = \frac{1}{(2\pi\im)^2} \oint_{\gamma_{\text{Br}}} \dw \oint_{\gamma_{0,1}} \frac{\dz}{z(z-w)}\frac{2}{1+2g_{\alpha,\beta}(w)} F_{\alpha,2}(w)F_{\beta,1}(w)F_{\alpha,1}(w)F_{\eps,2}(z) \\
    \times \exp \left[ N\left(\psi_2(w;x,y) - \psi_2(z;x,y)\right) + \varphi_{\alpha,2}(w;x_2,y_2) - \varphi_{\alpha,2}(z;x_1,y_1) \right]
\end{multline}
However when $k = 1$, the single contour integral in equation \eqref{eq:I2kdeformation} is not zero. We provide the following lemma which proves that it decays exponentially in $N$,
\begin{lemma} \label{lem:singleintdecay}
    For $m, \, m', \, \xi,$ and $\xi'$ defined in equations \eqref{eq:localm} and \eqref{eq:localxi} such that $x < \nicefrac{1}{2}$, $m, \, m' < \nicefrac{N}{2}$ and $x_1, \, x_2 = o(N)$ ,
    $$\left|\frac{1}{2\pi\im}\oint_{\gamma_{0,1}} \frac{\dz}{z} z^{\xi' - \xi} r_{\alpha,2}(z)^{N-(m'+m)} \frac{2}{1+2g_{\alpha,\beta}(z)} F_{\alpha,2}(z)F_{\beta,1}(z)F_{\alpha,1}(z)\right| < Ce^{-cN}$$
\end{lemma}
\begin{proof}
Since there is no singularity at $z=1$, we deform the contour $\gamma_{0,1}$ to a circle of radius $1$ centered at the origin. We also write, 
$$M_\alpha(z) = \frac{2}{1+2g_{\alpha,\beta}(z)} F_{\alpha,2}(z)F_{\beta,1}(z)F_{\alpha,1}(z)$$
for brevity. So the contour integral can be written as, 
$$\frac{1}{2\pi\im}\oint_{|z|=1} \frac{\dz}{z} z^{\xi'-\xi} r_{\alpha,2}(z)^{N-(m'+m)} M_\alpha(z)$$
On the unit circle, $M_\alpha(z)$ is analytic and so it is bounded. Additionally, $|r_{\alpha,2}(z)| < 1$ on the unit circle. We may bound the integral accordingly, 
\begin{equation}
    \frac{1}{2\pi}\Bigg| \oint _{|z|=1} \frac{\dz}{z} z^{\xi'-\xi} r_{\eps,2}(z)^{N-(m+m')}M_\alpha(z)\Bigg| \leq C |r_{\alpha,2}^\text{max}|^{N-(m+m')} 
\end{equation}
where $r_{\alpha,2}^\text{max}$ denotes the maximum value of $r_{\alpha,2}(z)$ on the unit circle and $C$ is some constant matrix. Since $x < \nicefrac{1}{2}$ and $x_1, \, x_2 = o(N)$, the exponent $N-(m+m') = O(N)$ and is positive. The lemma follows. 
\end{proof}

%%%%%%%%%%%%%%%%%%%%%%%%%%%%%%%%%%%%%%%%%%%%%%%%%%%%%%%%%%%%%%%%%%%%%%%%%%%%%%%%%%%%%%%%%%%%%%%%%%%%%%%%%%%%%%%%%
\subsection{Local Asymptotics of $I^\alpha_{2,2}$} \label{section:I22proofs}
The goal of this section is to prove Proposition \ref{prop:I22decay}, which outlines the local asymptotics of the contour integral $I^\alpha_{2,2}$. The proposition is split into three cases and we will consider each of these cases separately. Before we prove the individual cases, we will first state and prove a lemma that will be used, 
\begin{lemma} \label{lem:F1F2exp}
    Let $z^*$ be any complex number besides $-\eps^2$ and $-\eps^{-2}$. Under the change of variables 
    \begin{equation*}
        z - z^* = sN^{-1/2} \qquad w-z^* = rN^{-1/2}
    \end{equation*}
    we have the expansion, 
    \begin{equation}
        F_{\eps,j}(w)F_{\eps,k}(z) = N^{-1/2}F_{\eps,j}(z^*)F'_{\eps,k}(z^*)(s-r) +O(N^{-1})
    \end{equation}
    for $j \neq k$.
\end{lemma}
\begin{proof}
To deal with this, we should expand these two matrix valued functions around the point $z^*$,
\begin{align*}
    F_{\eps,1}(w) &= F_{\eps,1}(z^*) + F'_{\eps,1}(z^*)(w-z^*) + O((w-z^*)^2) \\
    F_{\eps,2}(z) &= F_{\eps,2}(z^*) + F'_{\eps,2}(z^*)(z-z^*) + O((z-z^*)^2) 
\end{align*}
The expansions are valid as long as we are away from the branch points. Applying the change of variable gives, 
\begin{equation*}
    F_{\eps,j}(w)F_{\eps,k}(z) = N^{-1/2} F_{\eps,j}(z^*)F'_{\eps,k}(z^*)s+ N^{-1/2} F_{\eps,j}'(z^*)F_{\eps,k}(z^*)r +O(N^{-1})
\end{equation*}
Lastly, one can check that 
\begin{equation*}
    F_{\eps,j}(z^*)F'_{\eps,k}(z^*) = -F_{\eps,j}'(z^*)F_{\eps,k}(z^*)
\end{equation*}
yielding the results stated.
\end{proof}

\subsubsection{Case 1: The asymptotic coordinate is in the rough region}
We start with the expression of $I^\alpha_{2,2}$ given in equation \eqref{eq:I22saddle1}. From this form of the integral we will deform the contours to the contours of steepest descent/ascent detailed in Section \ref{section:saddles} and Lemma \ref{lem:roughsaddles}. Importantly, Lemma \ref{lem:roughsaddles} guarantees we can successfully deform the $z$ contour. Since $\Pi^+$ is open and contains $-\alpha^2$ or $-\alpha^{-2}$ it must also contain a point on the real line not on the branch cuts $(-\infty,-\alpha^{-2}] \cup [-\alpha^2,0]$, thus guaranteeing that the $z$ contour stays inside $\Pi^+$. Recall that these contours are depicted in Figure \ref{fig:roughcontours}. We express the contour integral in the following manner, 
\begin{multline}
    \Big[I^\alpha_{2,2}(4m',2\xi'+j;4m,2\xi+i)\Big]_{i,j=0}^1 = \oint_{\gamma_{R,D}} \oint_{\gamma_{R,A}} \frac{\dz\dw}{z(z-w)}\frac{2F_{\alpha,2}(w)F_{\beta,1}(w)F_{\alpha,1}(w)}{1+2g_{\alpha,\beta}(w)} F_{\alpha,2}(z) \\
    \times \exp \left[ N\left(\psi_2(w;x,y) - \psi_2(z;x,y)\right) + \varphi_{\alpha,2}(w;x_2,y_2) - \varphi_{\alpha,2}(z;x_1,y_1) \right] 
\end{multline}
Now we can use saddle point methods to compute the leading order decay of this integral. Since there are two different saddle points, we must consider the contributions from the four combinations.\footnote{We can let $w = z_3^*$ or $w= z_4^*$. The same applies for $z$, thus resulting in four combinations.} The contributions where different saddle points are used for $w$ and $z$ are immediate. We get the leading order term, 
\begin{multline} \label{eq:I22rough_est1}
    \frac{F_{\alpha,2}(z_l^*)F_{\beta,1}(z_l^*)F_{\alpha,1}(z_l^*)F_{\alpha,2}(z_k^*)\exp\Big[N(\psi_2(z_l^*)-\psi_2(z_k^*)) + \varphi_2(z_l^*)-\varphi_2(z_k^*) + 2\im\,\theta_{lk}(x,y)\Big]}{\pi z_k^*(z_k^*-z_l^*)(1+2g_{\alpha,\beta}(z_l^*))N \sqrt{|\psi_2''(z_k^*)\psi_2''(z_l^*)|}}
\end{multline}
where either $l=3$ and $k=4$ or $l=4$ and $k=3$. The function $\theta_{lk}(x,y)$ depends on the angle of the steepest ascent and descent contours around the saddle points. The term $\psi_2(z_l^*)-\psi_2(z_k^*)$ is purely imaginary, and thus does not contribute to the decay of the term. The absolute value of the above can be bounded by $CN^{-1}$, where $C$ is some constant matrix. \\

Now we want to consider the leading order contributions from the terms where we consider $z$ and $w$ at the same saddle point. We start by expanding the function $\psi_2(z)$ around $z = z_k^*$,
$$\psi_2(z) = \psi_2(z_k^*) + \psi_2''(z_k^*)(z-z_k^*)^2 + O((z-z_k^*)^3)$$
In addition, we make the following change of variables,
\begin{equation}
    z-z_k^* = sN^{-1/2} \qquad\qquad w-z_k^* = rN^{-1/2}
\end{equation}
which leads to,
\begin{equation*}
    \frac{\dz\dw}{z-w} = N^{-1/2}\frac{\ds\dr}{s-r} + o(N^{-1/2})
\end{equation*}
Combining this gives, 
\begin{equation}
    N \Big[\psi_2(w) - \psi_2(z)\Big] = \psi_2''(z_k^*)(s^2 - r^2) + O(N^{-1/2})
\end{equation}
We also apply Lemma \ref{lem:F1F2exp} to get,
\begin{equation}
    \frac{2F_{\alpha,2}(w)F_{\beta,1}(w)F_{\alpha,1}(w)F_{\alpha,2}(z)}{1+2g_{\alpha,\beta}(w)} = \frac{2F_{\alpha,2}(z_k^*)F_{\beta,1}(z_k^*)F_{\alpha,1}(z_k^*)F'_{\alpha,2}(z_k^*)}{1+2g_{\alpha,\beta}(z_k^*)} N^{-1/2}(s-r) + O(N^{-1})
\end{equation}
Combining all this we get the leading order asymptotics,
\begin{equation} \label{eq:I22rough_est2}
    \frac{1}{(2\pi)^2}\frac{2F_{\alpha,2}(z_k^*)F_{\beta,1}(z_k^*)F_{\alpha,1}(z_k^*)F'_{\alpha,2}(z_k^*)}{z_k^*(1+2g_{\alpha,\beta}(z_k^*))} N^{-1} \int_{\Gamma_r} \int_{\Gamma_s} \exp \big[ \psi_2''(z_k^*)(s^2 - r^2)\big] \ds \dr
\end{equation}
Where $\Gamma_r$ and $\Gamma_s$ are lines through the origin whose angles depend on the angle of the contours of steepest descent and ascent, respectively. Moreover, as a result of choosing the steepest ascent/descent contours of $\psi_2$, we know the integral above converges. Combining the results from equations \eqref{eq:I22rough_est1} and \eqref{eq:I22rough_est2} we obtain the results presented in part (1) of Proposition \ref{prop:I22decay}.

\subsubsection{Case 2: The asymptotic coordinate is in the smooth region, but not in the strong-coupling region.} \label{section:case2}
We start with the definition of $I^\alpha_{2,2}$ presented in equation \eqref{eq:I22saddle1}. We then deform $\gamma_{0,1}$ to the steepest ascent contour through the saddle point $z_2^*$. We refer to this contour as $\gamma_{s,2}$ and it is described in Lemma \ref{lem:psi2contours}. The resulting contours are depicted in Figure \ref{fig:contour22alpha}.
\begin{figure}[h]
    \centering
    \begin{tikzpicture}
    \draw (0,-2) -- (0,2);
    \draw (-4,0) -- (2,0);
    \clip (-4,-2) rectangle (3,2);
    \draw[ultra thick, blue] (0.25,0) ellipse (1.5 and 1.25);
    \draw[ultra thick, red] (-0.2,0) ellipse (0.3 and 0.1); 
    \draw[ultra thick, red] (-3.5,0) ellipse (1.2 and 0.1);
    \node at (0.3,0.3) {\tcr{$\gamma_{\text{Br},1}$}};
    \node at (-2,0.3) {\tcr{$\gamma_{\text{Br},2}$}};
    \node at (2.1,0.3) {\textcolor{blue}{$\gamma_{s,2}$}};
    \end{tikzpicture}
    \caption{Depiction of the contours $\gamma_{s,2}$, $\gamma_{\text{Br},1}$, and $\gamma_{\text{Br},2}$ in the case of $(x,y) \in \mathfrak{S}$ and $(x,y) \not\in \mathfrak{C}$.}
    \label{fig:contour22alpha}
\end{figure}
We write,
\begin{multline}
    \Big[I^\alpha_{2,2}(4m',2\xi'+j;4m,2\xi+i)\Big]_{i,j=0}^1 = \frac{1}{(2\pi\im)^2}\oint_{\gamma_{\text{Br}}} \oint_{\gamma_{s,2}} \frac{\dz\dw}{z(z-w)}\frac{2F_{\alpha,2}(w)F_{\beta,1}(w)F_{\alpha,1}(w)}{1+2g_{\alpha,\beta}(w)}  \\
    \times F_{\alpha,2}(z) \exp \left[ N\left(\psi_2(w;x,y) - \psi_2(z;x,y)\right) + \varphi_{\alpha,2}(w;x_2,y_2) - \varphi_{\alpha,2}(z;x_1,y_1) \right]
\end{multline}
By Lemma \ref{lem:psimax} and since we are considering $(x,y) \not\in \mathfrak{C}$, we know the saddle point, $z_2^*$, lies away from the branch cuts $(-\infty, -\beta^{-2}) \cup (-\beta^2,0]$ and somewhere on the interval $(-\beta^{-2},-\beta^2)$. Corollary \ref{cor:contoursstronger} guarantees the contours $\gamma_{\text{Br},1}$ and $\gamma_{\text{Br},2}$ can be deformed so that they lie entirely inside of $\Omega^-_2$. Thus
\begin{equation}
    \oRe \psi_2(w;x,y) < \oRe \psi_2 (z_2^*;x,y) \leq \oRe \psi_2(z;x,y)
\end{equation}
for all $w \in \gamma_{\text{Br},1}\cup\gamma_{\text{Br},2}$ and for all $z \in \gamma_{s,2}$. Moreover, Lemma \ref{lem:realpsiasymptotics} implies that as $w \to \infty$, $\oRe \psi_2(w) \to -\infty$. Combining these to statement mean that the term
$$\exp\Big[N\big(\psi_2(w;x,y) - \psi_2(z;x,y)\big)\Big]$$
is $O(e^{-cN})$ for some $c > 0$, uniformly for $(z,w) \in \gamma_z \times \gamma_{\text{Br},1}\cup\gamma_{\text{Br},2}$.

\subsubsection{Case 3: The asymptotic coordinate is in the smooth region and strong-coupling region. } 
The asymptotics become more complicated in the intersection of the smooth region and the strong-coupling region. This is because the saddle point $z_2^*$ now lies in the branch cuts (of the $w$ integral). While the deformation of the $z$-contour remains the same, it is no longer true that both the contours $\gamma_{\text{Br},1}$ and $\gamma_{\text{Br},2}$ will lie entirely inside of $\Omega_2^-$.

Before we detail the proof of case (3) in Proposition \ref{prop:I22decay}, we should first explain when our analysis fall into this category. Suppose $z_2^* \in (-\alpha^{-2},-\beta^{-2}]$, then $z_4^* > -1$. This implies that $\psi'_{\alpha,2}(-\beta^{-2};x,y) < 0$. This inequality gives the following relationship between $x$ and $y$, 
\begin{equation}
    y < \frac{\beta^{-2}}{\beta^{-2}+1} + \beta^{-2}\left(\frac{1}{2}-x\right)\log r_{\alpha,2}(-\beta^{-2})
\end{equation}
Alternatively, supposing that $z_2^* \in [-\beta^{2},-\alpha^{2})$ implies $\psi'_{\alpha,2}(-\beta^{2};x,y) > 0$. This translates to the following inequality, 
\begin{equation}
    y > \frac{\beta^{2}}{\beta^{2}+1} + \beta^{2}\left(\frac{1}{2}-x\right)\log r_{\alpha,2}(-\beta^{2})
\end{equation}
These inequalities describe the location of the strong-coupling region of the model. Examples of this region are given back in Figure \ref{fig:I22regions}. We are now prepared to consider the analysis of $I^\alpha_{2,2}$ for $(x,y) \in \mathfrak{S} \cap \mathfrak{C}$.

\begin{figure}[h]
    \centering
    \begin{tikzpicture}
    \draw (0,-2) -- (0,2);
    \draw (-4,0) -- (2,0);
    \clip (-4,-2) rectangle (3,2);
    \draw[ultra thick, blue] (0.5,0) ellipse (0.8 and 0.75);
    \draw[ultra thick, red] (-0.2,0) ellipse (0.3 and 0.1); 
    \draw[ultra thick, red] (-3.5,0) ellipse (1.2 and 0.1);
    \node at (0.3,0.3) {\tcr{$\gamma_{\text{Br},1}$}};
    \node at (-2,0.3) {\tcr{$\gamma_{\text{Br},2}$}};
    \node at (1.5,0.35) {\textcolor{blue}{$\gamma_{s,2}$}};
    \end{tikzpicture}
    \caption{Example of the contours $\gamma_{s,2}$, $\gamma_{\text{Br},1}$, and $\gamma_{\text{Br},2}$ when $(x,y) \in \mathfrak{S} \cap \mathfrak{C}$.}
    \label{fig:contour22special}
\end{figure}

\begin{proof}[Proof of Proposition \ref{prop:I22decay} (3)]
We will assume $z_2^* \in [-\beta^{2},-\alpha^{2})$. The procedure is nearly identical in the case of $z_2^* \in (-\alpha^{-2},-\beta^{-2}]$, the only difference is that the $\gamma_{s,2}$ intersects $\gamma_{\text{Br},2}$ instead of $\gamma_{\text{Br},1}$.

We start by deforming the contours in the same way we did for the case (2) in Section \ref{section:case2}. Figure \ref{fig:contour22special} gives an example of the contour deformation in this case. We split up the integral according to the $w$-contours, 
\begin{multline} 
    \Big[I^\alpha_{2,2}(4m',2\xi'+j;4m,2\xi+i)\Big]_{i,j=0}^1 \\
    = \frac{1}{(2\pi\im)^2}\oint_{\gamma_{\text{Br},1}} \oint_{\gamma_{s,2}} \frac{\dz\dw}{z(z-w)}\frac{2F_{\alpha,2}(w)F_{\beta,1}(w)F_{\alpha,1}(w)}{1+2g_{\alpha,\beta}(w)}F_{\alpha,2}(z)  \\
    \times \exp \left[ N\left(\psi_2(w;x,y) - \psi_2(z;x,y)\right) + \varphi_{\alpha,2}(w;x_2,y_2) - \varphi_{\alpha,2}(z;x_1,y_1) \right] \\
    + \frac{1}{(2\pi\im)^2}\oint_{\gamma_{\text{Br},2}}  \oint_{\gamma_{s,2}} \frac{\dz\dw}{z(z-w)}\frac{2F_{\alpha,2}(w)F_{\beta,1}(w)F_{\alpha,1}(w)}{1+2g_{\alpha,\beta}(w)} F_{\alpha,2}(z) \\
    \times \exp \left[ N\left(\psi_2(w;x,y) - \psi_2(z;x,y)\right) + \varphi_{\alpha,2}(w;x_2,y_2) - \varphi_{\alpha,2}(z;x_1,y_1) \right]. 
\end{multline}
Despite the fact that the contours $\gamma_{\text{Br},1}$ and $\gamma_{s,2}$ intersect, there is no single integral created due to the orthogonality of the $F_\alpha$-matrices. Since $\gamma_{\text{Br},2}$ can still be deformed to lie entirely inside of $\Omega_2^-$, we know that 
\begin{multline*}
    \frac{1}{(2\pi\im)^2}\oint_{\gamma_{\text{Br},2}} \oint_{\gamma_{s,2}} \frac{\dz\dw}{z(z-w)}\frac{2 F_{\alpha,2}(w)F_{\beta,1}(w)F_{\alpha,1}(w)}{1+2g_{\alpha,\beta}(w)} F_{\alpha,2}(z)\\
    \times \exp \left[ N\left(\psi_2(w;x,y) - \psi_2(z;x,y)\right) + \varphi_{\alpha,2}(w;x_2,y_2) - \varphi_{\alpha,2}(z;x_1,y_1) \right] 
\end{multline*}
is $O(e^{-cN})$ for the same reasons as detailed in Section \ref{section:case2}. 

Now we separate the contour $\gamma_{\text{Br},1}$ into two pieces: $\gamma_{\text{Br},\uparrow}$, which is the half lying above the real axis, and $\gamma_{\text{Br},\downarrow}$, which is the half lying below the real axis. We deform both of these contours to the negative real axis, which is the contour of steepest descent through $z_2^*$. We will start by just considering the contour integral 
\begin{multline} 
    \frac{1}{(2\pi\im)^2}\oint_{\gamma_{\text{Br},\uparrow}} \oint_{\gamma_{s,2}} \frac{\dz\dw}{z(z-w)}\frac{2 F_{\alpha,2}(w)F_{\beta,1}(w)F_{\alpha,1}(w)}{1+2g_{\alpha,\beta}(w)}F_{\alpha,2}(z) \\
    \times \exp \left[ N\left(\psi_2(w;x,y) - \psi_2(z;x,y)\right) + \varphi_{\alpha,2}(w;x_2,y_2) - \varphi_{\alpha,2}(z;x_1,y_1) \right] 
\end{multline}
Expanding the saddle function $\psi_2(z)$ around the saddle point $z_2^*$ yields 
\begin{equation}
    \psi_2(z;x,y) = \psi_2(z_2^*) + \psi_2''(z_2^*)(z-z_2^*)^2 + O((z-z_2^*)^3).
\end{equation}
Assuming we are away from the macroscopic boundaries, it follows that $\psi_2''(z_2^*) \neq 0$. Using the change of variables,
\begin{equation} \label{eq:changeofvariable2}
    z - z_2^* = sN^{-1/2} \qquad w-z_2^* = rN^{-1/2}
\end{equation}
and the expansion in Lemma \ref{lem:F1F2exp}, we have
\begin{equation}
    \frac{2F_{\alpha,2}(w)F_{\beta,1}(w)F_{\alpha,1}(w)F_{\alpha,2}(z)}{z(1+2g_{\alpha,\beta}(w))} = C_1 N^{-1/2}(s-r) + O(N^{-1}) 
\end{equation}
where $C_1$ is the matrix-valued constant,
\begin{equation}
    C_1 = \lim_{a \to 0^+} \frac{2F_{\alpha,2}(z_2^*+ \im a)F_{\beta,1}(z_2^*+ \im a)F_{\alpha,1}(z_2^*+ \im a)F'_{\alpha,2}(z_2^*+ \im a)}{(z_2^*+ \im a)(1+2g_{\alpha,\beta}(z_2^*+ \im a))}.
\end{equation}
One should note that while $F_{\beta,1}(w)$ has a pole at $w = -\beta^2$, so does $g_{\alpha,\beta}(w)$ and these poles cancel out in the limit. Applying the change of variable given in equation \eqref{eq:changeofvariable2} to the whole double contour integral yields, 
\begin{multline} \label{eq:92}
    \frac{1}{(2\pi\im)^2}\oint_{\gamma_{\text{Br},\uparrow}} \oint_{\gamma_{0,1}} \frac{\dz\dw}{z(z-w)}\frac{2F_{\alpha,2}(w)F_{\beta,1}(w)F_{\alpha,1}(w)}{1+2g_{\alpha,\beta}(w)}F_{\alpha,2}(z) \\
    \times \exp \left[ N\left(\psi_2(w;x,y) - \psi_2(z;x,y)\right) + \varphi_{\alpha,2}(w;x_2,y_2) - \varphi_{\alpha,2}(z;x_1,y_1) \right] \\
    = - C_1 \exp\big[ \varphi_{\alpha,2}(z_2^*; x_2-x_1, y_2-y_1) \big] N^{-1} \int_{\mathbb{R}} \dr \int_{\im\mathbb{R}} \ds \exp \Big[\psi_2''(z_2^*)(r^2-s^2)\Big] + o(N^{-1}).
\end{multline}
By Lemma \ref{lem:psimax}, we know that $\oRe \psi''_2(z_2^*) < 0$, since $z_2^*$ is a maximum of $\oRe \psi''_2(z)$, and so the integral above is convergent. We repeat the procedure with the contour $\gamma_{\text{Br},\downarrow}$, noting that 
\begin{equation}
    \frac{2F_{\alpha,2}(w)F_{\beta,1}(w)F_{\alpha,1}(w)F_{\alpha,2}(z)}{z(1+2g_{\alpha,\beta}(w))} = C_2 N^{-1/2}(s-r) + O(N^{-1}) 
\end{equation}
where, 
\begin{equation}
    C_2 = \lim_{a \to 0^-} \frac{2F_{\alpha,2}(z_2^*+ \im a)F_{\beta,1}(z_2^*+ \im a)F_{\alpha,1}(z_2^*+ \im a)F'_{\alpha,2}(z_2^*+ \im a)}{(z_2^*+ \im a)(1+2g_{\alpha,\beta}(z_2^*+ \im a))}.
\end{equation}
This constant is different from $C_1$ above, so the two terms do not cancel each other out. Note that the matrices $F_{\alpha,k}$, $F_{\beta,1}$, and $g_{\alpha,\beta}$ have different branch cuts, so it is not true that $C_1 = \bar{C}_2$. Expanding the double contour integral below the branch cut gives, 
\begin{multline} \label{eq:94}
    \frac{1}{(2\pi\im)^2}\oint_{\gamma_{\text{Br},\downarrow}} \oint_{\gamma_{0,1}} \frac{\dz\dw}{z(z-w)}\frac{2F_{\alpha,2}(w)F_{\beta,1}(w)F_{\alpha,1}(w)}{1+2g_{\alpha,\beta}(w)} F_{\alpha,2}(z)\\
    \times \exp \left[ N\left(\psi_2(w;x,y) - \psi_2(z;x,y)\right) + \varphi_{\alpha,2}(w;x_2,y_2) - \varphi_{\alpha,2}(z;x_1,y_1) \right] \\
    = C_2 \exp\big[ \varphi_{\alpha,2}(z_2^*; x_2-x_1, y_2-y_1) \big] N^{-1} \int_{\mathbb{R}} \dr \int_{\im\mathbb{R}} \ds \exp \Big[\psi_2''(z_2^*)(r^2-s^2)\Big] + o(N^{-1})
\end{multline}
Combining equations \eqref{eq:92} and \eqref{eq:94}, we conclude that for $(x,y) \in \mathfrak{S} \cap \mathfrak{C}$,
\begin{equation*}
    I_{2,2}(4m',2\xi'+j;4m,2\xi+i) = O(N^{-1}) 
\end{equation*}
where $i,j \in \{0,1\}$. 
\end{proof}

%%%%%%%%%%%%%%%%%%%%%%%%%%%%%%%%%%%%%%%%%%%%%%%%%%%%%%%%%%%%%%%%%%%%%%%%%%%%%%%%%%%%%%%%%%%%%%%%%%%%%%%%%%%%%%%%%%%%%%%%%%%%%%%%%%%%%%%%%%%%%%%%%%%%%%%%%%%%%%%%%%%%%%%%%%%%%%%%%%%%%%%%%%%%%%%%%%%%%%%%%%%%%%%%%%
\subsection{Local Asymptotics of $I^\alpha_{2,1}$} \label{section:I21proofs}
The asymptotics of the contour integral $I^\alpha_{2,1}$ are much simpler, in the sense that they only depend on $(x,y)$ being away from any macroscopic boundary and $x < \nicefrac{1}{2}$. We use the work done in Section \ref{section:contours} as a launching off point. In Lemma \ref{lem:singleintdecay} we showed that the single integral that arises from manipulating the contours decays exponentially in $N$, so to prove Proposition \ref{prop:I21decay} we only need to show that the following integral decays exponentially as well, 
\begin{multline}
    \frac{1}{(2\pi\im)^2} \oint_{\gamma_{\text{Br}}} \dw \oint_{\gamma_{0,1}} \frac{\dz}{z(z-w)}\frac{2}{1+2g_{\alpha,\beta}(w)} F_{\alpha,2}(w)F_{\beta,1}(w)F_{\alpha,1}(w)F_{\eps,1}(z) \\
    \times \exp \left[ N\left(\psi_2(w;x,y) - \psi_1(z;x,y)\right) + \varphi_{\alpha,2}(w;x_2,y_2) - \varphi_{\alpha,1}(z;x_1,y_1) \right]
\end{multline}
From Lemma \ref{lem:psimax}, we know that $\psi_1(z)$ has a maximum at $z_1^*$, so $\gamma_{0,1}$ should be deformed to the steepest ascent contour through this point, $\gamma_{s,1}$. After this deformation, we only need to show that 
\begin{equation} \label{eq:claimI21}
     \oRe \psi_2(w) < \oRe \psi_1(z_1^*)
\end{equation}
for all $w \in \gamma_{\text{Br}}$. Once this point is justified, the proof will follow the proof of case (2) in Proposition \ref{prop:I22decay}. In particular, adding the statement in Lemma \ref{lem:realpsiasymptotics} to the statement above is enough to conclude
$$\exp\Big[N\big(\psi_2(w;x,y) - \psi_1(z;x,y)\big)\Big]$$
is $O(e^{-cN})$ for some $c > 0$, uniformly for $(z,w) \in \gamma_z \times \gamma_{\text{Br},1} \cup \gamma_{\text{Br},2}$.

We will now prove the claim in equation \eqref{eq:claimI21}. We cannot immediately apply Corollary \ref{cor:contoursstronger} as we did in Section \ref{section:case2}, since now we are trying to compare $\psi_1(z_1^*)$ to $\psi_2(w)$. We use the fact that $r_{\alpha,2}(w) = r_{\alpha,1}(w)^{-1}$ and the fact that $x < \nicefrac{1}{2}$ to write,
\begin{align}
    \oRe \psi_2(w) &= (1+y)\log|w| - \log|w-1| + \left(\frac{1}{2}-x\right)\log|r_{\alpha,2}(w)| \nonumber\\
    &= (1+y)\log|w| - \log|w-1| - \left(\frac{1}{2}-x\right)\log|r_{\alpha,1}(w)| \nonumber\\
    &\leq (1+y)\log|w| - \log|w-1| + \left(\frac{1}{2}-x\right)\log|r_{\alpha,1}(w)| = \oRe\psi_1(w) \label{eq:test}
\end{align}
where equality only occurs when $w \in (-\infty, -\alpha^{-2}] \cup [-\alpha^2,0]$. We define the following sets, 
$$\Omega_{2,1}^+ = \left\{ w \in \mathbb{C} : \oRe \psi_2(w) > \oRe \psi_1(z_1^*) \right\}$$
$$\Omega_{2,1}^- = \left\{ w \in \mathbb{C} : \oRe \psi_2(w) < \oRe \psi_1(z_1^*) \right\}$$
We want to show that $\gamma_{\text{Br}}$ can be deformed such that $\gamma_{\text{Br}} \subset \Omega_{2,1}^-$. For $z_1^* \in (-\beta^{-2},-\beta^2)$, this is immediate from combining Corollary \ref{cor:contoursstronger} with the inequality in equation \eqref{eq:test}. Care should be taken in the case that $z_1^* \in (-\alpha^{-2},-\beta^{-2}] \cup [-\beta^2,-\alpha^2)$. In this case, we need the fact that the inequality in equation \eqref{eq:test} is strict at $z_1^*$, so  
$$\oRe \psi_2(z_1^*) < \oRe \psi_1(z_1^*)$$
When we combine this inequality with the continuity of $\oRe \psi_2(w)$, we can conclude that $z_1^*$ is a positive distance away from $\Omega_{2,1}^+$, and so we can deform the contours in the desired manner. 

%%%%%%%%%%%%%%%%%%%%%%%%%%%%%%%%%%%%%%%%%%%%%%%%%%%%%%%%%%%%%%%%%%%%%%%%%%%%%%%%%%%%%%%%%%%%%%%%%%%%%%%%%%%%%%%%%%%%%%%%%%%%%%%%%%%%%%%%%%%%%%%%%%%%%%%%%%%%%
\section{An Intermediate Correlation Kernel via a Non-Intersecting Paths Process} \label{section:BDstuff}
In Section \ref{section:pathsprocess}, we introduced a non-intersection paths model that we claimed was equivalent to the split two-periodic Aztec diamond. In this section we begin by justifying the equivalence of the models. This explanation can be found in Section \ref{section:processes}. These results are not new, the bijection was first detailed in \cite[Section 5]{BD19} and then further described in \cite{CD23}. Once we introduce this process, the formulation of the kernel follows from a slight modification of \cite[Theorem 5.2]{BD19}. This work is done in Section \ref{section:intermediatekernelproof}.

\subsection{Non-Intersecting Paths Process} \label{section:processes}
Recall the definition of the Aztec diamond graph given by equations \eqref{eq:coordstart}-\eqref{eq:edgecoord}. From this definition, we would like to label the four types of edges found in the Aztec diamond graph. We call an edge of the form $((2j+1, 2k), (2j+2, 2k+1))$ an \textit{east} edge, an edge of the form $((2j, 2k+1), (2j, 2k+2))$ a \textit{west} edge, an edge of the form $((2j, 2k+1), (2j+1, 2k))$ a \textit{south} edge, and an edge of the form $((2j+1, 2k), (2j+2, 2k-1))$ a \textit{north} edge. The different types of edges are depicted in Figure \ref{fig:DiamondtoDR}.

\subsubsection{The Aztec Diamond to DR Paths}
We will start by relating the Aztec diamond graph to the DR lattice paths model, which was first introduced in \cite{Jo03}. Under an appropriate choice of weighting this model is equivalent to the Aztec diamond. We will explain this bijection here. To start, we will describe the DR graph of size $n$. The DR graph of size $n$ has vertices given by, 
\begin{equation}
    \tV^{\text{DR}}_n = \{ (2j, 2k-1) : 0 \leq j \leq n, 0 \leq k \leq n\}
\end{equation}
and edges given by, 
\begin{multline}
    \tE^{\text{DR}}_n = \{ ((2j, 2k-1),(2j+2, 2k-1)) : 0 \leq j \leq n-1, 0 \leq k \leq n \} \\
    \cup \{ ((2j, 2k-1),(2j, 2k-3)) : 0 \leq j \leq n, 1 \leq k \leq n \}\\
    \cup \{ ((2j, 2k-1),(2j+2, 2k-3)) : 0 \leq j \leq n-1, 1 \leq k \leq n\}
\end{multline}
\begin{figure}[h]
    \centering
    \begin{tikzpicture}[scale=0.75]
    \foreach \t in {0,2,4,6,8}
        {
        \draw (0,\t) -- (8,\t);
        \draw (\t,0) -- (\t,8);
        \draw (0,\t) -- (\t,0);
        \draw (8,\t) -- (\t,8);
        \foreach \k in {0,2,4,6,8}
            {
            \filldraw[black] (\t,\k) circle (2pt);
            }
        }
        \draw[line width=0.1cm] (0,2) -- (2,2);
        \draw[line width=0.1cm] (2,2) -- (2,0);

        \draw[line width=0.1cm] (0,4) -- (2,4);
        \draw[line width=0.1cm] (2,4) -- (4,2);
        \draw[line width=0.1cm] (4,2) -- (4,0);

        \draw[line width=0.1cm] (0,6) -- (2,6);
        \draw[line width=0.1cm] (2,6) -- (4,6);
        \draw[line width=0.1cm] (4,6) -- (4,4);
        \draw[line width=0.1cm] (4,4) -- (6,2);
        \draw[line width=0.1cm] (6,2) -- (6,0);

        \draw[line width=0.1cm] (0,8) -- (2,8);
        \draw[line width=0.1cm] (2,8) -- (4,8);
        \draw[line width=0.1cm] (4,8) -- (6,6);
        \draw[line width=0.1cm] (6,6) -- (6,4);
        \draw[line width=0.1cm] (6,4) -- (8,4);
        \draw[line width=0.1cm] (8,4) -- (8,2);
        \draw[line width=0.1cm] (8,2) -- (8,0);
        \foreach \k in {0,2,4,6,8}
            {
            \node at (-1,\k) {\k};
            \node at (\k, -1) {\k};
            }
    \end{tikzpicture}
    \caption{The DR graph of size $n = 4$ with a non-intersecting paths configuration.}
    \label{fig:DRgraph}
\end{figure}
An example DR graph is shown in Figure \ref{fig:DRgraph}. We can draw non-intersecting path on the DR graph that start at the vertices $\{(0,2k-1): 1 \leq k \leq n \}$ and end at the vertices $\{(2j,0): 1 \leq j \leq n\}$. Thus the DR graph of size $n$ has $n$ non-intersecting paths on it. There is a bijection between coverings of the Aztec diamond and DR paths. Generally speaking, west edges become horizontal edges, east edges become vertical edges, and north edges become diagonal edges. Specifically we write, 
\begin{itemize}
    \item If a dimer covers the edge $((2j,2k+1),(2j+1,2k+2) \in \tE^{\text{Az}}_n$ then the edge $((2j,2k+1),(2j+2,2k+1)) \in \tE^{\text{DR}}_n$ is contained in a DR path.
    \item If a dimer covers the edge $((2j,2k+1),(2j-1,2k) \in \tE^{\text{Az}}_n$ then the edge $((2j,2k+1),(2j,2k-1)) \in \tE^{\text{DR}}_n$ is contained in a DR path.
    \item If a dimer covers the edge $((2j,2k+1),(2j+1,2k) \in \tE^{\text{Az}}_n$ then the edge $((2j,2k+1),(2j+2,2k-1)) \in \tE^{\text{DR}}_n$ is contained in a DR path.
\end{itemize}
If we translate our edge weights in equation \eqref{eq:edgewts} to the language of north, south, east and west edges we see that south and west edges always have a weight of 1, while north and east edges can have a weight of 1, $\eps^2$ or $\eps^{-2}$, where $\eps \in  \{ \alpha, \beta\}$. Since south edges are the only ones that do not translate to the DR paths and they always have a weight of 1 in our convention, the edge weights from the Aztec diamond directly translate to edge weights on the DR graph without changing the statistics of the model.
\begin{figure}[h]
    \centering
    \begin{tikzpicture}[scale=0.5]
    \foreach \t in {1,3,5,7}
        {
        \draw (0,\t) -- (\t,0);
        \draw (\t,8) -- (8,\t);
        \draw (\t,0) -- (8,8-\t);
        \draw (0,\t) -- (8-\t,8);
        }
    \foreach \j in {0,1,2,3,4}
        \foreach \k in {0,1,2,3}
            {
            \filldraw[black] (2*\j,2*\k+1) circle (2pt);
            }
    \foreach \j in {0,1,2,3}
        \foreach \k in {0,1,2,3,4}
            {
            \filldraw[color=black,fill=white] (2*\j+1,2*\k) circle (2pt);
            }
    \draw[line width=0.1cm,red] (0,7) -- (1,8);
    \draw[line width=0.1cm,red] (0,5) -- (1,6);
    \draw[line width=0.1cm,red] (0,3) -- (1,4);
    \draw[line width=0.1cm,red] (2,7) -- (3,8);
    \draw[line width=0.1cm,red] (2,5) -- (3,6);
    \draw[line width=0.1cm,blue] (7,0) -- (8,1);
    \draw[line width=0.1cm,blue] (5,0) -- (6,1);
    \draw[line width=0.1cm,blue] (1,2) -- (2,3);
    \draw[line width=0.1cm,blue] (5,4) -- (6,5);
    \draw[line width=0.1cm,blue] (3,4) -- (4,5);
    \draw[line width=0.1cm,green] (0,1) -- (1,0);
    \draw[line width=0.1cm,green] (2,1) -- (3,0);
    \draw[line width=0.1cm,green] (4,7) -- (5,6);
    \draw[line width=0.1cm,green] (6,3) -- (7,2);
    \draw[line width=0.1cm,green] (4,3) -- (5,2);
    \draw[line width=0.1cm,yellow] (7,8) -- (8,7);
    \draw[line width=0.1cm,yellow] (5,8) -- (6,7);
    \draw[line width=0.1cm,yellow] (7,6) -- (8,5);
    \draw[line width=0.1cm,yellow] (7,4) -- (8,3);
    \draw[line width=0.1cm,yellow] (3,2) -- (4,1);
    \end{tikzpicture}
    \hspace{1cm}
    \begin{tikzpicture}[scale=0.5]
    \foreach \t in {0,2,4,6,8}
        {
        \draw (0,\t) -- (8,\t);
        \draw (\t,0) -- (\t,8);
        \draw (0,\t) -- (\t,0);
        \draw (8,\t) -- (\t,8);
        \foreach \k in {0,2,4,6,8}
            {
            \filldraw[black] (\t,\k) circle (2pt);
            }
        }
        \draw[line width=0.1cm,red] (0,8) -- (2,8);
        \draw[line width=0.1cm,red] (2,8) -- (4,8);
        \draw[line width=0.1cm,green] (4,8) -- (6,6);
        \draw[line width=0.1cm,blue] (6,6) -- (6,4);
        \draw[line width=0.1cm,green] (6,4) -- (8,2);
        \draw[line width=0.1cm,blue] (8,2) -- (8,0);
        
        \draw[line width=0.1cm,red] (0,6) -- (2,6);
        \draw[line width=0.1cm,red] (2,6) -- (4,6);
        \draw[line width=0.1cm,blue] (4,6) -- (4,4);
        \draw[line width=0.1cm,green] (4,4) -- (6,2);
        \draw[line width=0.1cm,blue] (6,2) -- (6,0);
        
        \draw[line width=0.1cm,red] (0,4) -- (2,4);
        \draw[line width=0.1cm,blue] (2,4) -- (2,2);
        \draw[line width=0.1cm,green] (2,2) -- (4,0);
        
        \draw[line width=0.1cm,green] (0,2) -- (2,0);
    \end{tikzpicture}
    \caption{On the left is a covering of the Aztec diamond of size $n = 4$. East, west, south, and north edges are depicted in blue, red, green, and yellow respectively. The corresponding paths on the DR graph are shown in the figure on the right.}
    \label{fig:DiamondtoDR}
\end{figure}

\subsubsection{The Tower Aztec Diamond Graph}
The tower Aztec diamond graph was first described in \cite{BD19} and later named in \cite{CD23}. Informally, the $(n,p)$ tower Aztec diamond graph  consists of two Aztec diamonds, one of size $n$ and the other of size $n-1$, stitched together by a strip of the rotated square grid of length $p$. An example of the tower Aztec diamond graph is show in Figure \ref{fig:TowerGraphs}. The $(n,p)$ tower Aztec diamond graph has the vertex sets, 
\begin{equation}
    \tW_{n,p}^{\text{Tow}}= \left\{(2j+1,2k):0 \leq n-1, \, -p-n+1 \leq k \leq n \right\}
\end{equation}
\begin{multline}
    \tB_{n,p}^{\text{Tow}} = \big\{ (2j, 2k+1) : 0 \leq j \leq n, \, 0 \leq k \leq n-1 \text{ or }\\
    0 \leq j \leq n-1, \, -p \leq k \leq -1 \text{ or } 1 \leq j \leq n-1, \, -p-n \leq k \leq -1-p \big\}
\end{multline}
and edge set give by, 
\begin{multline}
    E^{\text{Tow}}_{n,p} = \big\{((2j+1, 2k),(2j+2, 2k+1)) : 0 \leq j \leq n-1, \, 0 \leq k \leq n-1
\text{ or }\\
0 \leq j \leq n-2, \, -p-n+1 \leq k \leq -1 \big\} \\
\cup \big\{((2j+1, 2k),(2j, 2k+1)) : 0 \leq j \leq n-1, \, -p \leq k \leq n-1
\text{ or }\\
0 \leq j \leq n-2, \, -p-n+1 \leq k \leq -1-p \big\} \\ 
\cup \big\{((2j+1, 2k),(2j + 2, 2k-1)) : 0 \leq j \leq n-1, \, 1 \leq k \leq n
\text{ or }\\ 
0 \leq j \leq n-2, \, 1-p-n \leq k \leq n \big\} \\
\cup \big\{((2j + 1, 2k),(2j, 2k-1)) : 0 \leq j \leq n-1,\, 1-p \leq k \leq n
\text{ or }\\
1 \leq j \leq n-1, 1-p-n \leq k \leq n \big\}
\end{multline}
One should realize that it is not possible to have a dimer configuration on the tower Aztec diamond where one dimer contains both a vertex from the corridor and a vertex from either of the appending Aztec diamonds. In fact, there is only one possible dimer configuration on the corridor. Thus the statistics of the tower Aztec diamond can be simplified to the statistics on the two appending Aztec diamond. See \cite{CD23} for more detail.
\begin{figure}[h]
    \centering
    \includegraphics[scale=0.5]{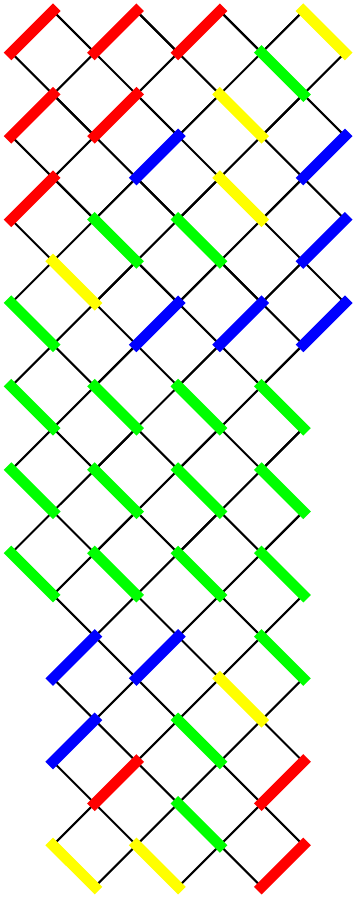}
    \hspace{1cm}
    \includegraphics[scale=0.5]{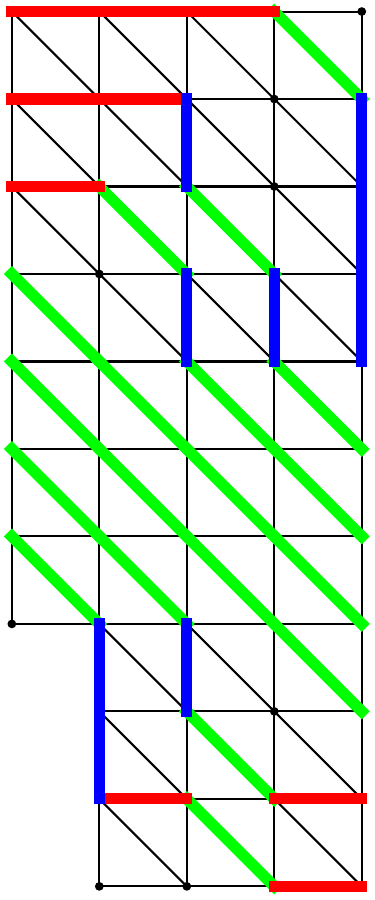}
    \caption{A dimer configuration on the $(4,3)$ tower Aztec diamond along with the analogous non-intersecting paths configuration on the $(4,3)$ tower DR graph.}
    \label{fig:TowerGraphs}
\end{figure}
%%%%%%%%%%%%%%%%%%%%%%%%%%%%%%%%%%%%%%%%%%%%%%%%%%%%%%%%%%%%%%%%%%%%%%%%%%%%%%%%%%%%
\subsubsection{Tower Aztec Diamond to Paths Process}
Just like in the case of the typical Aztec diamond, there is a bijection between the tower Aztec diamond and the tower DR graph, also depicted in Figure \ref{fig:TowerGraphs}. The vertices and edges of the tower DR graph are,
\begin{equation}
    \tV^{\text{Tow,DR}}_{n,p} = \big\{(2j, 2k-1) : 0 \leq j \leq n,\, - p \leq k \leq n \text{ or } 1 \leq j \leq n,\, -p-n + 1 \leq k \leq -p-1\big\}
\end{equation}
\begin{multline}
    \tE^{\text{Tow,DR}}_{n,p} = \big\{((2j, 2k-1),(2j + 2, 2k-1)) : 0 \leq j \leq n - 1,\, -p \leq k \leq n \big\}\\
    \cup \big\{((2j, 2k-1),(2j, 2k-3)) : 0 \leq j \leq n,\, -p + 1 \leq k \leq n\big\} \\
    \cup \big\{((2j, 2k-1),(2j + 2, 2k-3)) : 0 \leq j \leq n - 1,\, -p + 1 \leq k \leq n\big\} \\
    \cup \big\{((2j, 2k-1),(2j + 2, 2k-1)) : 1 \leq j \leq n -1,\, -n - p \leq k \leq -p - 1\big\} \\
    \cup \big\{((2j, 2k - 1),(2j, 2k-3)) : 1 \leq j \leq n,\, -n - p + 1 \leq k \leq-p - 1\big\} \\
    \cup \big\{((2j, 2k -1),(2j + 2, 2k - 3)) : 1 \leq j \leq n - 1, -n - p + 1 \leq k \leq -p -1\big\} \\
\end{multline}
Once again, we are considering a non-intersecting paths model on the tower DR graph. In this case there are $n+p$ paths starting at the vertices $\{(0,2k-1): -p \leq k \leq n \}$ and ending at the vertices $\{(2n,2k-1): -p-n+1 \leq k \leq p \}$. Once again, the weights from the tower Aztec diamond translate directly to the tower DR graph without changing the statistics as long as the south edges have weight one. 
\begin{figure}[h]
    \centering
    \includegraphics[scale=0.5]{TowerDR43.pdf}
    \hspace{1cm}
    \includegraphics[scale=0.5]{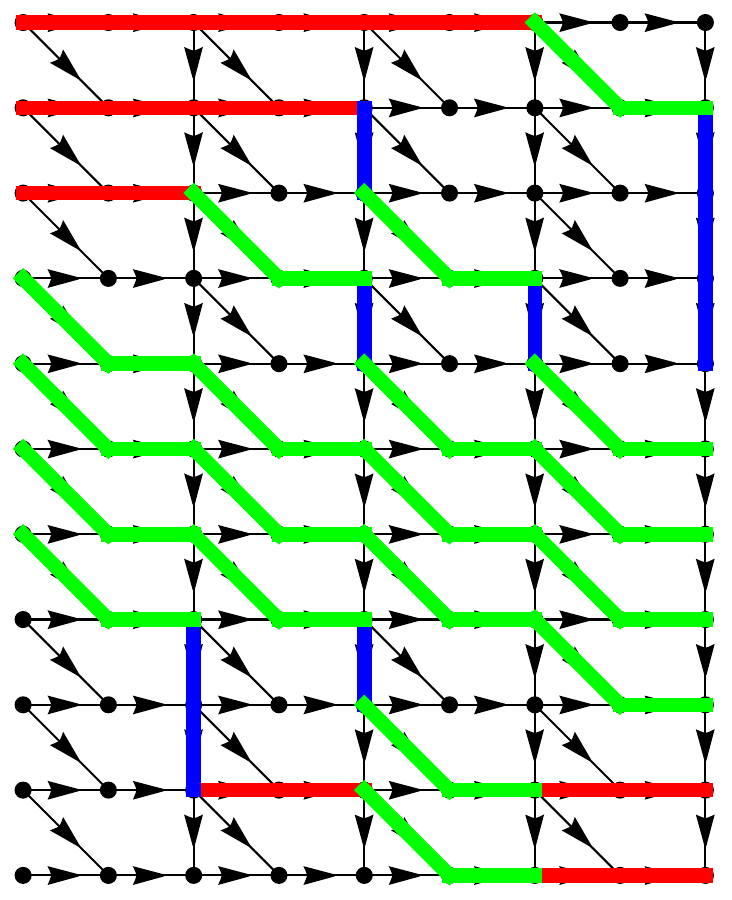}
    \caption{The tower DR graph and the equivalent tower graph from the work of Berggren and Duits. The two non-intersecting path configurations are in bijection.}
    \label{fig:TowerToBDPaths}
\end{figure}

The tower DR path process is nearly the process described in \cite{BD19}. To get exactly the process in their work, we should change the diagonal edges of the DR tower graph into a combination of a horizontal edge and diagonal edge, as depicted in Figure \ref{fig:TowerToBDPaths}. The additional horizontal edges will have weight one and thus not effect the statistics of the model. Once again we consider $n+p$ non intersection paths which start on the left at the top $n + p$ vertices and end on the right at the bottom $n+p$ vertices. 

Thus if we send the length of the corridor to infinity, the correlation kernel for the top part of the model,referred to as $K_{top}$, in \cite{BD19}) is also a correlation kernel of the Aztec diamond of size $n$. This correlation kernel is stated explicitly in \cite[Theorem 3.1]{BD19}.

\subsection{Proof of Lemma \ref{lemma:BDstuff}} \label{section:intermediatekernelproof}
Now we have introduced the necessary background to prove Lemma \ref{lemma:BDstuff}. We start by defining the matrices
\begin{equation} \label{eq:phiepsa}
    \Phi_{\eps,a}(z) = \frac{1}{(1-a^2 z^{-1})^2}\begin{pmatrix}1 & \frac{\eps^2}{a z} \\ \frac{1}{\eps^2 a} & 1\end{pmatrix}\begin{pmatrix}1 & \frac{\eps^2 a}{z} \\ \frac{a}{\eps^2} & 1\end{pmatrix} \begin{pmatrix}1 & \frac{1}{a z} \\ \frac{1}{a} & 1\end{pmatrix}\begin{pmatrix}1 & \frac{a}{z} \\ a & 1\end{pmatrix}
\end{equation}
and
\begin{equation} \label{eq:phian}
    \phi_{a,N}(z) = \Phi_{\alpha,a}(z)^{\frac{N}{2}}\Phi_{\beta,a}(z)^{\frac{N}{2}}
\end{equation}
Introducing this parameter $a \in (0,1)$ is necessary in order for us to directly apply \cite[Theorem 3.1]{BD19}. This is because the theorem does not permit singularities on the unit circle. We will call this weighting a split biased two-periodic weighting. Ultimately, we will send $a \to 1$ to obtain results for our model of interest. To compute the kernel, we will need to utilize the Weiner-Hopf factorization,
\begin{equation}
    \phi_{a,N}(z) = \tilde{\phi}_{-,a,N}(z)\tilde{\phi}_{+,a,N}(z)
\end{equation}
For a full description of the derivation and properties of the matrices $\tilde{\phi}_{-,a,N}(z)$ and $\tilde{\phi}_{+,a,N}(z)$ see \cite[Section 3]{BD19}. It is important to note that $\tilde{\phi}_{+,a,N}(z)^{\pm 1}$ is analytic for $|z| < 1$ and continuous for $|z| \leq 1$, while $\tilde{\phi}_{-,a,N}(z)^{\pm 1}$ is analytic for $|z|>1$ and continuous for $|z| \geq 1$. Moreover, the factorization commutes with the limit $a \to 1$. \\

The proof of Lemma \ref{lemma:BDstuff} depends on which side of the interface the coordinate $(4m',2\xi'+j)$ lies on. We will split this proof into two cases: when $m' \leq \nicefrac{N}{2}$ and $m' > \nicefrac{N}{2}$. 

%%%%%%%%%%%%%%%%%%%%%%%%%%%%%%%%%%%%%%%%%%%%
\subsubsection{Proof when $m' \leq \nicefrac{N}{2}$}
From \cite[Theorem 3.1]{BD19} we may immediately write the correlation kernel for the split biased two-periodic Aztec diamond,
\begin{multline} \label{eq:intialBDkernel}
    \Big[\mathbb{K}_{a,N}(4m',2\xi'+j;4m,2\xi+i)\Big]_{i,j=0}^1 = -\frac{\mathbb{I}_{m>m'}}{2 \pi \im} \oint_{\gamma_{0,1}} \frac{\dz}{z} z^{\xi'-\xi} \Phi_{\alpha,a}(z)^{\frac{N}{2}-m'}\Phi_{\eps,a}(z)^{m-\frac{N}{2}} \\
    + \frac{1}{(2\pi\im)^2} \oint_{\gamma_a} \dw \oint_{\gamma_{0,1,a}} \frac{\dz}{z(z-w)} \frac{w^{\xi'}}{z^{\xi}} \Phi_{\alpha,a}(w)^{\frac{N}{2}-m'}\Phi_{\beta,a}(w)^\frac{N}{2}\tilde{\phi}_{+,a,N}(w)^{-1}\tilde{\phi}_{-,a,N}(z)^{-1}\phi_{m,a,N}(z)
\end{multline}
where 
\begin{equation}
    \phi_{m,a,N}(z) = \begin{cases}
        \Phi_{\alpha,a}(z)^m & \text{if } m \leq \nicefrac{N}{2} \\
        \Phi_{\alpha,a}(z)^{\frac{N}{2}}\Phi_{\beta,a}(z)^{m-\frac{N}{2}} & \text{if } m > \nicefrac{N}{2}
    \end{cases}.
\end{equation}
The contour $\gamma_a$ is a simple closed curve around $a^2$, while the contour $\gamma_{0,1,a}$ is a simple closed curve containing $0$, $1$, and $\gamma_a$. To get the correlation kernel for the split two-periodic Aztec diamond, we wish to take the limit as $a \to 1$. The problem is the double contour integral in equation \eqref{eq:intialBDkernel} is singular at both $a^2$ and $a^{-2}$ (with respect to $w$), and these singularities are converging on opposite sides of the contour $\gamma_a$ in the limit. We must manipulate the integrand accordingly. We have
\begin{align}
    \Phi_{\alpha,a}(w)^{\frac{N}{2}-m'}\Phi_{\beta,a}(w)^{\frac{N}{2}}\tilde{\phi}_{+,a,N}(w)^{-1} &= \Phi_{\alpha,a}(w)^{-m'}\Phi_{\alpha,a}(w)^{\frac{N}{2}}\Phi_{\beta,a}(w)^{\frac{N}{2}}\tilde{\phi}_{+,a,N}(w)^{-1} \label{eq:2}\\
    &= \Phi_{\alpha,a}(w)^{-m'}\phi_{a,N}(w) \tilde{\phi}_{+,a,N}(w)^{-1}. \label{eq:3}
\end{align}
Now we can apply the eigen-decomposition of $\phi_{a,N}(w)=\Phi_{\alpha,a}(z)^{\frac{N}{2}}\Phi_{\beta,a}(z)^{\frac{N}{2}}$ to the above. We let $r_{a,1,N}(w)$ and $r_{a,2,N}(w)$ denote the eigenvalues of $\phi_{a,N}(w)$, which we do not state explicitly here. The key properties of the eigenvalues are stated below. Let $E_{a,N}(w)$ denote the eigenvector matrix of $\phi_{a,N}(w)$, with respect to the above eigenvalue ordering. Then $F_{a,1,N}(w)$ and $F_{a,2,N}(w)$ are the following matrices,
\begin{equation}
    F_{a,1,N}(w) = E_{a,N}(w)\begin{pmatrix} 1 & 0 \\ 0 & 0 \end{pmatrix}E_{a,N}(w)^{-1}
\end{equation}
\begin{equation}
    F_{a,2,N}(w) = E_{a,N}(w)\begin{pmatrix} 0 & 0 \\ 0 & 1 \end{pmatrix}E_{a,N}(w)^{-1}
\end{equation}
Using this decomposition, equation \eqref{eq:3} becomes, 
\begin{multline} \label{eq:1}
    \Phi_{\alpha,a}(w)^{\frac{N}{2}-m'}\Phi_{\beta,a}(w)^{\frac{N}{2}}\tilde{\phi}_{+,a,N}(w)^{-1} =\\
    \Phi_{\alpha,a}(w)^{-m'}\big(r_{a,1,N}(w)F_{a,1,N}(w) + r_{a,2,N}(w)F_{a,2,N}(w) \big) \tilde{\phi}_{+,a,N}(w)^{-1}
\end{multline}
The following lemma will be necessary for further simplification and is proven in Section \ref{section:analysisofeigen}. 
\begin{lemma} \label{lem:1}
    The eigenvalue $r_{a,1,N}(w)$ has a pole at $w=a^2$ and is analytic and non-zero at $w=a^{-2}$. The eigenvalue $r_{a,2,N}(w)$ is analytic and non-zero at $w=a^2$ and has a zero at $w=a^{-2}$. There are no other poles or zeros of the eigenvalue functions in the neighborhood of the contour $\gamma_a$.
\end{lemma}
\noindent We will also need the following statement which is a generalization of Lemma \ref{lem:FNanalytic}, 
\begin{lemma} \label{lem:FNanalyticGEN}
    For any $a \in (0,1]$, the matrices $F_{a,1,N}(w)$ and $F_{a,2,N}(w)$ are analytic at $w=a^2$ and $w=a^{-2}$.
\end{lemma}
This lemma is also proven in Section \ref{section:analysisofeigen}. We can now go back to considering the expression in equation \eqref{eq:1}. We distribute and consider the two parts of the sum separately. Firstly, note that  
\begin{equation}
    \Phi_{\alpha,a}(w)^{-m'}r_{a,2,N}(w)F_{a,2,N}(w)\tilde{\phi}_{+,a,N}(w)^{-1}
\end{equation}
is entirely analytic inside the contour $\gamma_a$ so, by the residue theorem, this part of the contour integral evaluates to zero. The other part of \eqref{eq:1} needs a little massaging,
\begin{align*}
    \Phi_{\alpha,a}(w)^{-m'}r_{a,1,N}(w)F_{a,1,N}(w)&\tilde{\phi}_{+,a,N}(w)^{-1} = \\
    &=\Phi_{\alpha,a}(w)^{-m'}\phi_{a,N}(w)\phi_{a,N}(w)^{-1}r_{a,1,N}(w)F_{a,1,N}(w)\phi_{a,N}(w)^{-1}\tilde{\phi}_{-,a,N}(w) \\
    &= \Phi_{\alpha,a}(w)^{\frac{N}{2}-m'}\Phi_{\beta,a}(w)^{\frac{N}{2}}r_{a,1,N}(w)^{-1}F_{a,1,N}(w)\tilde{\phi}_{-,a,N}(w)
\end{align*}
All of the parts of the above expression are analytic at $a^{-2}$, so we only have poles inside the contour $\gamma_a$.\footnote{The matrices $\Phi_{\eps,a}(w)$, where $\eps = \alpha, \, \beta$, have poles at $a^2$ while their inverses have poles at $a^{-2}$. The way the expression is written guarantees that these matrices have positive exponents.} Thus, by dominated convergence, we may take the limit $a \to 1$. When $a = 1$, we will drop the subscript denoting the value of $a$. Note that, 
\begin{equation}
    \tilde{\phi}_{-,N}(w) = \frac{1}{(1-w^{-1})^N}\Phi_{\alpha}(w)^{\frac{N}{2}}
\end{equation}
by the Weiner-Hopf factorization.\footnote{See \cite[Section 5]{BD19} for details on how to compute this.} Thus the correlation kernel for the split two-periodic Aztec diamond may be written in the following manner in the case where $0 < m' \leq \nicefrac{N}{2}$,
\begin{multline}
    \Big[\mathbb{K}_N(4m',2\xi'+j;4m,2\xi+i)\Big]_{i,j=0}^1 = -\frac{\mathbb{I}_{m>m'}}{2 \pi \im} \oint_{\gamma_{0,1}} \frac{\dz}{z} z^{\xi'-\xi} \Phi_{\alpha}(z)^{\frac{N}{2}-m'}\Phi_{\eps}(z)^{m-\frac{N}{2}} \\
    + \frac{1}{(2\pi\im)^2} \oint_{\gamma_1} \dw \oint_{\gamma_{0,1}} \frac{\dz}{z(z-w)} \frac{w^{\xi'+N}(z-1)^{N}}{z^{\xi+N}(w-1)^N} \Phi_\alpha(w)^{-m'}F_{1,N}(w)\Phi_\alpha(w)^{\frac{N}{2}}\Phi_\eps(z)^{m-\frac{N}{2}}
\end{multline}
where $\eps = \alpha$ if $m \leq N/2$ and $\eps = \beta$ if $m > N/2$. 
%%%%%%%%%%%%%%%%%%%%%%%%%%%%%%%%%%%%%%%%%%%%%%%%%%%%%%
\subsubsection{Proof when $m' > \nicefrac{N}{2}$}
We state the results of \cite[Lemma 3.1]{BD19} under the stated condition, 
\begin{multline}
    \Big[\mathbb{K}_{a,N}(4m',2\xi'+j;4m,2\xi+i)\Big]_{i,j=0}^1 = -\frac{\mathbb{I}_{m>m'}}{2 \pi \im} \oint_{\gamma_{0,1}} \frac{\dz}{z} z^{\xi'-\xi} \Phi_{\alpha,a}(z)^{\frac{N}{2}-m'}\Phi_{\eps,a}(z)^{m-\frac{N}{2}} \\
    + \frac{1}{(2\pi\im)^2} \oint_{\gamma_a} \dw \oint_{\gamma_{0,1,a}} \frac{\dz}{z(z-w)} \frac{w^{\xi'}}{z^{\xi}} \Phi_{\beta,a}(w)^{N-m'}\tilde{\phi}_{+,a,N}(w)^{-1}\tilde{\phi}_{-,a,N}(z)^{-1}\phi_{m,a,N}(z)
\end{multline}
The double contour integral above has the same issue as the double contour integral in equation \eqref{eq:intialBDkernel}, the $w$-part of the integral has poles at both $a^2$ and $a^{-2}$ that we must address before taking the limit. 
\begin{align}
    \Phi_{\beta,a}(w)^{N-m'}\phi_{+,a,N}(w)^{-1} &= \Phi_{\beta,a}(w)^{\nhalf-m'}\Phi_{\beta,a}(w)^\nhalf \phi_{+,a,N}(w)^{-1} \\
    &= \Phi_{\beta,a}(w)^{\nhalf-m'}\Phi_{\alpha,a}(w)^{-\nhalf} \Phi_{\alpha ,a}(w)^\nhalf\Phi_{\beta,a}(w)^\nhalf \phi_{+,a,N}(w)^{-1}
\end{align}
The above manipulations have introduced a $\phi_{a,N}(w)=\Phi_{\alpha ,a}(w)^\nhalf\Phi_{\beta,a}(w)^\nhalf$ to the above product, which we replace this with its eigen-decomposition to obtain
\begin{equation}
    \Phi_{\beta,a}(w)^{\nhalf-m'}\Phi_{\alpha,a}(w)^{-\nhalf}\big(r_{a,1,N}(w)F_{a,1,N}(w) + r_{a,2,N}(w)F_{a,2,N}(w) \big) \phi_{+,a,N}(w)^{-1}
\end{equation}
The expression
\begin{equation}
    \Phi_{\beta,a}(w)^{\nhalf-m'}\Phi_{\alpha,a}(w)^{-\nhalf}r_{a,2,N}(w)F_{a,2,N}(w)\phi_{+,a,N}(w)^{-1}
\end{equation}
only has poles outside of the unit circle, and so it contributes nothing to the integral. We are left to manipulate the following,
\begin{align}
    \Phi_{\beta,a}(w)^{\nhalf-m'}\Phi_{\alpha,a}(w)^{-\nhalf}r_{a,1,N}(w)&F_{a,1,N}(w)\phi_{+,a,N}(w)^{-1} \nonumber \\
    &= \Phi_{\beta,a}(w)^{N-m'} \phi_{a,N}(w)^{-1}r_{a,1,N}(w)F_{a,1,N}(w)\phi_{a,N}(w)^{-1}\phi_{-,a,N}(w) \\
    &= \Phi_{\beta,a}(w)^{N-m'} r_{a,1,N}(w)^{-1}F_{a,1,N}(w)\phi_{-,a,N}(w)
\end{align}
The resulting expression only has poles inside the unit circle, so we are now free to take the limit $a \to 1$. Doing so gives,
\begin{multline}
    \Big[\mathbb{K}_N(4m',2\xi'+j;4m,2\xi+i)\Big]_{i,j=0}^1 = -\frac{\mathbb{I}_{m>m'}}{2 \pi \im} \oint_{\gamma_{0,1}} \frac{\dz}{z} z^{\xi'-\xi} \Phi_{\alpha}(z)^{\frac{N}{2}-m'}\Phi_{\eps}(z)^{m-\frac{N}{2}} \\
    + \frac{1}{(2\pi\im)^2} \oint_{\gamma_1} \dw \oint_{\gamma_{0,1}} \frac{\dz}{z(z-w)} \frac{w^{\xi'+N}(z-1)^{N}}{z^{\xi+N}(w-1)^N} \Phi_{\beta}(w)^{N-m'}r_{1,N}^{-1}(w)F_{1,N}(w)\Phi_\alpha(w)^{\frac{N}{2}}\Phi_\eps(z)^{m-\frac{N}{2}}
\end{multline}
After taking the limit, it will benefit us later to make the following substitution,
\begin{align*}
    \Phi_{\beta}(w)^{N-m'} &= \Phi_{\beta}(w)^{\nhalf-m'}\Phi_{\alpha}(w)^{-\nhalf}\Phi_\alpha(w)^\nhalf\Phi_\beta(w)^\nhalf \\
    &= \Phi_{\beta}(w)^{\nhalf-m'}\Phi_{\alpha}(w)^{-\nhalf}\phi_N(w)
\end{align*}
And so the kernel becomes, 
\begin{multline}
    \mathbb{K}_N(4m',2\xi'+j;4m,2\xi+i) = -\frac{\mathbb{I}_{m>m'}}{2 \pi \im} \oint_{\gamma_{0,1}} \frac{\dz}{z} z^{\xi'-\xi} \Phi_{\alpha}(z)^{\frac{N}{2}-m'}\Phi_{\eps}(z)^{m-\frac{N}{2}} \\
    + \frac{1}{(2\pi\im)^2} \oint_{\gamma_1} \dw \oint_{\gamma_{0,1}} \frac{\dz}{z(z-w)} \frac{w^{\xi'+N}(z-1)^{N}}{z^{\xi+N}(w-1)^N} \Phi_{\beta}(w)^{\nhalf-m'}\Phi_{\alpha}(w)^{-\nhalf}F_{1,N}(w)\Phi_\alpha(w)^{\frac{N}{2}}\Phi_\eps(z)^{m-\frac{N}{2}}
\end{multline}
%%%%%%%%%%%%%%%%%%%%%%%%%%%%%%%%%%%%%%%%%%%%%%%%%%%%%%%%%%%%%%%%%%%%%%%%%%%%%%%%%%%%%%%%%%%%%%%%%%%%%%%%%%%%%%%%%%%%%%%%%%%%%%%%%%%%%%%%%%%%%%%%%%%%%%%%%%%%%
\section{Analysis of the Eigen-Decomposition} \label{section:analysisofeigen}
In this section, we take a closer look at the behavior of the eigenvalues and eigenvectors of $\phi_{a,N}(z)$. We first prove Lemma \ref{lem:r2Nexpansion}, which appears in the proof of Theorem \ref{theorem:kernel} in Section \ref{section:proofofkernel1}. Next, we analyze the poles and zeros of the eigenvalues $r_{a,1,N}(z)$ and $r_{a,2,N}(z)$. In particular, we prove Lemma \ref{lem:1}. Lastly, we prove Lemma \ref{lem:FNanalyticGEN}. Lemmas \ref{lem:1} and \ref{lem:FNanalyticGEN} are necessary for the statement of the intermediate correlation kernel in Lemma \ref{lemma:BDstuff}. They are also generalizations of Lemmas \ref{lem:rnpole} and \ref{lem:FNanalytic}, respectively, which both appear in the proof of Theorem \ref{theorem:kernel}.

%%%%%%%%%%%%%%%%%%%%%%%%%%%%%%%%%%%%%%%%%%%%%%%%
\subsection{Proof of Lemma \ref{lem:r2Nexpansion}} 
We are immediately ready to prove Lemma \ref{lem:r2Nexpansion}. Recall that we can express $r_{2,N}(z)$ in terms of the trace of $\phi_N(z)$, 
$$r_{2,N}(z) = \frac{1}{2}\left(t_N(z) - \sqrt{t_N(z)^2-4}\right)$$
where $t_N(z) = \text{tr } \phi_N (z)$ is given in equation \eqref{eq:trace}. Fix $z \in \mathbb{C} \backslash \mathcal{B}$ and choose $N$ so that $t_N(z)$ is sufficiently large. We know that this is the case for any $z$ away from the branch cuts because $|r_{\alpha,1}(z)|, \, |r_{\beta,1}(z)| >1$. We can then expand $r_{2,N}(z)$ as a series in $t_N(z)$ around infinity,
\begin{equation}
    r_{2,N}(z) = \sum_{k=1}^\infty c_k t_N(z)^{1-2k}.
\end{equation}
It is important to note that the coefficients $c_k$ are independent of $z$ and $N$, however their exact value is not important so we do not explicitly state them. For simplicity, we let $x = r_{\alpha,1}(z)^\nhalf$ and $y = r_{\beta,1}(z)^\nhalf$. We write, 
\begin{equation}
    t_N(z) = (x^{-1}y + xy^{-1})\Big[\big(\nicefrac{1}{2}+g_{\al,\beta}(z)\big)u(x,y) + \big(\nicefrac{1}{2} - g_{\al,\beta}(z)\big)\Big]
\end{equation}
where 
\begin{equation}
    u(x,y) = \frac{xy + x^{-1}y^{-1}}{xy^{-1} + x^{-1}y}
\end{equation}
For any $z \in \mathbb{C} \backslash \mathcal{B}$ and $N$ sufficiently large, $u(x,y)$ also tends to infinity. So we can write the series expansion of $t_N(z)^{1-2k}$ in $u(x,y)$ about infinity, 
\begin{equation}
    t_N(z)^{1-2k} = (x^{-1}y + xy^{-1})^{1-2k}\sum_{j=2k-1}^\infty c'_{k,j} \frac{\big(\nicefrac{1}{2}-g_{\al,\beta}(z)\big)^{j-2k+1}}{\big(\nicefrac{1}{2}+g_{\al,\beta}(z)\big)^j}u(x,y)^{-j}
\end{equation}
where $c'_{k,j}$ are coefficients independent of $z$ and $N$ that we don't state explicitly. We can plug the above into our expansion of $r_{2,N}(z)$ to get, 
\begin{equation}
    r_{2,N}(z) = \sum_{k=1}^\infty \sum_{j=2k-1}^\infty c''_{k,j} \frac{\big(\nicefrac{1}{2}-g_{\al,\beta}(z)\big)^{j-2k+1}}{\big(\nicefrac{1}{2}+g_{\al,\beta}(z)\big)^j}\left(\frac{1}{xy+x^{-1}y^{-1}}\right)^{j}(x^{-1}y+xy^{-1})^{j-2k+1}
\end{equation}
where $c''_{k,j} = c_kc'_{k,j}$. The last piece is simply a Laurent polynomial so we write, 
\begin{equation}
    (x^{-1}y+xy^{-1})^{j-2k+1} = \sum_{s=0}^{j-2k+1} d_{j,k,s} \left(\frac{y}{x}\right)^{2s-j+2k-1}
\end{equation}
We also write
\begin{equation}
    \left(\frac{1}{xy+x^{-1}y^{-1}}\right)^j = \frac{(xy)^j}{\big((xy)^2+1\big)^j}
\end{equation}
and expand in $xy$ around infinity. Combining these results gives, 
\begin{equation}
    r_{2,N}(z)= \sum_{k=1}^\infty \sum_{j=2k-1}^\infty c''_{k,j} \frac{\big(\nicefrac{1}{2}-g_{\al,\beta}(z)\big)^{j-2k+1}}{\big(\nicefrac{1}{2}+g_{\al,\beta}(z)\big)^j} \left(\sum_{r=j}^\infty d'_{r,j} \left(\frac{1}{xy}\right)^{2r-j}\right)\left(\sum_{s = 0}^{j-2k+1} d_{j,k,s}\left(\frac{y}{x}\right)^{2s-j+2k-1}\right)
\end{equation}
Once again we've introduced the coefficient $d'_{r,j}$ and $d''_{j,k,s}$ which are independent of $z$ and $N$. This expansion is convergent as long as $z \not\in \mathcal{B}$ and $N$ is sufficiently large. Besides $x$ and $y$, the only part of the expansion that depends on $z$ are the terms
$$\frac{\big(\nicefrac{1}{2}-g_{\al,\beta}(z)\big)^{j-2k+1}}{\big(\nicefrac{1}{2}+g_{\al,\beta}(z)\big)^j}$$
It is easy to check that these terms have no pole at $z = 1$ given the definition of $g_{\alpha,\beta}(z)$ in equation \eqref{eq:g}.

We next need to inspect the possible orders of $x$ and $y$ that appear in the expansion. Notice that the exponent in the last sum is at most $j - 2k+1$ and at least $-j+2k-1$. Moreover, the smallest exponent in the penultimate sum is $j$. Checking the possible powers on $x$ and $y$ gives,
\begin{equation}
    r_{2,N}(z) = \sum_{p=0}^\infty \sum_{q=0}^\infty c_{p,q}(z) x^{-(2p+1)}y^{-(2q+1)}
\end{equation}
The statement in Lemma \ref{lem:r2Nexpansion} follows from replacing $x$ and $y$ with $r_{\alpha,1}(z)$ and $r_{\beta,1}(z)$, respectively. 

%%%%%%%%%%%%%%%%%%%%%%%%%%%%%%%%%%%%%%%%%%%
\subsection{Poles and Zeros of $r_{a,k,N}(w)$}
Our ultimate goal in this section is to prove Lemma \ref{lem:1}. In doing so, we will also prove Lemma \ref{lem:rnpole}, which is a direct consequence of Lemma \ref{lem:1}. We begin by stating some equations necessary in defining the eigenvalues $r_{a,k,N}(w)$. We start by stating the eigenvalues of $\Phi_{\eps,a}(w)$,
\begin{multline}
    r_{a,\eps,1}(w) = \frac{1}{(w-a^2)^2}\Big((w+1)^2+\frac{1}{2}w(a+a^{-1})^2(\eps^2+\eps^{-2})\\ +(a+a^{-1})(\eps+\eps^{-1})\sqrt{w(w^2+x_{\eps,a}w+1)}\Big)
\end{multline}
\begin{multline}
    r_{a,\eps,2}(w) = \frac{1}{(w-a^2)^2}\Big((w+1)^2+\frac{1}{2}w(a+a^{-1})^2(\eps^2+\eps^{-2})\\ -(a+a^{-1})(\eps+\eps^{-1})\sqrt{w(w^2+x_{\eps,a}w+1)}\Big)
\end{multline}
where
\begin{equation*}
    x_{\eps,a} = \frac{1}{4}\left((a+a^{-1})^2(\eps^2+\eps^{-2})-2(a - a^{-1})^2\right)
\end{equation*}
The function $r_{a,\eps,1}(w)$ has a pole at $w = a^2$ and the function $r_{a,\eps,2}(w)$ has a zero at $w = a^{-2}$. Moreover, these are the only poles and zeros of these functions. These eigenvalues can be used to help compute the eigenvalues of the matrix-valued function $\phi_{a,N}(w)$. First we state, 
\begin{equation}
    \det \Phi_{\eps,a}(w) = \frac{(1-a^{-2}w^{-1})^2}{(1-a^2w^{-1})^2}
\end{equation}
and so
\begin{equation}
    \det \phi_{a,N}(w) = \frac{(1-a^{-2}w^{-1})^{2N}}{(1-a^2w^{-1})^{2N}}
\end{equation}
Also, 
\begin{multline}
    \text{tr } \phi_{a,N}(w) = \left(\frac{1}{2} + g_{a,\alpha,\beta}(w)\right)\left(r_{a,\alpha,1}(w)^\nhalf r_{a,\beta,1}(w)^\nhalf + r_{a,\alpha,2}(w)^\nhalf r_{a,\beta,2}(w)^\nhalf \right) \\
    + \left(\frac{1}{2} - g_{a,\alpha,\beta}(w)\right)\left(r_{a,\alpha,1}(w)^\nhalf r_{a,\beta,2}(w)^\nhalf + r_{a,\alpha,2}(w)^\nhalf r_{a,\beta,1}(w)^\nhalf \right)
\end{multline}
where
\begin{multline}
    g_{a,\alpha,\beta}(w) = \Big(2a^2(\alpha^2+\beta^2)(w^2-1)+w\left((\alpha^2-1)(\beta^2-1)(a^4+1)+2a^2(\alpha^2+1)(\beta^2+1)\right)\Big)\Big/ \\
    \Big(2\sqrt{(a^2+1)^2(\alpha^4+1)w-2\alpha^2\left((a^4+1)w-2a^2(w^2+w+1)\right)} \\
    \times \sqrt{(a^2+1)^2(\beta^4+1)w-2\beta^2\left((a^4+1)w-2a^2(w^2+w+1)\right)}\Big)
\end{multline}
Note that these are just the generalized versions of equations \eqref{eq:trace} and \eqref{eq:g}. Now we may write the functions $r_{a,1,N}(w)$ and $r_{a,2,N}(w)$ in terms of the trace and determinant above,
\begin{equation}
    r_{a,1,N}(w) = \frac{1}{2}\left(\text{tr } \phi_{a,N}(w) + \sqrt{(\text{tr }\phi_{a,N}(w))^2 - 4 \det \phi_{a,N}(w)}\right)
\end{equation}
\begin{equation}
    r_{a,2,N}(w) = \frac{1}{2}\left(\text{tr } \phi_{a,N}(w) - \sqrt{(\text{tr }\phi_{a,N}(w))^2 - 4 \det \phi_{a,N}(w)}\right)
\end{equation}
We are now ready to prove Lemma \ref{lem:1}.

\begin{proof}
We will start by justifying that the only potential poles or zeros of the function $r_{a,k,N}(w)$ with positive real part can occur at the points $w=a^2$ and $w=a^{-2}$. Recall the equation, 
\begin{equation} \label{eq:detrelation}
    r_{a,1,N}(w)r_{a,2,N}(w) = \det \phi_{a,N}(w)
\end{equation}
From this we can deduce that, other than the points $w = a^2$ and $w = a^{-2}$, any pole of $r_{a,1,N}(w)$ must be a zero of $r_{a,2,N}(w)$ and vice versa. Additionally, since 
\begin{equation}
    r_{a,1,N}(w) + r_{a,2,N}(w) = \text{tr } \phi_{a,N}(w)
\end{equation}
any additional poles or zeros must also be poles of $\text{tr } \phi_{a,N}(w)$. Other than $w = a^2$, the potential poles of $\text{tr } \phi_{a,N}(w)$ come from the poles of $g_{a,\alpha,\beta}(w)$. These are explicitly, 
\begin{equation*}
    \frac{2 \left(a^2-1\right)^2 \eps^2-\left(a^2+1\right)^2 \left(\eps^4+1\right) \pm \left(a^2+1\right) \left(\eps^2-1\right) \sqrt{\left(a^2+1\right)^2(\eps^4+1) - 2 \left(a^4-6 a^2+1\right) \eps^2}}{8 a^2 \eps^2}
\end{equation*}
For $\eps = \alpha, \, \beta$. One should observe that these values are always negative and real for $a \in (0,1]$. For $a = 1$, these are the points $w = -\eps^2$ and $w = -\eps^{-2}$. 

Now that we have establish that the only poles or zeros of $r_{a,1,N}(w)$ and $r_{a,2,N}(w)$ in the neighborhood of $\gamma_a$ are at $w = a^2$ and $w = a^{-2}$, we can look specifically at the order of these poles by considering the Laurent expansions around these points. We start with the following expansions, about $w=a^2$ 
\begin{equation}
    r_{a,\eps,1}(w) = \frac{c_{a,\eps,+} }{(w-a^2)^2} + O\big((w-a^2)^{-1}\big)
\end{equation}
\begin{equation}
    r_{a,\eps,2}(w) = d_{a,\eps,+} + O\big(w-a^2\big)
\end{equation}
About $w = a^{-2}$ we have,
\begin{equation}
    r_{a,\eps,1}(w) = c_{a,\eps,-} +  O\big(w-a^{-2}\big)
\end{equation}
\begin{equation}
    r_{a,\eps,2}(w) = d_{a,\eps,-}(w-a^{-2})^2 + O\big((w-a^2)^{-3}\big)
\end{equation}
The coefficients $c_{a,\eps,\pm}$ and $d_{a,\eps,\pm}$ are stated explicitly in Appendix \ref{appendix:coefficients}. We also have the following expansions of the function $g_{a,\alpha,\beta}(w)$,
\begin{equation}
    g_{a,\alpha,\beta}(w) = \frac{1}{2} + b_{a,+}(w-a^2) + O\big((w-a^2)^2\big)
\end{equation}
\begin{equation}
    g_{a,\alpha,\beta}(w) = \frac{1}{2} + b_{a,-}(w-a^{-2}) + O\big((w-a^{-2})^2\big)
\end{equation}
where coefficients $b_{a,-}$ and $b_{a,+}$ are also given in Appendix \ref{appendix:coefficients}. This gives us the information necessary to write the leading order terms of the Laurent expansion of $\text{tr } \phi_{a,N}(w)$ about $w = a^2$ and $w = a^{-2}$
\begin{equation}
    \text{tr } \phi_{a,N}(w) = \frac{(c_{a,\alpha,+}c_{a,\beta,+})^\nhalf}{2(w-a^2)^{2N}} + O\big((w-a^2)^{-2N+1}\big)
\end{equation}
\begin{equation}
    \text{tr } \phi_{a,N}(w) = \frac{1}{2}(c_{a,\alpha,-}c_{a,\beta,-})^\nhalf + O\big(w-a^{-2}\big)
\end{equation}
Using the above expansions and the form of the determinant, this is enough to see that $r_{a,1,N}(w)$ has a pole of order $2N$ at $w = a^2$ and is analytic and non-zero at $w = a^{-2}$. By equation \eqref{eq:detrelation}, we can conclude that $r_{a,2,N}(w)$ is analytic and non-zero at $w = a^2$ and has a zero of order $2N$ at $w = a^{-2}$.
\end{proof}
%%%%%%%%%%%%%%%%%%%%%%%%%%%%%%%%%%%%%%%%%

%%%%%%%%%%%%%%%%%%%%%%%%%%%%%%%%%%%%%%%%%%%%%
\subsection{Analysis of the $F_N$-matrices} 
For the proofs of Theorem \ref{theorem:kernel} and Lemma \ref{lemma:BDstuff} we need to have an understanding of the poles of the matrix valued functions $F_{a,N,k}(z)$. In particular, our goal is to prove Lemma \ref{lem:FNanalyticGEN}. In turn we will also prove Lemma \ref{lem:FNanalytic}, as it is a direct consequence Lemma \ref{lem:FNanalyticGEN}. We will avoid explicitly stating the $F_{a,N}$ matrices as they are quite involved and their form is not insightful. 
\begin{proof}
    We write the eigen-decomposition of $\phi_{a,N}(z)$ as, 
    $$\phi_{a,N}(z) = E_{a,N}(z)D_{a,N}(z)E_{a,N}(z)^{-1}$$
    Where 
    $$D_{a,N}(z) = \begin{pmatrix}
        r_{a,1,N}(z) & 0 \\
        0 & r_{a,2,N}(z)
    \end{pmatrix}$$
    And the $F_{a,N}$ matrices are defined by, 
    $$F_{a,1,N}(z) = E_{a,N}(z)\begin{pmatrix}
        1 & 0\\ 0 & 0
    \end{pmatrix}E_{a,N}(z)^{-1}$$
    $$F_{a,2,N}(z) = E_{a,N}(z)\begin{pmatrix}
        0 & 0\\ 0 & 1
    \end{pmatrix}E_{a,N}(z)^{-1}$$
    First we should note the following equality, 
    \begin{equation}
        F_{a,1,N}(z) + F_{a,2,N}(z) = I
    \end{equation}
    Where $I$ is the identity matrix. This means for any $z \in \mathbb{C}$, $F_{a,1,N}(z)$ and $F_{a,2,N}(z)$ must have poles of the same order. We recall some facts about the poles of relevant functions, 
    \begingroup
    \begin{center}
    \renewcommand{\arraystretch}{1.5}
    \begin{tabular}{c|c|c}
        & $z = a^2$ & $z=a^{-2}$ \\
        \hline
        $r_{a,1,N}(z)$ & pole of order $2N$ & finite (non-zero) \\
        \hline
        $r_{a,2,N}(z)$ & finite (non-zero) & zero of order $2N$ \\
        \hline 
        $\phi_{a,N}(z)$ & pole of order $2N$ & finite (non-zero) \\
        \hline 
        $\phi_{a,N}(z)^{-1}$ & finite (non-zero) & pole of order $2N$\\
    \end{tabular}
    \end{center}
    \endgroup
    \noindent The statements about $r_{a,1,N}(z)$ and $r_{a,2,N}(z)$ are proven in Lemma \ref{lem:1}, while the facts about $\phi_{a,N}(z)^{\pm1}$ can be directly computed via the definition in equations \eqref{eq:phiepsa} and \eqref{eq:phian}. Now we write, 
    \begin{equation} \label{eq:phia1}
        \phi_{a,N}(z) = r_{a,1,N}(z)F_{a,1,N}(z) + r_{a,2,N}(z)F_{a,2,N}(z)
    \end{equation}
    The order of the pole at $z = a^2$ must agree on both sides of the equation. Since $F_{a,1,N}(z)$ and $F_{a,2,N}(z)$ must have poles of the same order, the only way the equation is consistent is if $F_{a,1,N}(z)$ and $F_{a,2,N}(z)$ have no pole at $z = a^2$. Similarly, by analyzing the equality 
    \begin{equation} \label{eq:phia2}
        \phi_{a,N}(z)^{-1} = r_{a,1,N}(z)^{-1}F_{a,1,N}(z) + r_{a,2,N}(z)^{-1}F_{a,2,N}(z)
    \end{equation}
    We can deduce that $F_{a,1,N}(z)$ and $F_{a,2,N}(z)$ have no pole at $z = a^{-2}$. 
\end{proof}

%%%%%%%%%%%%%%%%%%%%%%%%%%%%%%%%%%%%%%%%%%%%%%%%%%%%%%%%%%%%%%%%%%%%%%%%%%%%%%%%%%%%%%%%%%%%%%%%%%%%%%%%%%%%%%%%%%%%%%%%%%%%%%%%%%%%%%%%%%%%%%%%%%%%%%%%%%%%%
\appendix

\section{Explicit Statement of Coefficients} \label{appendix:coefficients}
Below we have recorded some coefficients named, but not explicitly stated, in the proof of Lemma \ref{lem:1}.
\begin{align*}
    c_{a,\eps,+} = \frac{\left(\eps^2+1\right)^2 \left(a^2+1\right)^2}{\eps^2} && c_{a,\eps,-} = \frac{\left(\eps^2+1\right)^2}{\eps^2 \left(a^2-1\right)^2}
\end{align*}
\begin{align*}
    d_{a,\eps,+} = \frac{\eps^2 \left(a^2-1\right)^2}{\left(\eps^2+1\right)^2 a^4} && d_{a,\eps,-} = \frac{\eps^2 a^4}{\left(\eps^2+1\right)^2 \left(a^2+1\right)^2}
\end{align*}
\begin{align*}
    b_{a,+} = \frac{\left(a^4-1\right) \left(\alpha^2-\beta^2\right)^2}{\left(\alpha^2+1\right)^2 \left(\beta^2+1\right)^2 \left(a^3+a\right)^2} && b_{a,-} = -\frac{a^2 \left(a^2-1\right) \left(\alpha^2-\beta^2\right)^2}{\left(\alpha^2+1\right)^2 \left(\beta^2+1\right)^2 \left(a^2+1\right)}
\end{align*}

\vspace{1cm}

\noindent \textbf{Data Availability.} No datasets were generated or analyzed for this research.\\

\noindent \textbf{Conflict of Interest.} There are no relevant financial or non-financial interest to disclose.

\bibliographystyle{abbrv}
\bibliography{bibliography}

\end{document}